\documentclass{CSML}
\pdfoutput=1

\usepackage{lastpage}
\newcommand{\dOi}{13(3:29)2017}
\lmcsheading{}{1--\pageref{LastPage}}{}{}%
{Feb.~10,~2017}{Sep.~19, 2017}{}

\usepackage{hyperref}
\hypersetup{hidelinks}

\usepackage{amsmath}
\usepackage{amsfonts}
\usepackage{amssymb}
\usepackage{amsthm}
\usepackage{xstring}
\usepackage{url}
\usepackage{tabularx}
\usepackage{xspace}

\usepackage{tikz}
\usetikzlibrary{decorations.pathreplacing}
\usetikzlibrary{decorations.pathmorphing}
\usetikzlibrary{shapes,shapes.symbols,automata,arrows}
\usetikzlibrary{calc}
\usetikzlibrary{patterns}
\usetikzlibrary{positioning}

\tikzset{p0/.style = {shape = circle, draw, thick, minimum size = 0.7cm}}
\tikzset{p1/.style = {rectangle, minimum size=.7cm, draw, thick}}
\tikzset{>=stealth, shorten >=1pt}
\tikzset{every edge/.style = {thick, ->, draw}}
\tikzset{every loop/.style = {thick, ->, draw}}

\makeatletter
\tikzset{circle split part fill/.style  args={#1,#2}{%
 alias=tmp@name, 
  postaction={%
    insert path={
     \pgfextra{%
     \pgfpointdiff{\pgfpointanchor{\pgf@node@name}{center}}%
                  {\pgfpointanchor{\pgf@node@name}{east}}%
     \pgfmathsetmacro\insiderad{\pgf@x}
      \fill[#1] (\pgf@node@name.base) ([xshift=-\pgflinewidth]\pgf@node@name.east) arc
                          (0:180:\insiderad-\pgflinewidth)--cycle;
      \fill[#2] (\pgf@node@name.base) ([xshift=\pgflinewidth]\pgf@node@name.west)  arc
                           (180:360:\insiderad-\pgflinewidth)--cycle; 
         }}}}}  
 \makeatother

\usepackage{booktabs}

\usepackage{pifont}

\definecolor{myred}{rgb}{1,0.604,0.604}
\definecolor{mydarkred}{rgb}{1,0.345,0.345}
\definecolor{myblue}{rgb}{0.635,0.675,0.966}
\definecolor{mydarkblue}{rgb}{0.412,0.475,0.957}
\definecolor{myyellow}{rgb}{1,0.976,0.604}
\definecolor{mydarkyellow}{rgb}{1,0.961,0.345}

\tikzset{
	assign/.style = { fill=myblue },
	choice/.style = { fill=myred },
	check/.style  = { fill=myyellow }
}

\makeatletter
\DeclareRobustCommand{\rvdots}{%
  \vbox{
    \baselineskip4\p@\lineskiplimit\z@
    \kern-\p@
    \hbox{.}\hbox{.}\hbox{.}
  }}
\makeatother

\newcommand{\ext}{\mathrm{ext}}
\newcommand{\req}{\mathit{req}}

\newcommand{\initmark}{I}
\newcommand{\finitemark}{f}
\newcommand{\vinit}{v_{\initmark}}

\newcommand{\qbf}{\mathrm{QBF}}

\newcommand{\myquot}[1]{``#1''}

\newcommand{\bigo}[0]{\mathcal{O}}

\newcommand{\size}[1]{|#1|}
\newcommand{\card}[1]{\size{#1}}
\newcommand{\set}[1]{\{ #1 \}}
\newcommand{\nats}{\mathbb{N}}

\renewcommand{\epsilon}{\varepsilon}

\DeclareMathOperator*{\argmax}{arg\,max}

\newcommand{\eps}{\ensuremath{\boldsymbol{\epsilon}}}
\newcommand{\inc}{\ensuremath{\mathbf{i}}}

\newcommand{\ttrue}{\mathit{true}}
\newcommand{\ffalse}{\mathit{false}}

\newcommand{\arena}{\mathcal{A}}
\newcommand{\game}{\mathcal{G}}
\newcommand{\cost}{\mathrm{Cst}}

\newcommand{\col}{\Omega}
\newcommand{\wincond}{\mathrm{Win}}

\newcommand{\parity}{\mathrm{Parity}}
\newcommand{\finparity}{\mathrm{FinParity}}
\newcommand{\cp}{\mathrm{CostParity}}
\newcommand{\bincp}{\cp}

\newcommand{\streett}{\mathrm{Streett}}

\newcommand{\finstreett}{\mathrm{FinStreett}}
\newcommand{\streettc}{\mathrm{CostStreett}}

\newcommand{\mem}{\mathcal{M}}
\newcommand{\init}{m_\initmark}
\newcommand{\update}{\mathrm{Upd}}
\newcommand{\nxt}{\mathrm{Nxt}}

\newcommand{\paritydist}{\mathrm{Cor}}
\newcommand{\streettdist}{\mathrm{StCor}}

\newcommand{\answer}[1]{\mathrm{Ans}({#1})}

\newcommand{\incseq}{\mathit{IncSeq}}

\newcommand{\up}{\textsc{UP}}
\newcommand{\coup}{\textsc{co-UP}}
\newcommand{\np}{\textsc{NP}}
\newcommand{\conp}{\textsc{co-NP}}
\newcommand{\exptime}{\textsc{ExpTime}}
\newcommand{\aptime}{\textsc{APTime}}
\newcommand{\ptime}{\textsc{PTime}}
\newcommand{\twoexp}{\textsc{2ExpTime}}
\newcommand{\pspace}{\textsc{PSpace}\xspace}
\newcommand{\threeexp}{\textsc{3ExpTime}}

\newcommand{\relreq}{\textsc{RelReq}}
\newcommand{\dominates}{\sqsupseteq}
\newcommand{\dominatedby}{\sqsubseteq}
\newcommand{\dominationequivalent}{\approx}

\newcommand{\jump}{\curvearrowright}
\newcommand{\extgame}{\ensuremath{\game'}}
\newcommand{\jumpgame}{\ensuremath{\game'_\jump}}
\newcommand{\finitegame}{\ensuremath{\game'_f}}

\newcommand{\reach}{\mathcal{R}}

\title[Optimal Strategies in Parity Games with Costs]{Easy to Win, Hard to Master:\\ Optimal Strategies in Parity Games with Costs}
\titlecomment{This is a revised and extended version of work first presented at CSL '16 \cite{WeinertZimmermann16}}

\author{Alexander Weinert}

\author{Martin Zimmermann}
\address{Reactive Systems Group, Saarland University, 66123 Saarbr\"ucken, Germany}
\email{\{weinert,zimmermann\}@react.uni-saarland.de}

\thanks{Supported by the project ``TriCS'' (ZI 1516/1-1) of the German Research Foundation (DFG)}


\begin{document}

\maketitle

\begin{abstract} The winning condition of a parity game with costs requires an arbitrary, but fixed bound on the cost incurred between occurrences of odd colors and the next occurrence of a larger even one.
Such games quantitatively extend parity games while retaining most of their attractive properties, i.e, determining the winner is in NP and co-NP and one player has positional winning strategies. 

We show that the characteristics of parity games with costs are vastly different when asking for strategies realizing the minimal such bound: The solution problem becomes PSPACE-complete and exponential memory is both necessary in general and always sufficient. 
Thus, solving and playing parity games with costs optimally is harder than just winning them.
Moreover, we show that the tradeoff between the memory size and the realized bound is gradual in general.
All these results hold true for both a unary and binary encoding of costs.

Moreover, we investigate Streett games with costs.
Here, playing optimally is as hard as winning, both in terms of complexity and memory.
 \end{abstract}

\section{Introduction}

Recently, the focus of research into infinite games for the synthesis of reactive systems moved from studying qualitative winning conditions to quantitative ones. This paradigm shift entails novel research questions, as quantitative conditions induce a (partial) ordering of winning strategies. In particular, there is a notion of semantic optimality for strategies which does not appear in the qualitative setting. Thus, in the quantitative setting, one can ask whether computing optimal strategies is harder than computing arbitrary ones, whether optimal strategies are necessarily larger than arbitrary ones, and whether there are tradeoffs between different quality measures for strategies, e.g., between the size of the strategy and its semantic quality (in terms of satisfaction of the winning condition).

As an introductory example consider the classical (max)-parity condition, which is defined for an infinite sequence drawn from a finite subset of the natural numbers, so-called colors. The parity condition is satisfied if almost all occurrences of an odd color are \emph{answered} by a later occurrence of a larger even color, e.g., the sequence 
\[ \pi = 1\, 0 \, 2 \,\, 1\, 0\, 0 \, 2 \,\, 1\, 0\, 0\, 0 \, 2 \,\, 1\, 0\, 0\, 0\, 0 \, 2 \,\, 1\, 0\, 0\, 0\, 0\, 0 \, 2 \,\, 1\, 0\, 0\, 0\, 0\, 0\, 0 \, 2 \,\, 1\, 0\, 0\, 0\, 0\, 0\, 0\, 0 \, 2 \, \cdots \]
satisfies the parity condition, as every $1$ is eventually answered by a $2$.

The finitary parity condition~\cite{ChatterjeeHenzingerHorn09} is obtained by additionally requiring the existence of a bound~$b$ such that almost every odd color is answered within at most $b$ steps, i.e., $\pi$ does not satisfy the finitary parity condition, as the length of the zero-blocks is unbounded. Thus, solving a finitary parity game is a boundedness problem: In order to satisfy the condition, an arbitrary, but fixed bound has to be met. 
In particular, winning strategies for finitary parity  games are naturally ordered by the minimal bound they realize along all consistent plays. Thus, finitary parity games induce an optimization problem: Compute an optimal winning strategy, i.e., one that  guarantees the smallest possible bound.

Other examples for such quantitative winning conditions include mean payoff~\cite{EhrenfeuchtMycielski79, ZwickPaterson95} and energy~\cite{BouyerFLMS08,ChakrabartiAHS03} conditions and their combinations and extensions, request-response conditions~\cite{HornThomasWallmeierZimmermann15, Zimmermann09}, parity with costs~\cite{FZ14}, and parameterized extensions of Linear Temporal Logic (LTL)~\cite{AlurEtessamiLaTorrePeled01,FaymonvilleZimmermann14,KupfermanPitermanVardi09,Zimmermann13,Zimmermann15}. Often, these conditions are obtained by interpreting a classical qualitative winning condition quantitatively, as exemplified by the finitary parity condition introduced above. 

Often, the best algorithms for solving such boundedness conditions are as fast as the best ones for their respective qualitative variant, while the fastest algorithms for the optimization problem are worse. For example, solving games with winning conditions in  Prompt-LTL, a quantitative variant of LTL, is $\twoexp$-complete~\cite{KupfermanPitermanVardi09} (i.e., as hard as solving classical LTL games~\cite{PnueliRosner89a}), while computing optimal strategies is only known to be in $\threeexp$~\cite{Zimmermann13}. The same is true for the bounds on the size of winning strategies, which jump from tight doubly-exponential bounds to triply-exponential upper bounds. The situation is similar for other winning conditions as well, e.g., request-response conditions~\cite{HornThomasWallmeierZimmermann15}. These examples all have in common that there are no known lower bounds on the complexity and the memory requirements in the optimization variant, except for the trivial ones for the qualitative case. A notable exception are finitary parity games, which are solvable in polynomial time~\cite{ChatterjeeHenzingerHorn09} and thus simpler than parity games (according to the state-of-the-art).

In this work, we study optimal strategies in parity games with costs, a generalization of finitary parity games. In this setting, we are able to show that computing optimal strategies is indeed harder than computing arbitrary strategies, and that optimal strategies have exponentially larger memory requirements in general. A parity game with costs is played in a finite directed graph whose vertices are partitioned into those of Player~$0$ and those of Player~$1$. Starting at an initial vertex, the players move a token through the graph: If it is placed at a vertex of Player~$i$, then this player has to move it to some successor. Thus, after $\omega$~rounds, the players have produced an infinite path through the graph, a so-called play. The vertices of the graph are colored by natural numbers and the edges are labeled by (non-negative) costs. These two labelings induce the parity condition with costs: There has to be a bound~$b$ such that almost all odd colors are followed by a larger even color such that the cost incurred between these two positions is at most $b$. Thus, the sequence~$\pi$ from above satisfies the parity condition with costs, if the cost of the zero-blocks is bounded. Note that the finitary parity condition is the special case where every edge has cost one and the parity condition is the special case where every edge has cost zero. 

Thus, to win a parity game with costs, Player~$0$ has to bound the costs between requests and their responses along all plays. If Player~$0$ has any such strategy, then she has a positional strategy~\cite{FZ14}, i.e., a strategy that determines the next move based only on the vertex the token is currently at, oblivious to the history of the play. Let $n$ denote the number of vertices of the graph the game is played in and let $W$ denote its largest edge cost. Then, a positional winning strategy uniformly bounds the costs to some bound~$b \le nW$, which we refer to as the cost of the strategy. Furthermore, Mogavero et al.\ showed that the winner of a parity game with costs can be determined in $\up \cap \coup$~\cite{MogaveroMS15}. All previous work on parity games with costs was concerned with the boundedness variant, i.e., the problems ask to find some bound, but not necessarily the best one. Here, in contrast, we study optimal strategies in parity games with costs.

When considering parity games with costs as a boundedness problem, the actual edge costs can be abstracted away: It is only relevant whether an edge has cost zero or not. Thus, it suffices to consider only costs~$0$ and $1$ (typically denoted as $\eps$ and $\inc$). We call this setting one of \emph{abstract costs}. 

For the optimization variant, abstracting away the actual costs is no longer valid and encoding larger costs  by subdividing them into edges of cost~$\inc$, i.e., of cost~$1$, comes at the price of an exponential blowup in the graph's size. Thus, we also consider the case of costs in~$\nats$, given in binary encoding. Here, the upper bound~$nW$ on the cost of an optimal strategy might be exponential in the size of the game.

Furthermore, we also study Streett conditions, which generalize parity conditions by relaxing the hierarchical structure of the requests and responses: In a parity condition, a large even color answers requests for \emph{all} smaller odd colors. In contrast, in a Streett condition, any two kinds of responses are potentially independent of each other. It is known that solving the boundedness problem for finitary Streett games and for Streett games with costs is $\exptime$-complete~\cite{FZ14}, which has to be compared to the $\conp$-completeness of solving classical Streett games~\cite{Horn05}. Furthermore, finite-state strategies of exponential size suffice for Player~$0$ to implement a winning strategy in a Streett game with costs~\cite{FZ14}. As above, one obtains $mnW$ as upper bound on the cost of an optimal strategy in a Streett game with costs with $n$ vertices and largest cost $W$, where $m$ is the size of a finite-state winning strategy.

\subsection{Our Contribution}

Our first four results are concerned with the special case of abstract costs, i.e., costs~$0$ and $1$ only; the remaining ones are about games with costs in ~$\nats$.

The first result shows that determining whether Player~$0$ has a strategy in a parity game with costs whose cost is smaller than a given bound~$b$ is $\pspace$-complete. Thus, computing the bound of an optimal strategy is strictly harder than just deciding whether or not some bound exists (unless $\pspace \subseteq \up \cap \coup$). The hardness result is shown by a reduction from $\qbf$ and uses the bound~$b$ to require Player~$0$'s strategy to implement a satisfying Skolem function for the formula, where picking truth values is encoded by requests of odd colors. The lower bound is complemented by a polynomial space algorithm that is obtained from an alternating polynomial time Turing machine that simulates a finite-duration variant of parity games with costs that is won by Player~$0$ if and only if she can enforce a cost of at most $b$ in the original game. To obtain the necessary polynomial upper bound on the play length we rely on the upper bound $n$ on the optimal bound and on a first-cycle variant of parity games (cf.\ \cite{AminofRubin14}) tailored to our setting.

Our second result concerns memory requirements of optimal strategies. A corollary of the correctness of the finite-duration game yields exponential upper bounds: If Player~$0$ has a strategy of cost~$b$, then she also has one of cost~$b$ and of size~$(b+2)^d = 2^{d \log (b+2)}$, where $d$ is the number of odd colors in the game. A similar result holds true for Player~$1$ as well: If he can exceed a given bound~$b$, then he can also do so with a strategy of size~$n(b+2)^d$.

Furthermore, we show that the exponential upper bounds are asymptotically tight: We present a family~$\game_d$ of parity games with costs such that $\game_d$ has $d$ odd colors and Player~$0$ requires strategies of size~$2^{d-1}$ to play optimally in each $\game_d$. This result is  based on using the bound~$b$ to require Player~$0$ to store which odd colors have an open request and in which order they were posed. Our result improves a linear bound presented by Chatterjee and Fijalkow~\cite{CF13arxiv}. Dually, we present an exponential lower bound on the memory size necessary for Player~$1$ to exceed a given bound~$b$. 

Moreover, we study the tradeoff between memory size and cost of a strategy witnessed by the results above: Arbitrary winning strategies are as small as possible, i.e., positional, but in general have cost~$n$. In contrast, optimal strategies realize a smaller bound, but might have exponential size. Hence, one can trade cost for memory and vice versa.
We show that this tradeoff is gradual in the games~$\game_d$: There are strategies~$\sigma_1, \sigma_2, \ldots, \sigma_d$ such that $1 = \size{\sigma_1} < \size{\sigma_2} < \cdots < \size{\sigma_d} = 2^{d-1}$ and $b_1 > b_2 > \cdots > b_d$, where $b_j$ is the cost of $\sigma_j$. Furthermore, we show that the strategy~$\sigma_j$ has minimal size among all strategies of cost~$b_j$. Equivalently, the strategy~$\sigma_j$ has minimal cost among all strategies whose size is not larger than $\sigma_j$'s size. 

Both lower bounds we prove and the tradeoff result already hold for the special case of finitary parity games, which can even be solved in polynomial time~\cite{ChatterjeeHenzingerHorn09}. Hence, in this case, the gap between just winning and playing optimally is even larger.

After the results for the special case of abstract costs (i.e., $0$ and $1$ only), we consider the general case of arbitrary non-negative costs given in binary encoding. We show that determining whether Player~$0$ has a strategy in a parity game with costs whose cost is smaller than a given bound~$b$ is still $\pspace$-complete, i.e., having \emph{larger} costs does not influence the complexity of the problem. The lower bound on the complexity carries over from the special case of abstract costs but the proof of the upper bound is affected by this generalization: The upper bound on the cost of an optimal strategy is now exponential, which implies that the finite-duration variant has exponentially long plays as well. We devise a shortcut criterion to skip parts of a play and prove that this yields the desired alternating polynomial-time algorithm, which places to problem in $\pspace$. 

As before, this reduction also yields exponential upper bounds on the necessary memory to implement a winning strategy: The memory requirements do not increase asymptotically when considering arbitrary costs, and they still match the lower bounds. 

Finally, we consider quantitative Streett conditions. We show that, given Streett game with costs and a bound~$b$, determining whether Player~$0$ has a strategy with cost at most~$b$ is $\exptime$-complete. Thus, playing quantitative Streett games optimally is not harder than just winning them. This is due to the fact that just winning them is already very hard. Furthermore, we present tight exponential bounds on the memory necessary to implement a winning strategy for Player~$0$ in such a game. All lower bounds already hold for the special case of finitary Streett games while the upper bounds hold for arbitrary costs encoded in binary.

\subsection{Related Work}

Tradeoffs in infinite games have been studied before, e.g., in stochastic and timed games, one can trade memory for randomness, i.e., randomized strategies are smaller than deterministic ones~\cite{ChatterjeeAH04,ChatterjeeHP08}. A detailed overview of more recent results in this direction and of tradeoffs in multi-dimensional winning conditions is given in the thesis of Randour~\cite{RandourThesis}. The nature of these results is quite different from ours.

Lang investigated optimal strategies in the resource reachability problem on pushdown graphs~\cite{Lang14}, where there exists a finite number of counters, which may be increased and reset, but not read during a play.
He shows that in order to keep the values of the counters minimal during the play, exponential memory in the number of counters is both necessary and sufficient for Player~$0$.
While the author shows the corresponding decision problem to be decidable, he does not provide a complexity analysis of the problem.
Furthermore, the setting of the problem is quite different from the model considered in this work: He considers infinite graphs and multiple counters, but only reachability conditions, while we consider finite graphs and implicit counters tied to the acceptance condition, which is a general parity~condition.

Also, Fijalkow et al.\ proved the non-existence of a certain tradeoff between size and quality of strategies in boundedness games~\cite{FijalkowHKS15}, which refuted a conjecture with important implications for automata theory and logics. Such games are similar to those considered by Lang in that they are played in potentially infinite arenas and have multiple counters. 

Recently, Bruy{\'e}re et al. introduced window-parity games~\cite{BruyereHautemRandour16}, another quantitative variant of parity games, and proved tight complexity bounds for the scenario with multiple colorings of the arena.
They were also able to show a tight connection between window-parity and finitary parity games.

Finally, the winning conditions considered here have also been studied in the setting of delay games. In such games, Player~$0$ may delay her moves to obtain a lookahead on her opponent's moves, thereby gaining an advantage that allows her to win games she loses without delay. Now, there are potential tradeoffs between quality, size, and amount of delay. Most importantly, one can trade delay for quality and vice versa~\cite{Zimmermann17} which allows Player~$0$ to improve the quality of her strategies by taking advantage of the delay. 

\subsection{Organization of the Paper}

In Section~\ref{sec:defs}, we introduce basic definitions about infinite games; in Section~\ref{sec:costparity}, we introduce parity games with costs. First, we study the variant with abstract costs: we prove the $\pspace$-completeness result (Section~\ref{sec:complexity}), the exponential bounds on the memory requirements of optimal strategies (Section~\ref{sec:memory}), and the gradual tradeoff between cost and size of winning strategies (Section~\ref{sec:tradeoffs}). Then, we turn our attention to the setting of integer-valued costs in Section~\ref{sec:concrete} and to Streett games with costs in Section~\ref{sec:streett}. Finally, we conclude in Section~\ref{sec:conc} by discussing further research.

\section{Preliminaries}\label{sec:defs}

We denote the non-negative integers by $\nats$ and define $[n] = \set{0, 1, \ldots, n-1}$ for every $n \ge 1$.

An \textit{arena}~$\arena=(V, V_0, V_1, E, \vinit)$ consists of a finite, directed graph~$(V, E)$, a partition~$\{V_0, V_1\}$ of $V$ into the positions of Player~$0$ (drawn as circles) and Player~$1$ (drawn as rectangles), and an initial vertex~$\vinit \in V$. The size of $\arena$, denoted by $\size{\arena}$, is defined as $\size{V}$.

A \textit{play} in $\arena$ is an infinite path~$\rho = v_0 v_1 v_2 \cdots$ through $(V, E)$ starting in $\vinit$.
To rule out finite plays, we require every vertex to be non-terminal.
A \textit{game}~$\game = ( \arena, \wincond )$ consists of an arena $\arena$ with vertex set~$V$ and a set~$\wincond \subseteq V^\omega$ of winning plays for Player~$0$.
The set of winning plays for Player~$1$ is $V^\omega \setminus \wincond$.

A \textit{strategy} for Player~$i$ is a mapping $\sigma \colon V^*V_i \rightarrow V$ where $(v, \sigma(wv)) \in E$ for all $wv \in V^* V_i$.
We say that $\sigma$ is \textit{positional} if $\sigma(wv) = \sigma(v)$ for every $wv \in V^*V_i$.
We often view positional strategies as a mapping~$\sigma \colon V_i \rightarrow V$.
A play $v_0 v_1 v_2 \cdots$ is \textit{consistent} with a strategy~$\sigma$ for Player~$i$, if $v_{j+1} = \sigma( v_0 \cdots v_j)$ for every~$j$ with $v_j \in V_i$. A strategy~$\sigma$ for Player~$i$ is a \textit{winning strategy} for $\game$ if every play that is consistent with $\sigma$ is won by Player~$i$.
If Player~$i$ has a winning strategy, then we say she wins $\game$.
\textit{Solving} a game amounts to determining its winner.

A \textit{memory structure}~$\mem = (M, \init, \update)$ for an arena $(V, V_0, V_1, E, \vinit)$ consists of a finite set~$M$ of memory states, an initial memory state $\init \in M$, and an update function~$\update\colon M \times E \rightarrow M$.
The update function can be extended to finite play prefixes in the usual way: $\update^+(m, v) = m$ and $\update^+(m, w v v') = \update(\update^+(m, w v), (v,v'))$ for $w \in V^*$ and $(v,v') \in E$.
A next-move function $\nxt \colon V_i \times M \rightarrow V$ for Player~$i$ has to satisfy $(v, \nxt(v, m)) \in E$ for all $v \in V_i$ and all $m \in M$.
It induces a strategy~$\sigma$ for Player~$i$ with memory~$\mem$ via $\sigma(v_0\cdots v_j) = \nxt(v_j, \update^+(\init, v_0 \cdots v_j))$.
A strategy is called \textit{finite-state} if it can be implemented by a memory structure.
We define $\card{\mem} = \card{M}$.
The size of a finite-state strategy is the size of a smallest memory structure implementing it.

An arena $\arena = (V, V_0, V_1, E, \vinit)$ and a memory structure $\mem = (M, \init, \update)$ for $\arena$ induce the expanded arena $\arena\times\mem = (V \times M, V_0 \times M, V_1 \times M, E', (\vinit, \init))$ where~$E'$ is defined via $((v,m), (v',m')) \in E'$ if and only if $(v,v') \in E$ and $\update(m, (v,v') ) = m'$.
Every play $\rho = v_0 v_1 v_2\cdots$ in $\arena$ has a unique extended play $\ext(\rho) = (v_0, m_0) (v_1, m_1)
(v_2, m_2) \cdots$ in $\arena \times \mem$ defined by $m_0 = \init$ and $m_{j+1} = \update(m_j, (v_j, v_{j+1}))$, i.e., $m_j = \update^+(\init, v_0 \cdots v_j)$. The extended play of a finite play prefix in $\arena$ is defined similarly.

\section{Parity Games with Costs}\label{sec:costparity}

In this section, we introduce the parity condition with costs~\cite{FZ14}. Fix an arena~$\arena = (V, V_0, V_1, E, \vinit)$. A cost function for $\arena$ is an edge-labeling~$\cost \colon E \rightarrow \set{\eps,\inc}$.\footnote{Note that using the abstract costs~$\eps$ and $\inc$ essentially entails a unary encoding of costs. We discuss the case of a binary encoding of arbitrary costs in Section~\ref{sec:concrete}.}
Edges labeled with $\inc$ are called increment-edges while edges labeled with
$\eps$ are called $\eps$-edges. We extend the
edge-labeling to a cost function over plays obtained by
counting the number of increment-edges traversed during the play, i.e., $\cost(\rho) \in \nats \cup \set{\infty}$ for any play~$\rho$.
The cost of a finite play infix is defined analogously. Also, fix a coloring~$\col \colon V \rightarrow
\nats$ of $\arena$'s vertices.
The classical parity condition requires almost all occurrences of odd colors to be answered by a later occurrence of a larger even color.
Hence, let $\answer{c} = \set{c' \in \nats \mid c' \ge c \text{ and $c'$ is even}}$ be the
set of colors that answer a request of color~$c$.

Let $\rho = v_0 v_1 v_2 \cdots$ be a play.
We define the cost-of-response at position~$j \in \nats$ of $\rho$ by 
\[
\paritydist(\rho, j) = \min \set{ \cost (v_j \cdots v_{j'}) \mid  j' \ge j \text{ and } \col(v_{j'}) \in \answer{\col(v_j)} }\enspace, 
\]
where we use $\min \emptyset = \infty$, i.e., $\paritydist(\rho, j)$ is the cost of the infix of $\rho$ from position~$j$ to its first answer, and $\infty$ if there is no answer.

We say that a request at position~$j$ is answered with cost~$b$, if $\paritydist(\rho, j) = b$.
Consequently, a request with an even color is answered with cost zero.
The cost-of-response of an unanswered request is infinite, even if it only incurs finite cost during the remainder of the play, i.e., if there are only finitely many increment-edges succeeding the request.

The parity condition with costs is defined as
\[\cp(\col, \cost) = \set{ \rho \in V^\omega \mid \limsup\nolimits_{j\rightarrow \infty} \paritydist(\rho, j) < \infty } \enspace,\]
i.e., $\rho$ satisfies the condition, if there exists a bound~$b \in \nats$ such that all but finitely many requests are answered with cost less than $b$.
In particular, only finitely many requests may be unanswered, even with finite cost.
Note that the bound~$b$ may depend on the play $\rho$.

A game~$\game = (\arena, \cp(\col, \cost))$ is called a parity game with costs and its size is defined to be $\size{\arena}$.
If $\cost$ assigns $\eps$ to every edge, then $\cp(\col, \cost)$ is a classical (max-) parity condition, denoted by $\parity(\col)$.
Dually, if $\cost$ assigns $\inc$ to every edge, 
then $\cp(\col, \cost)$ is equal to the
finitary parity condition over $\col$, as introduced by Chatterjee et al.~\cite{ChatterjeeHenzingerHorn09} and denoted by $\finparity(\col)$. In these cases, we refer to $\game$ as a parity or a finitary parity game, respectively.

As most of our examples are finitary parity games, we omit the edge-labeling when drawing them for the sake of readability.
For the same reason, we sometimes use non-negative integer costs on the edges of finitary parity games.
Such games can be transformed into finitary parity games as defined above by subdividing these edges and coloring the newly added vertices with color~$0$.
In all cases in this work, this only incurs a polynomial blowup.

\begin{figure}
  \centering
  \hfill
  \begin{tikzpicture}[thick]
	\node[p0,assign] (0-0) {$1$};
	\node[p1,right=of 0-0,assign] (0-1) {$0$};
	\node[p0,right=of 0-1,assign] (0-2) {$2$};

	\path
	  ($(0-0) - (1cm,0)$) edge (0-0)
	  (0-0) edge node[anchor=south] {$\inc$} (0-1)
	  (0-1) edge[loop above] node[anchor=south] {$\inc$} (0-1)
	  (0-1) edge node[anchor=south] {$\inc$} (0-2)
	  (0-2) edge[bend left] node[anchor=north] {$\inc$} (0-0);
  \end{tikzpicture}
  \hfill
  \begin{tikzpicture}[thick]
	\node[p0,assign] (0-0) {$1$};
	\node[p1,right=of 0-0,assign] (0-1) {$0$};
	\node[p0,right=of 0-1,assign] (0-2) {$2$};

	\path
	  ($(0-0) - (1cm,0)$) edge (0-0)
	  (0-0) edge node[anchor=south] {$\inc$} (0-1)
	  (0-1) edge[loop above] node[anchor=south] {$\eps$} (0-1)
	  (0-1) edge node[anchor=south] {$\inc$} (0-2)
	  (0-2) edge[bend left] node[anchor=north] {$\inc$} (0-0);
  \end{tikzpicture}
  \hspace*{\fill}
  \caption{Two parity games with costs. Player~$1$ only wins the left game.}
  \label{fig:parity-cost-example}
\end{figure}
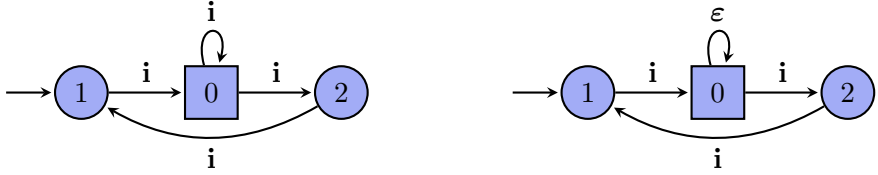

Player~$1$ has two ways of winning a parity game with costs: Either he violates the classical parity condition, or he delays answers to requests arbitrarily.
Consider the two parity games with costs shown in Figure~\ref{fig:parity-cost-example}.
In the game on the left-hand side, Player~$1$ has a winning strategy, by taking the self-loop of the middle vertex~$j$ times upon the~$j$-th visit to it via the increment-edge from the leftmost vertex.
Thus, he delays answers to the request for~$1$ arbitrarily and wins by the second condition.
In the game on the right-hand side, however, Player~$1$ does not have a winning strategy.
If he eventually remains in the vertex labeled with $0$, then there are only finitely many requests, only one of which is unanswered.
Thus, the cost of the play is~$0$, i.e., it is won by Player~$0$.
If he, on the other hand, always leaves the middle vertex eventually, then each request is answered with cost~$2$, hence Player~$0$ wins as well.

\begin{thm}
\label{thm:previouswork}\leavevmode
	\begin{enumerate}
		\item\label{thm:previouswork:parity}
		 Parity games can be solved in quasi-polynomial time and the problem is in $\up \cap \coup$. The winner has a positional winning strategy~\cite{CaludeJainKhoussainovLiStephan17,Jurdzinski98,EmersonJutla91,Mostowski91}.
		\item\label{thm:previouswork:finitary}
		Solving finitary parity games is in $\ptime$. If Player~$0$ wins, then she has a positional winning strategy, but Player~$1$ has in general no finite-state winning strategy~\cite{ChatterjeeHenzingerHorn09}.
		\item\label{thm:previouswork:cost}
		Parity games with costs can be solved in quasi-polynomial time and the problem is in $\up \cap \coup$.  If Player~$0$ wins, then she has a positional winning strategy, but Player~$1$ has in general no finite-state winning strategy~\cite{MogaveroMS15,CaludeJainKhoussainovLiStephan17,FZ14}.
	\end{enumerate}
\end{thm}
A winning strategy for Player~$0$ in a parity game with costs does not have to realize a uniform bound~$b$ on the value~$\limsup_{j\rightarrow \infty} \paritydist(\rho , j)$ among all plays~$\rho$ that are consistent with $\sigma$, but the bound may depend on the play.
To capture the cost of a strategy, we first define the cost of a play~$\rho$ as $\cost(\rho) = \limsup_{j\rightarrow \infty} \paritydist(\rho , j)$ and the cost of a strategy~$\sigma$ as $\cost(\sigma) = \sup_{\rho} \cost(\rho)$, where the supremum ranges over all plays~$\rho$ that are consistent with~$\sigma$. A strategy is optimal for $\game$ if it has minimal cost among all strategies for $\game$.
Analogously, for a strategy~$\tau$ for Player~$1$, we define $\cost(\tau) = \inf_\rho \cost(\rho)$, where $\rho$ again ranges over all plays consistent with~$\tau$.

A corollary of Theorem~\ref{thm:previouswork}(\ref{thm:previouswork:cost}) yields an upper bound on the cost of an optimal strategy: A straightforward pumping argument shows that a positional winning strategy, which always exists if there exists any winning strategy, realizes a uniform bound~$b \le n$ for every play, where~$n$ is the number of vertices of the game.

\begin{cor}
\label{corollary_costupperbound}
Let $\game$ be a parity game with costs with $n$ vertices. If Player~$0$ wins $\game$, then she has a strategy~$\sigma$ with $\cost(\sigma) \le n$, i.e., an optimal strategy has cost at most~$n$.
\end{cor}

This bound is tight, as it trivial to construct finitary parity games~$\game_n$ with~$n+1$ vertices and a unique play such that Player~$0$ wins~$\game_n$ with respect to bound~$b = n$, but not with respect to any bound~$b' < n$.

\section{The Complexity of Solving Parity Games with Costs Optimally}\label{sec:complexity}

In this section we study the complexity of determining the cost of an optimal strategy for a parity game with costs. Recall that solving such games is in $\up \cap \coup$ (and therefore unlikely to be $\np$-complete or $\conp$-complete) while solving the special case of finitary parity games is in $\ptime$. Our main result of this section shows that checking whether a strategy of cost at most $b$ exists is $\pspace$-complete, where hardness already holds for finitary parity games. Therefore, this decision problem is harder than just solving the game (unless $\pspace \subseteq \up \cap \coup$, respectively $\pspace \subseteq \ptime$). 

\begin{thm}
The following problem is $\pspace$-complete: \myquot{Given a parity game with costs~$\game$ and a bound~$b \in \nats$, does Player~$0$ have a strategy~$\sigma$ for $\game$ with $\cost(\sigma) \le b$?}
\end{thm}

Note that we do not specify how $b$ is encoded: We will argue at the beginning of Section~\ref{sec_pspacemembership} that the problem is trivial for bounds~$b > n$, i.e., the complexity of the problem is independent of the encoding of $b$.

The proof of the theorem is split into two lemmas, Lemma~\ref{lemma_pspacemembership} showing membership and Lemma~\ref{lemma_pspacehardness} showing hardness, which are presented in Section~\ref{sec_pspacemembership} and Section~\ref{sec_pspacehardness}, respectively.

\subsection{Solving Parity Games with Costs Optimally is in Polynomial Space}
\label{sec_pspacemembership}

The remainder of this section is dedicated to showing that parity games with costs with respect to a given bound can be solved in polynomial space.
To this end, we fix a parity game with costs $\game = (\arena, \cp(\col, \cost))$ with $\arena = (V, V_0, V_1, E, \vinit)$ and a bound~$b$.
Let $n = \card{V}$.
First, let us remark that we can assume w.l.o.g.\ $b < n$:
If $b \geq n$, then, due to Corollary~\ref{corollary_costupperbound}, we just have to check whether Player~$0$ wins $\game$. 
This is possible in polynomial space due to Theorem~\ref{thm:previouswork}(\ref{thm:previouswork:cost}).

To obtain a polynomial space algorithm, we first turn  the quantitative game $\game$ into a qualitative parity game $\extgame$ in which the cost of open requests is explicitly tracked up to the bound $b$.
To this end, we use functions~$r$ mapping odd colors to $\set{\bot} \cup [b+1] = \set{\bot} \cup \set{0,\dots,b}$, where $\bot$ denotes that no open request of this color is pending.
Additionally, whenever the bound~$b$ is exceeded for some request, all open requests are reset and a so-called overflow counter is increased, up to value~$n$. This accounts for a bounded number of unanswered requests, which are allowed by the parity condition with costs. 
Intuitively, Player~$1$ wins $\extgame$ if he either exceeds the upper bound $b$ at least $n$ times, or if he enforces an infinite play of finite cost with infinitely many unanswered requests.
If he wins by the former condition, then he can also enforce infinitely many excesses of $b$ via a pumping argument.
The latter condition accounts for plays in which Player~$1$ wins without violating the bound~$b$ repeatedly, but by violating the classical parity condition.
We show that Player~$0$ has a strategy $\sigma$ in $\game$ with $\cost(\sigma) \leq b$ if and only if she wins $\extgame$ from its initial vertex.

The resulting game~$\extgame$ is of exponential size in the number of odd colors~$d$ and can therefore in general not be solved in polynomial space in~$n$.
Thus, in a second step, we construct a finite-duration variant~$\finitegame$ of $\extgame$, which is played on the same arena as $\extgame$, but the winner of a play is determined after a polynomial number of moves.
We show that Player~$0$ wins $\extgame$ if and only if she wins $\finitegame$.
To conclude, we show how to simulate $\finitegame$ on the fly on an alternating Turing machine in polynomial time in~$n$, which yields a polynomial space algorithm by removing the alternation~\cite{ChandraKS81}.

We begin by defining $\extgame$.
Let $R =( \set{\bot} \cup[b+1])^D$ be the set of request functions, where~$D$ is the set of odd colors occurring in $\game$.
Here, $r(c) = \bot$ denotes that there is no open request for the color $c$, while $r(c) \neq \bot$ encodes that the oldest open request of $c$ has incurred cost~$r(c)$. 
Using these functions, we define the memory structure $\mem = ([n+1] \times R, \init, \update)$, where the first component~$[n+1]$ implements the previously mentioned overflow counter.
It suffices to bound this counter by $n$, since, as we will show, if Player~$1$ can enforce~$n$ overflows in $\game$, then he can also enforce infinitely many by a pumping argument.
If this counter reaches~$n$, we say that it is saturated.

The initial memory state $\init$ is the pair $(0, r_{v_\initmark})$, where, for an arbitrary~$v \in V$, $r_v$ is the function mapping all odd colors to $\bot$, if $\col(v)$ is even.
If $\col(v)$ is odd, however, $r_v$ maps~$\col(v)$ to $0$, and all other odd colors to $\bot$.
The update function $\update(m, e)$ is defined such that traversing an edge $e = (v, v')$ updates the memory state $(o,r)$ to the memory state $(o',r')$ by performing the following steps in order:
\begin{itemize}
	\item If $e$ is an increment-edge, then for each $c$ with $r(c) \neq \bot$, set $r'(c) = r(c) + 1$.
		  For all other~$c$, set~$r'(c) = r(c) = \bot$.
		  If~$e$ is an~$\eps$-edge, however, then set~$r'(c) = r(c)$ for all~$c$.
	\item Now, if there exists a color $c$ such that $r'(c) > b$, then set~$r'(c) = \bot$ for all~$c$ and set $o'$ to the minimum of $o+1$ and $n$.
		Otherwise, set~$o'$ to~$o$.
	\item If $\col(v')$ is even, set~$r(c')$ to~$\bot$ for every $c' \leq \col(v')$.
	\item If $\col(v')$ is odd, then set $r'(\col(v'))$ to the maximum of the previous value of $r'(\col(v'))$ and $0$, where $\max\set{\bot, 0} = 0$.
\end{itemize}
The resulting $o'$ is at most $n$ and the resulting function~$r'$ is an element of $R$.
We show an example of the evolution of the memory states on a play prefix in Figure~\ref{fig:pspace-membership:example-requests}.
In particular, the move from the fifth vertex to the sixth one causes an overflow that resets all requests, but also increments the overflow counter.
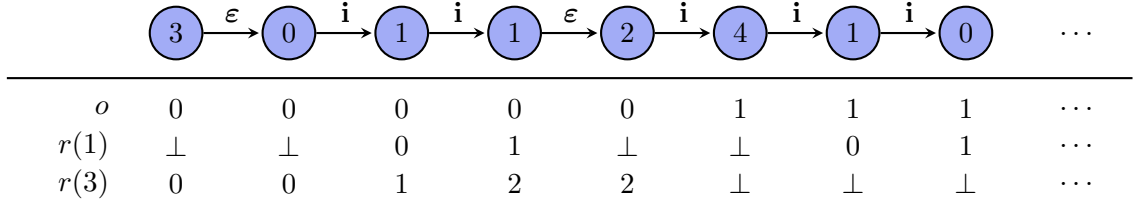
\begin{figure}
\centering
	\begin{tikzpicture}[thick,xscale=1.5]
		\begin{scope}
			\foreach \color [count=\index from 0] in {3,0,1,1,2,4,1,0} {
				\node[p0,assign] (node-\index) at (\index,0) {\color};
			}
			\node (node-8) at (8,0) {$\dots$};
			
			\foreach \label [count=\from from 0,count=\to from 1] in {\eps,\inc,\inc,\eps,\inc,\inc,\inc} {
				\path (node-\from) edge node [anchor=south] {$\label$} (node-\to);
			}
		\end{scope}
		
		\draw[] (-1.5,-.6) -- (8.5,-.6);
		
		\begin{scope}[shift={(0,-1)}]
			\node[anchor=east] at (-.5,0) {$o$};
			\foreach \overflowcounter [count=\xpos from 0] in {0,0,0,0,0,1,1,1,\dots} {
				\node at (\xpos,0) {$\overflowcounter$};
			}
		\end{scope}
		
		\begin{scope}[shift={(0,-1.5)}]
			\node[anchor=east] at (-.5,0) {$r(1)$};
			\foreach \requestone [count=\xpos from 0] in {\bot,\bot,0,1,\bot,\bot,0,1,\dots} {
				\node at (\xpos,0) {$\requestone$};
			}
		\end{scope}
		
		\begin{scope}[shift={(0,-2)}]
			\node[anchor=east] at (-.5,0) {$r(3)$};
			\foreach \requestthree [count=\xpos from 0] in {0,0,1,2,2,\bot,\bot,\bot,\dots} {
				\node at (\xpos,0) {$\requestthree$};
			}
		\end{scope}

	\end{tikzpicture}
	\caption{Example of the evolution of the request-functions during a play for~$b =2$.}
	\label{fig:pspace-membership:example-requests}
\end{figure}

We define the parity game $\extgame = (\arena \times \mem, \parity(\col'))$, with $\col'(v, o, r) = \col(v)$ for $o < n$ and $\col'(v, n, r) = 1$.
Note that every play that encounters a vertex of the form $(v, n, r)$ at some point, for some~$v \in V$ and some~$r \in R$, is winning for Player~$1$, since these vertices form a winning sink component for him.
Intuitively, reaching a vertex of this form means that Player~$1$ is able to often open~$n$ requests that are not answered with cost at most~$b$.
We show that this implies him being able to open infinitely many such requests.
However, there is another way of winning the original parity game with costs~$\game$ for him, i.e., by violating the underlying parity condition of~$\game$.
This is possible even if no request incurs a cost greater than~$b$.
Hence,~$\extgame$ is a parity game.

Let $\arena' = \arena \times \mem = (V', V_0', V_1', E', \vinit')$, in particular $V' = V \times M $ and $V_i' = V_i \times M$.
Even though $\extgame$ has no cost function, we say that an edge $((v,m),(v',m'))$ of $\extgame$ is an increment-edge if $(v,v')$ is an increment-edge in $\game$, otherwise we call it an $\eps$-edge.

It suffices to solve $\extgame$ to determine whether Player~$0$ can bound the cost in $\game$ by $b$.

\begin{lem}
\label{lem:cost-parity-to-parity}
Player~$0$ has a strategy $\sigma$ in $\game$ with $\cost(\sigma) \leq b$ if and only if Player~$0$ wins~$\extgame$.
\end{lem}

\begin{proof}
	We first introduce some notation.
	Let $(v_0, o_0, r_0) (v_1, o_1, r_1) (v_2, o_2, r_2) \cdots$ be a play or a play prefix in $\arena'$.
	An \emph{overflow position} is a $j$ such that either $j = 0$ or $o_j = o_{j-1} +1$.
	Note that we have $r_j = r_{v_j}$ for every overflow position, i.e., the request function is reset at each such position.

	For the direction from right to left, assume that Player~$0$ wins~$\extgame$ and let $\sigma' \colon V_0' \rightarrow V'$ be a positional winning strategy for her from $\vinit'$ in $\extgame$.
	We define the finite-state strategy~$\sigma$ for Player~$0$ in $\game$ using the memory structure~$\mem$ and the next-move function~$\nxt$ with $\nxt(v,m) = v'$, if $\sigma'(v,m) = (v',m')$ for some $m'$.
	Let $\rho = v_0 v_1 v_2 \cdots$ be a play that is consistent with $\sigma$. 
	A straightforward induction shows that the unique extended play $\ext(\rho) = (v_0, o_0, r_0)(v_1, o_1, r_1)(v_2, o_2, r_2)\cdots$ in $\extgame$ is consistent with $\sigma'$ and therefore winning for Player~$0$.
	This in particular implies~$o_j < n$ for all~$j \in \nats$, as the vertices~$(v, n, r)$ form a rejecting sink component.
	
	First, assume that~$\rho$ traverses infinitely many increment-edges. 
	The overflow-counter in $\ext(\rho)$ stabilizes at some value less than~$n$ at some point, i.e., there is a position~$j$ such that $o_{j'} = o_j < n$ for every $j' >j$. 
	We claim $\paritydist(\rho, j') \le b$ for every $j' >j$, which finishes this direction of the proof. 
	Assume towards a contradiction that a request at some position after $j$ of $\rho$ is unanswered for $b+1$ increment-edges.
	During every traversal of one of these increment-edges, its associated counter in $\ext(\rho)$ is increased by one and not reset until $b+1$ increment-edges are traversed, which implies encountering an overflow position. 
	This contradicts the choice of the position~$j$.
	Thus, if~$\rho$ traverses infinitely many increment-edges, then almost every request is answered with cost at most~$b$, i.e., $\cost(\rho) \le b$.
	
	Now, consider the case where $\rho$ contains only finitely many increment-edges. 
	Such a play satisfies the parity condition with costs if and only if it satisfies the parity condition. Thus, it suffices to note that $\rho$ and $\ext(\rho)$ coincide on their color sequences, due to~$o_j < n$ for all~$j \in \nats$, and that $\ext(\rho)$ satisfies the parity condition, as it is winning for Player~$0$.

	For the other direction, we prove the contrapositive.
	Assume that Player~$0$ does not win~$\extgame$.
	Then, due to determinacy of parity games, Player~$1$ wins $\extgame$, say using the positional strategy $\tau' \colon V_1' \rightarrow V'$.
	This strategy is only useful as long as the overflow counter is not yet saturated, as a play is trivially winning for Player~$1$ as soon as the sink component is reached. 
	Thus, whenever the overflow counter is increased, we reset it to the smallest possible value for which $\tau'$ is still able to enforce a winning play for Player~$1$.
	
	We introduce the set~$\reach$ that contains all vertices~$(v, o, r)$ that are visited by some play that is consistent with~$\tau'$.
	Then, given a vertex~$v$, let~$o_v = \min\set{o \mid (v,o,r_v) \in \reach}$ with~$\min\emptyset = n$.
	Note that we have $o_{v_\initmark} = 0$ and that the strategy~$\tau'$ is winning from~$(v,o_v,r_v)$ in $\extgame$ for all~$v \in V$.
	Now we define a new memory structure~$\mem' = (M', m_\initmark', \update')$ with $M' = M = [n+1] \times R$, $m_\initmark' = m_\initmark = (0,r_{v_\initmark})$, and  
	\[
	\update'((o,r),(v,v')) = \begin{cases}
		(o,r') &\text{if } \update((o,r),(v,v')) = (o,r')\\
		(o_{v'},r') &\text{if } \update((o,r),(v,v')) = (o+1,r')\\
	\end{cases}\enspace.
	\]
	Note that we have $r' = r_{v'}$ in the second case.
	Now, let $\tau$ be the finite-state strategy implemented by $\mem'$ and the next-move function~$\nxt$ given by $\nxt(v,m) = v'$, if $\tau'(v,m) = (v',m')$ for some $m' \in M$.

	Let $\rho = v_0 v_1 v_2 \cdots $ be some play in~$\game$ that is consistent with $\tau$ and, moreover, let $\rho' = (v_0, o_0, r_0) (v_1, o_1, r_1) (v_2, o_2, r_2) \cdots$ be the extended play of~$\rho$ with respect to $\mem'$. 
	We say that $j$ is a \emph{reset position} if $j=0$ or if $\update((o_{j-1},r_{j-1}),(v_{j-1},v_{j})) = (o_{j-1}+1,r_{j})$, i.e., the second case in the definition of $\update'$ is applied.
	Note that~$\rho'$ is not necessarily a play in $\extgame$, but every maximal infix of~$\rho'$ between two reset positions is a play infix in that game that is consistent with~$\tau'$.
	Furthermore, at every reset position, instead of incrementing the overflow counter, we set it to the minimal value~$o_v$.
	As a reset position in~$\rho'$ is only reached when incurring an overflow, for every reset position but the first one there exists at least one request in~$\rho$ that is open for at least~$b+1$ increment-edges.
	
	We now prove that the play~$\rho$ in~$\game$ is winning for Player~$1$.
	Recall that the play $\rho' = (v_0,o_0,r_0)(v_1,o_1,r_1)(v_2,o_2,r_2)\cdots$ is the extension of~$\rho$ with respect to~$\mem'$.
	In order to show~$\rho$ to be winning, we proceed in two steps.
	First, we show that we have~$o_j < n$ for all~$j$, i.e., the strategy~$\tau$ always uses meaningful moves of~$\tau'$ for its choice of move.
	This allows us to argue that~$\tau$ is indeed winning for Player~$1$.
	
	First, note that even though~$\rho'$ is not a play in~$\extgame$, every vertex~$(v_j, o_j, r_j)$ of $\rho'$ is in~$\reach$.
	We show this claim by induction over~$j$.
	For $j = 0$, we obtain~$(v_0,o_0,r_0) = v'_\initmark \in \reach$.
	For $j > 0$, we obtain~$v'_{j-1} = (v_{j-1},o_{j-1},r_{j-1}) \in \reach$ by induction hypothesis.
	Hence, let~$\pi$ be a play prefix ending in~$v'_{j-1}$ that is consistent with~$\tau'$.
	If $v'_{j-1} \in V'_0$, then $\pi \cdot (v_j, o_j, r_j)$ is consistent with~$\tau'$.
	Otherwise, i.e., if $v'_{j-1} \in V'_1$, then we obtain $\tau(\pi) = (v_j, o, r_j)$ for some~$o \in [n+1]$.
	In case $o = o_j$, $(v_j, o_j, r_j) \in \reach$ follows directly.
	If $o \neq o_j$, however, then $o_j = o_{v_j}$ and $r_j = r_{v_j}$, i.e., $(v_j, o_j, r_j) \in \reach$ by definition of~$o_{v_j}$.
		
 	Next, we show $o_{j+1} \leq o_j+1$ for all $j \in \nats$.
	Assume we have $o_{j+1} > o_j$ for some~$j$.
	As argued above, the vertex~$(v_j, o_j, r_j)$ is in~$\reach$.
	As $o_j \neq o_{j+1}$, there is an edge from $(v_j, o_j, r_j)$ to $(v_{j+1}, o_j+1, r_{j+1})$ in $\arena'$, where $r_{j+1} = r_{v_{j+1}}$. 
	By construction of $\tau$, we obtain~$(v_{j+1}, o_j+1, r_{j+1}) \in \reach$ as argued above.
	Hence, we indeed have $o_{j+1} = o_{v_{j+1}} \le o_j+1$.
	
	Now, we argue $o_j < n$ for all $j \in \nats$ 
by proving the following property by induction over~$j$:
	\begin{quote}
		If $o_j = k$, then for every $k' \le k$ there is a reset position~$j_{k'} \le j$ with $o_{j_{k'}} = k'$.
	\end{quote}
	Let us first argue that this indeed implies $o_j < n$. 
	Towards a contradiction, assume $o_j = n$. 
	Then, there are $n+1$ reset positions, one for each value~$k$ in the range~$0 \le k \le n$ for the overflow counter. 
	Thus, two such positions~$j', j''$ share the same vertex~$v_{j'} = v_{j''}$, which implies that they also share the same overflow counter value $o_{j'} = o_{v_{j'}} = o_{v_{j''}} = o_{j''}$. 
	This yields the desired contradiction to $o_{j'}$ and $o_{j''}$ being distinct. 
	
	For the induction start, we  have $o_0 = 0$ and pick $j_0 = 0$, which is a reset position. 
	Now, let $j >0$ and~$o_j = k$.
	If $k \le o_{j-1}$, then the induction hypothesis yields the necessary positions. 
	Hence, assume we have $k > o_{j-1}$, which implies $k = o_{j-1}+1$ as shown above. Then, $j$ is a reset position and we can define $j_{k} = j$ and obtain the remaining $j_{k'}$ for $k' < k$ from the induction hypothesis.
	
	It remains to show that~$\rho$ is indeed winning for Player~$1$.
	First assume that the overflow counter of~$\rho'$ stabilizes, i.e., there exists some~$j \in \nats$ such that~$o_{j'} = o_j$ for all~$j' > j$.
	Then, there exists a suffix of~$\rho'$ that is consistent with~$\tau'$, which therefore violates the parity condition.
	Hence, it suffices to note that the colors of~$\rho'$ and~$\rho$ coincide, i.e.,~$\rho$ violates the parity condition and thus also the parity condition with costs with respect to any bound.
	
	Now assume that the overflow counter of~$\rho'$ does not stabilize.
	Then, there are infinitely many reset positions in~$\rho'$.
	Between any two adjacent such positions, by construction, there exists a request that remains unanswered for at least~$b+1$ steps in~$\rho$.
	Hence,~$\rho$ violates the parity condition with costs with respect to bound~$b$ and is winning for Player~$1$.
\end{proof}

As the parity game $\extgame$ is of exponential size and can therefore not be constructed and solved in polynomial space in $\size{\game}$, we now construct a finite-duration variant $\finitegame$ of $\extgame$.
One such variant is obtained by playing the parity game up to the first vertex repetition and declaring the winner according to the maximal color on the induced cycle~\cite{AminofRubin14}.
However, one can show that such a play in $\extgame$ is still of exponential length in the worst case. 
In the following, we exploit the structure of the arena to proclaim a winner after a polynomial number of moves.
In particular, we define a preorder on the memory elements and stop a play as soon as a pseudo-cycle is reached, i.e., an infix whose projection to $V$ is a cycle in $\arena$ and whose memory states at the start and at the end are in the order relation.

We first introduce the preorder on memory states. 
To this end, note that not all open requests are \myquot{relevant}.
In fact, a small request that is opened while a larger one is already open is irrelevant.
Answering the larger request is more urgent, as it has already incurred at least as much cost as the newly opened request and answering the larger one answers the smaller one as well.
Following this intuition, we define the relevant requests of a request function~$r$ as follows:~$c$ is relevant in~$r$ if and only if~$r(c) \neq \bot$ and if there does not exist a color~$c' > c$ such that~$r(c') \geq r(c)$.
Then, for each open request~$c$ in~$r$, there exists some~$c' \geq c$ such that the request for~$c'$ is relevant and~$r(c') \geq r(c)$.
Also, the largest open request and the one with the highest cost are always relevant.
We denote the set of relevant requests of~$r$ by~$\relreq(r)$.

Using the relevant requests, we define a preorder~$\dominatedby$ on request functions: For two request functions~$r$ and $r'$, we say that $r'$ dominates~$r$, if for each color~$c$ that is relevant in~$r$, there exists a relevant request~$c'$ in~$r'$ that is \myquot{more urgent} than that for color~$c$.
Formally, we write~$r \dominatedby r'$ if and only if for each~$c$ that is relevant in~$r$, there exists a~$c'$ with $c \leq c'$ that is relevant in~$r'$ with~$r(c) \leq r'(c')$.
The relation~$\dominatedby$ is reflexive and transitive.
Moreover, if~$r \dominatedby r'$ and~$r \dominates r'$ both hold true, we write~$r \dominationequivalent r'$.
Finally, we have that~$r \dominationequivalent r'$ implies~$\relreq(r) = \relreq(r')$ and~$r(c) = r'(c)$ for all~$c \in \relreq(r)$.

We extend the preorder~$\dominatedby$ to memory elements.
Since the overflow counter is non-decreasing and every one of its increments brings a play closer towards the winning sink vertices for Player~$1$, we value this component of the memory state more strongly.
Following this intuition, we say that the memory state~$(o,r)$ is dominated by the memory state~$(o',r')$, written $(o,r) \dominatedby (o',r')$, if either~$o < o'$, or if both~$o = o'$ and~$r \dominatedby r'$ hold true.
Similarly, we extend the notation of~$\dominationequivalent$ such that~$(o,r) \dominationequivalent (o',r')$, if and only if both~$(o,r) \dominatedby (o',r')$ and~$(o,r) \dominates (o',r')$ hold true.

The preorder on memory elements is preserved under concatenation.

\begin{lem}
\label{lem:dominating-memory:stable-concatenation}
Let~$(v,o_1,r_1),(v,o_2,r_2)$ be vertices in~$\extgame$, let~$(v,v') \in E$, and, for~$j \in \set{1,2}$, let $\update((o_j,r_j),(v,v')) = (o'_j,r'_j)$.
If~$(o_1,r_1) \dominatedby (o_2,r_2)$ and~$o_2 < n$, then~$(o'_1,r'_1) \dominatedby (o'_2,r'_2)$.
\end{lem}
\begin{proof}
First assume~$o_1 < o_2$.
Due to the construction of~$\extgame$, this implies~$o'_1 \leq o'_2$.
If~$o'_1 < o'_2$, the given statement holds true.
If~$o'_1 = o'_2$, however, then~$r'_1 = r_{v'}$ and thus~$(o'_1,r'_1) \dominatedby (o'_2,r'_2)$, since for every~$(v,o,r)$ with incoming edges in $\arena'$, we have~$r \dominates r_v$.
Thus, assume~$o_1 = o_2$ and~$r_1 \dominatedby r_2$ for the remainder of this proof.

First, assume that~$o'_1 = o_1 + 1$
We show the following claim: If the move to~$v'$ causes an overflow when starting from~$(v,o_1,r_1)$, then the same move causes an overflow when starting from~$(v,o_2,r_2)$.
This then implies~$r'_1 = r'_2 = r_{v'}$ and hence, $(o'_1,r'_1) \dominatedby (o'_2,r'_2)$.
Formally, we claim that if~$o'_1 = o_1 + 1$, then~$o'_2 = o_2 + 1$.
It remains to show~$o'_2 = o_2 + 1$.
Since the move from~$v$ to~$v'$ causes an overflow when starting in~$(v,o_1,r_1)$, we have~$\cost(v, v') = \inc$.
Let~$c_1$ be some color that causes the overflow in~$r_1$, i.e.,~$c_1 \in \set{c \mid r_1(c) = b}$.
Since~$r_2 \dominates r_1$, there exists a color~$c_2 \geq c_1$ with~$r_2(c_2) \geq r_1(c_1)$.
Moreover, since the range of~$r_2$ is upwards bounded by~$b$, this implies~$r_2(c_2) = b$.
Hence, the move from~$v$ to~$v'$ also causes an overflow in~$r_2$, which implies~$o'_2 = o_2 + 1$.
This completes the proof in the case~$o'_1 = o_1 + 1$.

If, however,~$o'_1 = o_1$, we again distinguish two cases:
If~$o'_2 = o_2 + 1$, then~$o'_1 < o'_2$ and hence,~$(o'_1,r'_1) \dominatedby (o'_2,r'_2)$.
On the other hand, if~$o'_2 = o_2$, then let~$c_1$ be relevant in~$r'_1$.
We show that there exists a color~$c_2 \geq c_1$ with~$r'_2(c_2) \geq r'_1(c_1)$.
Should~$c_2$ be non-relevant, then there exists a larger one in~$r'_2$ that dominates~$c_2$.

If a request for~$c_1$ was already open in~$r_1$, then let~$c_2 \geq c_1$ with~$r_2(c_2) \geq r_1(c_1)$.
Such a color~$c_2$ exists due to~$r_2 \dominates r_1$.
Since the request for~$c_1$ was not answered during the move to~$v'$, and since~$o_2 = o_2' < n$, neither was the request for~$c_2$ during the same move.
Hence, we have~$r'_2(c_2) \geq r'_1(c_1)$.
If, however, a request for~$c_1$ was not already open in~$r_1$, then the request for~$c_1$ must have been opened by moving to~$v'$, i.e.,~$\col(v') = c_1$.
Thus, we directly obtain~$r'_1(c_1) = 0$ and $r'_2(c_1) \geq 0$.
Picking~$c_2 = c_1$ concludes the proof in this case.
\end{proof}

We now define the winning condition in the finite game~$\finitegame$.
To this end, let~$\pi = (v_0,o_0,r_0)\cdots(v_j,o_j,r_j)$ be a play prefix in~$\extgame$ and let~$\pi' = (v_k, o_k, r_k)\cdots(v_{k'},o_{k'},r_{k'})$ be an infix of~$\pi$.
We say~$\pi'$ is a dominating cycle\footnote{Note that a dominating cycle is only a cycle when projected to a play in $\arena$.} if~$v_k = v_{k'}$,~$o_k = o_{k'} < n$, and either
\begin{itemize}
	\item the maximal color occurring on~$\pi'$ is even and~$r_k \dominates r_{k'}$, or
	\item the maximal color occurring on~$\pi'$ is odd and~$r_k \dominatedby r_{k'}$.
\end{itemize}
We call the former and latter type of dominating cycles even and odd, respectively.
Moreover, we say that a play prefix~$\pi = (v_0,o_0,r_0)\cdots(v_j,o_j,r_j)$ is \emph{settled} if either~$o_j = n$, or if~$\pi$ contains a dominating cycle.\footnote{This definition differs from the one presented in the conference version~\cite{WeinertZimmermann16}. The definition here is easily amenable to the case of integer-valued cost functions in Section~\ref{sec:concrete}, simplifies the proofs of Lemma~\ref{lem:pspace-mem:unary:unsettled-bound} and Lemma~\ref{lem:unary:infinite-to-finite}, and fixes a bug in the conference version, which caused plays to be settled too early in favor of Player~$1$.}
Fix~$\ell = (n+1)^6$.

\begin{lem}
\label{lem:pspace-mem:unary:unsettled-bound}
Let~$\pi$ be a play prefix of~$\extgame$.
If~$\card{\pi} > \ell$, then~$\pi$ is settled.	
\end{lem}

\begin{proof}
Let $\pi = (v_0, o_0, r_0) \cdots (v_j, o_j, r_j)$ be an unsettled play prefix of~$\extgame$.
We show~$\card{\pi} \leq \ell$, which implies the given statement.
Note that, since~$\pi$ is not settled, it does not contain a vertex repetition, since such a repetition induces a dominating cycle.

The structure of our argument is sketched in Figure~\ref{fig:pspace-membership:play-analysis}: We recall the definition of overflow positions, define debt-free, request-adding, and relevance-reducing positions, and show
\begin{enumerate}
	\item that there are at most $n$ overflow positions in $\pi$,
	\item that there are at most $n$ debt-free positions between any two adjacent overflow positions,
	\item that there are at most $d$ request-adding positions between any two adjacent debt-free positions, where $d$ is the number of odd colors,
	\item that there are at most~$d$ relevance-reducing positions between any two adjacent request-adding positions,
	\item that there are at most $b+1$ increment-edges between any two adjacent relevance-reducing positions, and
	\item that there are at most~$n$ vertices between two such increment-edges.
\end{enumerate}
Aggregating these bounds then yields the desired result.

\begin{figure}
	\centering
	\begin{tikzpicture}[thick,
		vertex/.style={draw,shape=circle,fill=black,minimum size=.15cm,inner sep=0,line width=.3mm},
		specvertex/.style={draw,shape=circle,fill=mydarkyellow,minimum size=.15cm,inner sep=0, line width=.15mm}
		]
	
		\begin{scope}
			\node[anchor=east] at (-.2, 0) {$\pi$};
			
			\foreach \col [evaluate=\col as \xpos using \col] in {0,2,4,6,8}
				\node[specvertex] (\col-0) at (\xpos, 0) {};
				
			\node[vertex] (10-0) at (10, 0) {};
					
			\foreach \col [evaluate=\col as \xpos using \col] in {1,3,5,7,9}
				\node[] (\col-0) at (\xpos, 0) {$\dots$};
					
			\foreach \i [evaluate=\i as \nexti using int(\i+1)] in {0,...,9}
				\path (\i-0) edge (\nexti-0);
			
			\node[anchor=west,align=left] (overflow-label) at (10.5,0) {$\le n$ overflow \\ positions};
		\end{scope}
		
		\begin{scope}[shift={(0,-1)}]
			\node[anchor=east] at (-.2,0) {$\pi_1$};
			
			\foreach \col [evaluate=\col as \xpos using \col] in {0,1,2,4,6,7,8,10}
				\node[specvertex] (\col-1) at (\xpos,0) {};
					
			\foreach \col [evaluate=\col as \xpos using \col] in {3,5,9}
				\node[] (\col-1) at (\xpos,0) {$\dots$};
					
			\foreach \i [evaluate=\i as \nexti using int(\i+1)] in {0,...,9}
				\path (\i-1) edge (\nexti-1);	
				
			\node[anchor=west,align=left] (debt-free-label) at (10.5,0) {$\le n$ debt-free \\ positions};
		\end{scope}

		\begin{scope}[shift={(0,-2)}]
			\node[anchor=east] at (-.2,0) {$\pi_2$};
			
			\foreach \col [evaluate=\col as \xpos using \col] in {0,1,2,4,6,8,10}
				\node[specvertex] (\col-2) at (\xpos,0) {};
				
			\node[vertex] (8-2) at (8,0) {};
			\node[vertex] (10-2) at (10,0) {};
					
			\foreach \col [evaluate=\col as \xpos using \col] in {3,5,7,9}
				\node[] (\col-2) at (\xpos,0) {$\dots$};
				
			\foreach \i [evaluate=\i as \nexti using int(\i+1)] in {0,...,9}
				\path (\i-2) edge (\nexti-2);
				
			\node[anchor=west,align=left] (request-add-label) at (10.5,0) {$\leq d$ request-\\adding positions};
		\end{scope}
		
		\begin{scope}[shift={(0,-3)}]
			\node[anchor=east] at (-.2,0) {$\pi_3$};
			
			\foreach \col in {2,3,5,7,9}
				\node[specvertex] (\col-3) at (\col,0) {};
				
			\node[vertex] (0-3) at (0,0) {};
			\node[vertex] (10-3) at (10,0) {};
					
			\foreach \col in {1,4,6,8}
				\node[] (\col-3) at (\col,0) {$\dots$};
				
			\foreach \i [evaluate=\i as \nexti using int(\i+1)] in {0,...,9}
				\path (\i-3) edge (\nexti-3);
				
			\node[anchor=west,align=left] (rel-req-answer-label) at (10.5,0) {$\leq d$ relevance-\\ reducing positions};
		\end{scope}
		
		\begin{scope}[shift={(0,-4)}]
			\node[anchor=east] at (-.2,0) {$\pi_4$};
			
			\foreach \col in {0,1,3,4,5,7,8,10}
				\node[vertex] (\col-4) at (\col,0) {};
					
			\foreach \col in {2,6,9}
				\node[] (\col-4) at (\col,0) {$\dots$};
				
			\foreach \i [evaluate=\i as \nexti using int(\i+1)] in {0,3,4,7}
				\path (\i-4) edge node[anchor=south] (increment-\i) {$\inc$} (\nexti-4);
			
			\foreach \i [evaluate=\i as \nexti using int(\i+1)] in {1,2,5,6,8,9}
				\path (\i-4) edge node[anchor=south] {$\eps$} (\nexti-4);
				
			\node[anchor=west,align=left] (increment-label) at (10.5,0) {$\leq b+1$ \\ in\-crement edges};

			\draw[decorate,decoration={brace,amplitude=3pt}]
				($(2-4.south east) + (.1cm,0)$) --
					node[anchor=north,yshift=-.1cm,align=center] {$\leq n$ \\ vertices}
					($(2-4.south west) - (.1cm,0)$);
			\draw[decorate,decoration={brace,amplitude=3pt}]
				($(6-4.south east) + (.1cm,0)$) --
					node[anchor=north,yshift=-.1cm,align=center] {$\leq n$ \\ vertices}
					($(6-4.south west) - (.1cm,0)$);
			\draw[decorate,decoration={brace,amplitude=3pt}]
				($(9-4.south east) + (.1cm,0)$) --
					node[anchor=north,yshift=-.1cm,align=center] {$\leq n$ \\ vertices}
					($(9-4.south west) - (.1cm,0)$);
		\end{scope}

		\draw[dashed,gray]
			(3-0.south west) ..
			controls (2-1.north west) and (1-0.south west) ..
			(0-1.north west);
		\draw[dashed,gray]
			(3-0.south east) ..
			controls (3-1.north east) and (10-0.south east) ..
			(10-1.north east);
					
		\draw[dashed,gray]
			(5-1.south west) ..
			controls (5-2.north west) and (0-1.south west) ..
			(0-2.north west);
		\draw[dashed,gray]
			(5-1.south east) ..
			controls (5-2.north east) and (10-1.south east) ..
			(10-2.north east);	
		
		\draw[dashed,gray]
			(5-2.south west) ..
			controls (4-3.north west) and (0-2.south west) ..
			(0-3.north west);
		\draw[dashed,gray]
			(5-2.south east) ..
			controls (6-3.north east) and (10-2.south east) ..
			(10-3.north east);
		
		\draw[dashed,gray]
			(4-3.south west) ..
			controls (3-4.north west) and (0-3.south west) ..
			(0-4.north west);
		\draw[dashed,gray]
			(4-3.south east) ..
			controls (5-4.north east) and (10-3.south east) ..
			(10-4.north east);
	\end{tikzpicture}

	\caption{Bounding the length of unsettled play prefixes. The relevant special vertices are marked in yellow.}
	\label{fig:pspace-membership:play-analysis}	
\end{figure}
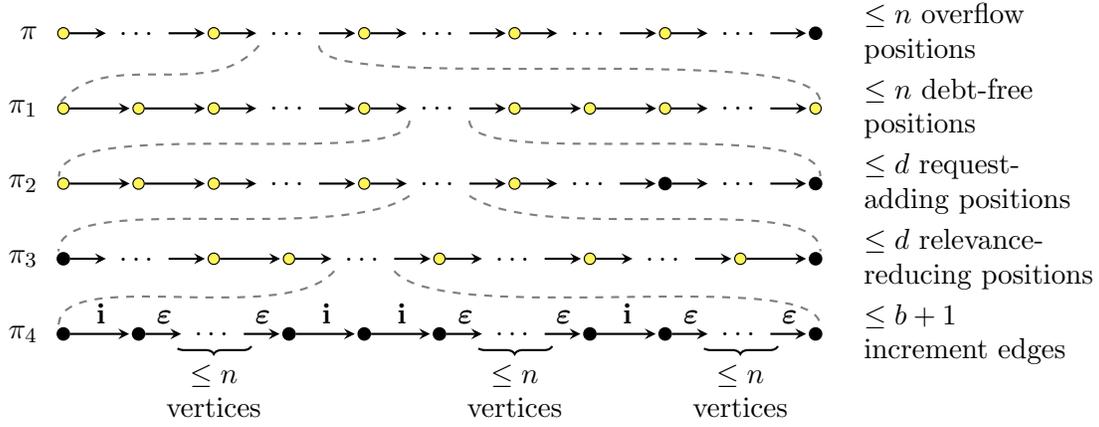
Recall that an overflow position of $\pi$ is a $k$ with $k = 0$ or with $o_k = o_{k-1}+1$.
As~$\pi$ is unsettled and the $o_k$ are non-decreasing, $\pi$ has at most $n$ overflow positions, $n-1$ real increments and the initial position.
Hence, by splitting $\pi$ at the overflow positions we obtain at most $n$ non-empty infixes of $\pi$, each without overflow positions.
We say such an infix has type~$1$.

Fix a non-empty type~$1$ infix $\pi_1$.
A \emph{debt-free position} of $\pi$ is a $k$ with $r_k = r_{v_k}$, i.e., a position that has no other costs than those incurred by visiting $v_k$.
As all vertices of $\pi_1$ share the same overflow counter value, there are at most $n$ debt-free positions in $\pi_1$: $n+1$ such positions would induce a vertex repetition, which we have ruled out above.
Hence, by splitting $\pi_1$ at the debt-free positions we obtain at most $n+1$ non-empty infixes of $\pi_1$, each without debt-free and overflow positions.
We say such an infix has type~$2$.

Fix a non-empty type~$2$ infix $\pi_2$.
A \emph{request-adding position} of $\pi$ is a $k$ with odd $\col(v_k)$ such that $r_{k-1}(c) = \bot$ for all $c \geq \col(v_k)$.
We define $d$ as the number of odd colors assigned by $\col$ and claim that there are at most $d$ request-adding positions in $\pi_2$.
Assume there are $d+1$. Then, two request-adding positions~$k<k'$ share a color, call it $c$.
As $k'$ is request-adding, only requests strictly smaller than $c$ are open at position~$k'-1$, i.e., $c$ and all larger requests have to be answered in between $k$ and $k'$.
Hence, there is a debt-free position between $k$ and $k'$, which contradicts $\pi_2$ being of type 2.
Hence, by splitting $\pi_2$ at the request-adding positions we obtain at most $d+1$ non-empty infixes of $\pi_2$, each without request-adding, debt-free, and overflow positions.
We say such an infix has type~$3$.

Fix a non-empty type~$3$ infix~$\pi_3$.
A \emph{relevance-reducing position} of~$\pi$ is a~$k$ such that $\relreq(r_{k-1}) \supsetneq \relreq(r_k)$.
We show that~$\pi_3$ contains at most~$d$ relevance-reducing positions.
To this end, we first argue that there is at least one request that is open throughout~$\pi_3$.
First note that some request must be open at the beginning of~$\pi_3$, as otherwise the first vertex would be at a debt-free position, which do not occur in~$\pi_3$.
Let~$c^*$ be the maximal open request at the beginning of~$\pi_3$.
Due to~$\pi_3$ not containing debt-free nor request-adding positions, all colors~$c$ visited during~$\pi_4$ must satisfy~$c \leq c^*$.
Hence, the request for~$c^*$ remains open throughout~$\pi_3$.
We now show that the sets of relevant requests along~$\pi_3$ form a descending chain in the subset-relation.
Assume towards a contradiction that an infix $(v,o,r)(v',o',r')$ of~$\pi_3$ and a color~$c$ exist such that~$c \notin \relreq(r)$, but~$c \in \relreq(r')$, i.e, $\col(v') = c$.
Then,~$c > c^*$, i.e., there exists a request-adding position in~$\pi_3$, a contradiction.
Hence, the sets of relevant requests indeed form a descending chain.
As at the beginning of~$\pi_3$ at most~$d$ requests are relevant, there are at most~$d$ relevance-reducing positions.
Thus, by removing relevance-reducing positions, we obtain at most~$d+1$ non-empty infixes, each without relevance-reducing, request-adding, debt-free, and overflow positions.
We say such an infix has type~$4$.

Fix a non-empty type~$4$ infix $\pi_4$.
We show that $\pi_4$ contains at most $b$ increment-edges.
As argued above, there exists some color~$c$ for which a request is open at the beginning of~$\pi_4$, which is not answered throughout this infix.
Thus, $b+1$ increment-edges in $\pi_4$ would lead to an overflow position.
However, $\pi_4$ has no overflow positions by construction.
Thus, there are at most $b$ increment-edges in $\pi_4$.
Hence, by splitting $\pi_4$ at the increment-edges, we obtain a decomposition of $\pi_4$ into at most $b+1$ infixes, each without increment-edges and without relevance-reducing, request-adding, debt-free, and overflow positions.
We say such an infix has type~$5$.

Fix a non-empty type~$5$ infix~$\pi_5$.
We show that~$\pi_5$ is of length at most~$n$.
Assume towards a contradiction that~$\pi_5$ contains at least~$n+1$ vertices.
Then there exists an infix~$\pi' = (v,o,r)\cdots(v,o,r')$ of~$\pi_5$, since~$\pi_5$ does not contain overflow positions.
As argued above, we have~$\relreq(r) = \relreq(r')$.
Moreover, as~$\pi_5$ contains no increment-edges, we furthermore obtain~$r \dominationequivalent r'$.
Thus,~$\pi'$ is a dominating cycle, which contradicts~$\pi$ being unsettled.
Hence,~$\pi_5$ is of length at most~$n$.

Aggregating all these bounds yields an upper bound of~$(n+1)^6$ on the length of the unsettled play prefix~$\pi$, as we have $d \le n$ and $b \le n$.
\end{proof}

Using the notion of settled plays, we now define the finite game~$\finitegame = (\arena', \wincond_\finitemark)$, in which both players try to settle the play in their favor.
Due to Lemma~\ref{lem:pspace-mem:unary:unsettled-bound}, every play in~$\arena'$ is settled.
Thus, let~$\rho$ be an infinite play in~$\arena'$ and let~$\pi$ be the minimal settled prefix of~$\rho$, which can be settled due to three mutually exclusive criteria.
If~$\pi$ is settled due to containing an even dominating cycle, then~$\rho$ is winning for Player~$0$.
Otherwise, i.e., if~$\pi$ is settled due to saturating the overflow counter or due to containing an odd dominating cycle, $\rho$ is winning for Player~$1$.
The game $\finitegame$ is indeed a game of finite duration, as the winner is certain after $\ell$ moves, due to Lemma~\ref{lem:pspace-mem:unary:unsettled-bound}.
Hence, $\finitegame$ is determined~\cite{Zermelo13}.
Moreover, in order to solve~$\extgame$, it suffices to solve~$\finitegame$.

\begin{lem}
\label{lem:unary:infinite-to-finite}
Player~$0$ wins~$\extgame$ if and only if she wins~$\finitegame$.	
\end{lem}

\begin{proof}
We first show that Player~$0$ wins~$\extgame$ if she wins~$\finitegame$.
Let $\sigma_\finitemark'$ be a winning strategy for Player~$0$ in $\game_\finitemark'$.
We construct a winning strategy~$\sigma'$ for Player~$0$ in $\extgame$ by simulating a play in $\game_\finitemark'$ that is consistent with $\sigma_\finitemark'$.
As this strategy is only useful as long as the simulating play is not settled, we have to keep the simulating play short by removing settling dominating cycles.
We define the simulation~$h\colon (V \times M)^+ \rightarrow (V \times M)^+$ and the strategy~$\sigma'$ simultaneously.
The function~$h$ satisfies the following invariant:
\begin{quote}
	Let $\pi$ be consistent with $\sigma'$ and end in $(v,o,r)$.
	Then, $h(\pi)$ is consistent with $\sigma'_\finitemark$, is unsettled, and ends in $(v,o',r')$ with $(o',r') \dominates (o,r)$. 
\end{quote}
Since~$h(\pi)$ is consistent with the winning strategy~$\sigma'_f$ for Player~$0$ and unsettled, this implies that neither the overflow counter of~$h(\pi)$ nor that of~$\pi$ reaches the value~$n$.

To begin, let $h(\vinit') = \vinit'$, which satisfies the invariant.
Now, assume we have a play prefix~$\pi$ consistent with $\sigma'$ ending in $(v, o, r)$ and let $h(\pi) = (v_0, o_0, r_0) \cdots (v_j, o_j, r_j)$.
We consider two cases, depending on whose turn it is at the last vertex $(v,o,r)$ of $\pi$.

If $(v,o,r) \in V'_1$, Player~$1$ moves to some successor of $(v, o, r)$ in~$\extgame$, say $(v^*, o^*, r^*)$.
Furthermore, define~$(o^*_\finitemark, r^*_\finitemark) = \update((o_j, r_j),(v_j, v^*))$, which is the corresponding memory update in~$\finitegame$.
If $(v,o,r) \in V'_0$, let $\sigma'_\finitemark(h(\pi)) = ( v^*,  o^*_\finitemark, r^*_\finitemark)$ in~$\finitegame$.
We mimic the move to $(v^*, o^*_f, r^*_f)$ in~$\extgame$ by defining $\sigma'(\pi) = (v^*, o^*, r^*)$ with $(o^*,r^*) = \update((o,r), (v, v^*) ))$.
This is well-defined due to the invariant~$v_j = v$.

In both cases, for a given successor~$v^*$ of~$v$ in~$\arena$, we have $(o^*_\finitemark, r^*_\finitemark) = \update((o_j, r_j),(v, v^*))$ and $(o^*,r^*) = \update((o,r), (v, v^*) ))$.
It remains to define $h(\pi \cdot (v^*, o^*, r^*))$.
We distinguish two cases: If $ \pi^*_f = h(\pi) \cdot (v^*, o^*_\finitemark, r^*_\finitemark)$ is not settled, we pick $h(\pi \cdot (v^*, o^*, r^*)) = \pi^*_f$.
This satisfies the invariant due to Lemma~\ref{lem:dominating-memory:stable-concatenation}.

Now assume $\pi^*_f$ is settled.
Since~$\pi^*_f$ is consistent with the winning strategy~$\sigma'_f$ for Player~$0$,~$\pi^*_f$ is settled due to containing an even dominating cycle.
Moreover, since~$h(\pi)$ is not settled, the dominating cycle is a suffix of~$\pi^*_f$.
Thus, the cycle starts in a vertex $(v_{j'}, o_{j'}, r_{j'})$ with $v_{j'} = v^*$ and~$r_{j'} \dominates r^*_f$.
Removing the settling cycle, we define $h(\pi \cdot (v^*, o^*, r^*)) =  (v_0, o_0, r_0) \cdots (v_{j'}, o_{j'}, r_{j'})$, which satisfies the invariant due to transitivity of~$\dominatedby$.

Now, consider a play~$\rho$ consistent with $\sigma'$ and let~$\pi_j$ be the prefix of length~$j$ of~$\rho$.
As argued before, neither the overflow counter of the~$\pi_j$ nor that of the~$h(\pi_j)$ reaches~$n$.
Hence, the colors of the last vertices of~$\pi_j$ and~$h(\pi_j)$ coincide for all~$j$.

Towards a contradiction, assume that the maximal color occurring infinitely often along~$\rho$ is odd, call it~$c$.
After some finite prefix,~$c$ cannot occur on even dominating cycles in the $h(\pi_j)$ anymore, since each occurrence on such a cycle implies at least one occurrence of an even higher even color in $\rho$.
Hence, after this prefix, each time a vertex of color~$c$ is visited, say at the end of the prefix~$\pi_j$, a vertex of the same color is appended to the simulated play~$h(\pi_j)$.
Moreover, this vertex is never removed from the simulated play, since only vertices occurring on even dominating cycles are removed from the simulated play.
Hence, the simulated play becomes longer with each visit to a vertex of color~$c$ after a finite prefix.
This contradicts the~$h(\pi_j)$ being unsettled, as every play of length~$\ell + 1$ is settled due to Lemma~\ref{lem:pspace-mem:unary:unsettled-bound}.
Thus, the maximal color occurring infinitely often in~$\rho$ must be even, i.e.,~$\sigma'$ is winning for Player~$0$ in~$\extgame$.

For the other direction, we show that Player~$1$ wins~$\extgame$ if he wins~$\finitegame$.
Due to determinacy of~~$\extgame$ and $\finitegame$, this suffices to show the result.
Let $\tau_\finitemark'$ be a winning strategy for Player~$1$ in~$\game_\finitemark'$.
Similarly to the previous case, we simulate play prefixes~$\pi$ in~$\extgame$ by play prefixes~$h(\pi)$ in~$\finitegame$ and define a simulation function~$h$ and a winning strategy~$\tau'$ for Player~$1$ in~$\extgame$ simultaneously.
Again, we need to make sure that the simulating play prefixes remain short and unsettled, as long as the overflow counter of the play in~$\extgame$ is not saturated.
We do so by removing settling dominating cycles from the simulating play prefixes in~$\finitegame$.
Formally, $h$ satisfies the following invariant:
\begin{quote}
	Let $\pi$ be consistent with $\tau'$ and end in $(v,o,r)$ with $o <n$.
	Then, $h(\pi)$ is consistent with $\tau_\finitemark'$, is unsettled, and ends in $(v,o',r')$ with $(o',r') \dominatedby (o,r)$. 
\end{quote}
If~$o=n$, we can stop the simulation and let Player~$1$ pick arbitrary successor vertices in~$\extgame$, since the play has reached the winning sink component for Player~$1$ in $\extgame$.

At the beginning, we pick $h(\vinit') = \vinit'$, which satisfies the invariant.
Now, let~$\pi$ be a play prefix consistent with $\tau'$ ending in $(v, o, r)$ and let~$h(\pi)$ be defined.
We obtain the next vertex similarly to the previous case, i.e., as an arbitrary successor of~$(v,o,r)$ if~$(v,o,r) \in V'_0$, and by applying~$\tau'_f$ to~$h(\pi)$ if~$(v,o,r) \in V'_1$.
In the latter case, we moreover define~$\tau'$ as simulating the move of~$\tau'_f$ as previously.
In both cases, we obtain~$v^*$ as the first component of the next vertex in both~$\extgame$ and~$\finitegame$.
Again, we define~$(o^*,r^*)$ and~$(o^*_f,r^*_f)$ as previously.

Let~$\pi^* = \pi \cdot (v^*, o^*, r^*)$.
It remains to define~$h(\pi^*)$.
If $o^* = n$, Player~$1$ has already won and we can define $h(\pi^*)$ arbitrarily since the invariant contains an empty premise.
If~$o^* < n$, however, we define~$h(\pi^*)$ as in the previous case, i.e., as~$\pi^*_f$, if~$\pi^*_f$ is not settled, and by removing the settling odd dominating cycle otherwise.
No even dominating cycle may occur, since~$\pi^*_f$ is consistent with~$\tau'_f$.
This maintains the invariant due to the same argument as previously.

Now consider a play~$\rho$ that is consistent with $\tau'$.
If the overflow counter along $\rho$ reaches the value $n$, then $\rho$ is winning for Player~$1$.
Thus, we consider the case where the counter is always smaller than $n$.
In this case, however, the maximal color seen infinitely often in~$\rho$ is odd, due to the same argument as in the previous case:
If it were even, vertices of that color would be appended to the simulation infinitely often without being removed, contradicting the simulated play being unsettled.
Hence,~$\tau'$ is winning for Player~$1$ in~$\extgame$.
\end{proof}

The combination of Lemmas~\ref{lem:cost-parity-to-parity} and \ref{lem:unary:infinite-to-finite} shows that Player~$0$ wins~$\game$ with respect to the bound~$b$ if and only if she wins~$\finitegame$.
Thus it remains to show that we can simulate $\finitegame$ on an alternating Turing machine in polynomial time.

\begin{lem}
\label{lemma_pspacemembership}
The following problem is in $\pspace$: \myquot{Given a parity game with costs~$\game$ and a bound~$b \in \nats$, does Player~$0$ have a strategy~$\sigma$ for $\game$ with $\cost(\sigma) \le b$?}
\end{lem}

\begin{proof}
Given~$\game$ and~$b$, we show how to simulate the finite-duration game $\finitegame$ on an alternating polynomial time Turing machine using the game semantics of such machines, i.e., two players construct a single path of a run of the machine.
The existential and universal player take the roles of Player~$0$ and Player~$1$, respectively.
The Turing machine keeps track of the complete prefix of the simulated play of $\game_\finitemark'$.
Since every vertex of~$\arena \times \mem$ can be represented in polynomial size and since the length of the play is bounded from above by~$(n+1)^6$ due to Lemma~\ref{lem:pspace-mem:unary:unsettled-bound}, the Turing machine can keep track of the history explicitly and check after each step whether a dominating cycle has occurred in polynomial time.
If the play is settled due to an even dominating cycle, the machine accepts, if it is settled otherwise, the machine rejects.
Note that this algorithm involves neither the explicit construction of $\extgame$ nor that of $\finitegame$.
The Turing machine accepts $\game$ and $b$ if and only if Player~$0$ wins~$\game_\finitemark'$.
Due to Lemma~\ref{lem:pspace-mem:unary:unsettled-bound}, this machine terminates after polynomially many steps.
Hence, $\aptime = \pspace$~\cite{ChandraKS81} completes the proof.
\end{proof}

\subsection{Playing Parity Games with Costs Optimally is PSPACE-hard}
\label{sec_pspacehardness}

Next, we turn our attention to proving a matching lower bound on the complexity, which already holds for finitary parity games, i.e., parity games with costs in which every edge is an increment-edge.
This result is proven by a reduction from the canonical $\pspace$-hard problem~$\qbf$:
Given a quantified boolean formula~$\varphi = Q_1 x_1 Q_2 x_2 \ldots Q_n x_n \psi$ with $Q_i \in \set{\exists, \forall}$ and where~$\psi$ is a boolean formula over the variables~$x_1, x_2, \ldots, x_n$, determine whether $\varphi$ evaluates to true.
We assume w.l.o.g.\ that $\psi$ is in conjunctive normal form such that every conjunct has exactly three literals, i.e., $\psi = \bigwedge_{j = 1}^m (\ell_{j,1} \vee \ell_{j,2} \vee \ell_{j,3} )$, where every~$\ell_{j, k}$ is either~$x$ or~$\overline{x}$ for some~$x \in \set{x_1,\dots,x_n}$.
We call each $\ell_{j,k}$ for $k \in \set{1,2,3}$ a literal and each conjunct of three literals a clause.
Furthermore, we assume w.l.o.g.\ that the quantifiers~$Q_j$ are alternating with $Q_1 = Q_n = \exists$.

This proof uses the standard framework for reducing~$\qbf$ to infinite two-player games of polynomial size:
Player~$0$ implements existential choices, i.e., existential quantification and disjunctions.
Dually, Player~$1$ implements universal quantification and disjunction.
Intuitively, the players pick truth values for the variables in the order as they appear in the quantifier prefix.
Then, Player~$1$ picks a clause and then Player~$0$ picks a literal from that clause.
She wins if and only if the literal evaluates to true under the assignment constructed earlier. 
This requirement has to be encoded by the winning condition, as the polynomial state space of the constructed game is insufficient to encode all possible assignments.
If this is the case, then Player~$0$ wins the game if and only if the formula evaluates to true. 

Here, we employ a finitary parity condition with respect to a given bound to achieve this.
To this end, we encode assigning a truth value to a variable by opening requests, e.g., setting $x_j$ to false requests color~$4j+1$ and setting it to true requests color~$4j+3$.
Crucially, only one color can be requested, but the latter one only after a delay of one step.
Now, Player~$0$ picking a literal~$\ell$ answers the corresponding request, e.g., if $\ell = \overline{x_j}$, then the color~$4j+2$ occurs, and if $\ell =  x_j$, then the color~$4j+4$ occurs.
Note that $4j+4$ answers both the request corresponding to setting $x_j$ to false and the one setting it to true.
But again, the answer~$4j+4$ comes only after a delay of one step, which allows to invalidate this answer in case $x_j$ has been assigned false, by requiring a well-chosen bound.
This is possible with a finitary parity condition with respect to a given bound, but neither with a classical parity condition nor with a finitary parity condition with an existentially quantified bound, which explains the increase in complexity. 

\begin{lem}
\label{lemma_pspacehardness}
The following problem is $\pspace$-hard: \myquot{Given a finitary parity game~$\game$ and a bound~$b \in \nats$, does Player~$0$ have a strategy~$\sigma$ for $\game$ with $\cost(\sigma) \le b$?}
\end{lem}

\begin{proof}
	Let $\varphi = Q_1 x_1 Q_2 x_2 \ldots Q_n x_n \psi$ be a quantified boolean formula with $\psi = \bigwedge_{j = 1}^m C_j$ and $C_j = (\ell_{j,1} \vee \ell_{j,2} \vee \ell_{j,3} )$, where every $\ell_{j,k}$ is either $x$ or $\overline{x}$ for some $x \in \set{x_1,\dots,x_n}$.
	We construct a finitary parity game $\game_\varphi$ such that Player~$0$ has a strategy~$\sigma$ for $\game_\varphi$ with $\cost(\sigma) = 3n + 5$ if and only if the formula $\varphi$ evaluates to true.
	The arena consists of three parts: In the first part, which begins with the initial vertex $v_\initmark$, Player~$0$ and Player~$1$ determine an assignment for the variables $x_1$ through $x_n$, where Player~$0$ and Player~$1$ pick values for the existentially and universally quantified variables, respectively.
	Each choice of a truth value by either player incurs a request.
	In the second part, Player~$1$ first picks a clause, after which Player~$0$ picks a literal from that clause.
	In the last part, the play then proceeds without any choice by the players and checks whether or not the chosen literal was set to true in the first part of the arena.
	If it was set to true, then the corresponding request is answered with cost $3n + 5$.
	Otherwise, that request is answered with cost $3n + 6$.
	Furthermore, all other potentially open requests are answered with cost at most $3n + 5$ and the play returns to the initial vertex $v_\initmark$.
	Thus, all these gadgets are traversed infinitely often and the traversals are independent of each other.
	This idea indeed requires the target of the reduction to be a finitary parity game instead of a classical one, as Player~$0$ is able to answer all requests within at most $3n + 6$ steps independently of the truth value of~$\varphi$.
	Only the additional requirement for her to do so within at most~$3n + 5$ moves forces her to provide a witness for~$\varphi$ being true.
	
	If~$\varphi$ evaluates to true, then Player~$0$ can enforce that all requests are answered with cost at most $3n + 5$.
	Hence, there exists a strategy~$\sigma$ for Player~$0$ with $\cost(\sigma) \leq 3n + 5$.
	If~$\varphi$ evaluates to false, however, then Player~$1$ can enforce requests that remain unanswered for at least $3n + 6$ steps.
	Thus, there exists no strategy~$\sigma$ for Player~$0$ with $\cost(\sigma) \leq 3n + 5$.
	We begin by constructing the arena~$\arena$ together with its coloring~$\col$.
	
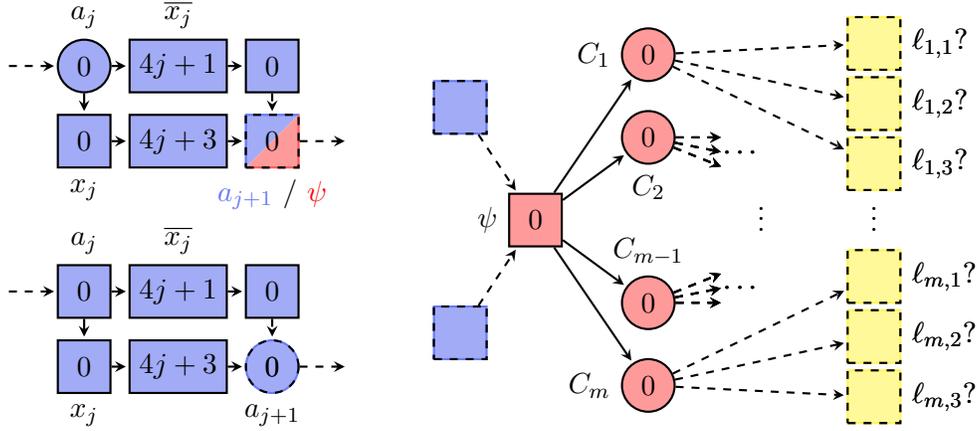
\begin{figure}
	\centering
		\begin{tikzpicture}[thick]
		
        
        \node[p0,assign] (p0-0-0) at (-7,2.05) {$0$};
        \node[p1,assign] (p0-1-0) at (-7,1.05) {$0$};
        \node[p1,assign] (p0-0-1) at (-5.75,2.05) {$4j + 1$};
        \node[p1,assign] (p0-1-1) at (-5.75,1.05) {$4j + 3$};
        \node[p1,assign] (p0-0-2) at (-4.5,2.05) {$0$};
       
        \node[p1,dashed] (p0-1-2) at (-4.5,1.05) {$0$};
        \begin{scope}
        	\clip (p0-1-2.south west) -- (p0-1-2.north east) -- (p0-1-2.north west) -- cycle;
        	\path[draw=mydarkblue,assign,ultra thick] (p0-1-2.south west) rectangle (p0-1-2.north east);
        \end{scope}
        \begin{scope}
        	\clip (p0-1-2.south west) -- (p0-1-2.north east) -- (p0-1-2.south east) -- cycle;
        	\path[draw=mydarkred,choice,ultra thick] (p0-1-2.south west) rectangle (p0-1-2.north east);
        \end{scope}
        \node[p1,dashed] (p0-1-2) at (-4.5,1.05) {$0$};

        \node[anchor=south] at (p0-0-0.north) {$a_j$};
        
        \node[anchor=north] at (p0-1-0.south) {$x_j$};
        \node[anchor=south] at (p0-0-1.north) {$\overline{x_j}$};
        
        \path[draw,dashed] ($(p0-0-0) - (1,0)$) edge (p0-0-0);
        \path[draw,dashed] (p0-1-2) edge ($(p0-1-2) + (1,0)$);
        
        \path[draw] (p0-0-0) edge (p0-0-1);
        \path[draw] (p0-0-0) edge (p0-1-0);
        \path[draw] (p0-0-1) edge (p0-0-2);
        \path[draw] (p0-1-0) edge (p0-1-1);
        \path[draw,dashed] (p0-1-1) edge (p0-1-2);
        \path[draw,dashed] (p0-0-2) edge (p0-1-2);
        
        \node[anchor=north] at (p0-1-2.south) {\textcolor{mydarkblue}{$a_{j + 1}$} / \textcolor{red}{$\psi$}};
        
        \node[p1,assign] (p1-0-0) at (-7,-.95) {$0$};
        \node[p1,assign] (p1-1-0) at (-7,-1.95) {$0$};
        \node[p1,assign] (p1-0-1) at (-5.75,-.95) {$4j + 1$};
        \node[p1,assign] (p1-1-1) at (-5.75,-1.95) {$4j + 3$};
        \node[p1,assign] (p1-0-2) at (-4.5,-.95) {$0$};
        \node[p0,draw=mydarkblue,assign] (p1-1-2) at (-4.5,-1.95) {$0$};
        \node[p0,dashed] (p1-1-2) at (-4.5,-1.95) {$0$};

        \node[anchor=south] at (p1-0-0.north) {$a_j$};
        
        \node[anchor=north] at (p1-1-0.south) {$x_j$};
        \node[anchor=south] at (p1-0-1.north) {$\overline{x_j}$};
        
        \path[draw,dashed] ($(p1-0-0) - (1,0)$) edge (p1-0-0);
        \path[draw,dashed] (p1-1-2) edge ($(p1-1-2) + (1,0)$);
        
        \path[draw] (p1-0-0) edge (p1-0-1);
        \path[draw] (p1-0-0) edge (p1-1-0);
        \path[draw] (p1-0-1) edge (p1-0-2);
        \path[draw] (p1-1-0) edge (p1-1-1);
        \path[draw,dashed] (p1-1-1) edge (p1-1-2);
        \path[draw,dashed] (p1-0-2) edge (p1-1-2);
        
        \node[anchor=north] at (p1-1-2.south) {$a_{j + 1}$};
        
			
		\node[p1,dashed] (lit-true) at (-2,-1.5) {};
		\node[p1,draw=mydarkblue,assign] at (lit-true) {};
		\node[p1,dashed] (lit-true) at (-2,-1.5) {};
		
		\node[p1,dashed] (lit-false) at (-2,1.5) {};
		\node[p1,draw=mydarkblue,assign] at (lit-false) {};
		\node[p1,dashed] (lit-false) at (-2,1.5) {};
		
		\node[p1,label=left:{$\psi$},choice] (pick-clause) at (-1,0) {$0$};
		
		\path[draw,dashed] (lit-true) edge (pick-clause);
		\path[draw,dashed] (lit-false) edge (pick-clause);
		
		\foreach \label/\labelpos/\ypos in {1/left/2.2,2/below/1.1,m-1/above/-1.1,m/left/-2.2}
			\node[p0,label=\labelpos:$C_{\label}$,choice] (c-\label) at (.5,\ypos) {$0$};
			
		\node (c-dots) at (2,0) {$\rvdots$};
				
		\foreach \x [count=\i from 1] in {2.35,1.55,0.75} {
			\node[p1,dashed,label=right:$\ell_{1,\i}?$] (l-1-\i) at (3.5,\x) {};
			\node[p1,draw=mydarkyellow,check] at (l-1-\i) {};
			\node[p1,dashed,label=right:$\ell_{1,\i}?$] (l-1-\i) at (3.5,\x) {};
		}
		
		\foreach \x [count=\i from 1] in {-0.75,-1.55,-2.35} {
			\node[p1,dashed,label=right:$\ell_{m,\i}?$] (l-m-\i) at (3.5,\x) {};
			\node[p1,draw=mydarkyellow,check] at (l-m-\i) {};
			\node[p1,dashed,label=right:$\ell_{m,\i}?$] (l-m-\i) at (3.5,\x) {};	
		}
			
		\node (x-dots) at (3.5,0) {$\rvdots$};
		
		\foreach \j in {1,2,m-1,m}
		  \path (pick-clause) edge (c-\j);
		
		\foreach \literal in {1,2,3} {
			\foreach \clause in {1,m}
				\path[draw,dashed] (c-\clause) edge (l-\clause-\literal);
			
			\foreach \clause/\offset in {2/0cm,m-1/.4cm} {
				\foreach \ypos in {0cm,-.2cm,-.4cm}
					\path[draw, dashed] (c-\clause) edge ($(c-\clause) + (1cm,\offset) + (0,\ypos)$);
				
				\node (c-\clause-dots) at ($(c-\clause) + (0,\offset) + (1.25cm,-.2cm)$) {$\cdots$};
			}
		}
		
		\end{tikzpicture}
		
		\caption{The gadget for existentially and universally quantified variables (left, from top to bottom), and the middle part of the arena (right).}
		\label{fig:pspace-hard:choose-literal}
	\end{figure}	
	
	The left-hand side of Figure~\ref{fig:pspace-hard:choose-literal} shows the gadgets that assign a truth value to variable~$x_j$.
	The vertex~$a_j$ belongs to Player~$0$ if~$x_j$ is existentially quantified, and to Player~$1$ if~$x_j$ is universally quantified.
	The dashed edges indicate the connections to the pre- and succeeding gadget, respectively.
	We construct the first part of~$\arena$ out of~$n$ copies of this gadget.
	Moreover, the vertex~$a_1$ has an incoming edge from the end of~$\arena$, in order to allow for infinite plays, and is the initial vertex~$v_\initmark$ of the arena.	
	In the remainder of this proof, let $c_{x_j} =4j + 3$ and $c_{\overline{x_j}} = 4j + 1$ be the colors associated with assigning true or false to $x_j$, respectively.

	The second part of the arena starts with a vertex~$\psi$ of Player~$1$, from which he picks a clause by moving to a vertex~$C_j$ of Player~$0$.
	Each vertex $C_j$ is connected to three gadgets, one for each of the three literals contained in $C_j$.
	We show this construction in the right-hand side of Figure~\ref{fig:pspace-hard:choose-literal}.
   Note that moving from the vertex of color~$c_{x_j}$ or of color~$c_{\overline{x_j}}$ to the vertex~$\psi$ takes~$3(n-j) + 1$ or~$3(n-j)+2$ steps, respectively..
    
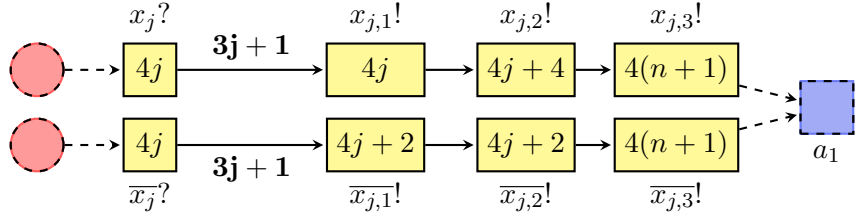
\begin{figure}
	\centering
			\begin{tikzpicture}[thick]
				\node[p1,dashed] (a1-true) at (9,0) {};
				\node[p1,assign,draw=mydarkblue] at (a1-true) {};
				\node[p1,dashed,label=south:$a_1$] (a1-true) at (9,0) {};
			
				\begin{scope}[shift={(0,.5)}]
					\node[p0,dashed] (c-j) at (-1.5,.0) {};
					\node[p0,draw=mydarkred,choice] (c-j) at (-1.5,0) {};
					\node[p0,dashed] (c-j) at (-1.5,0) {};
					
					\node[p1,check,label=north:$x_j?$] (equal-begin) at (0,0) {$4j$};
					\node[p1,check,label=north:$x_{j,1}!$,minimum width=1.3cm] (check-1) at (3,0) {$4j$};
					\node[p1,check,label=north:$x_{j,2}!$] (check-2) at (5,0) {$4j+4$};
					\node[p1,check,label=north:$x_{j,3}!$] (check-3) at (7,0) {$4(n+1)$};
					
					\path[draw] (equal-begin) edge node[anchor=south] {$\mathbf{3j+1}$} (check-1);
					\path[draw] (check-1) edge (check-2);
					\path[draw] (check-2) edge (check-3);
					\path[draw,dashed] (check-3) edge (a1-true);
					\path[draw,dashed] (c-j) edge (equal-begin);
				\end{scope}
				
				\begin{scope}[shift={(0,-.5)}]				
					\node[p0,dashed] (c-j) at (-1.5,0) {};
					\node[p0,draw=mydarkred,choice] at (c-j) {};
					\node[p0,dashed] at (c-j) {};
					
					\node[p1,check,label=south:$\overline{x_j}?$] (equal-begin) at (0,0) {$4j$};
					\node[p1,check,label=south:$\overline{x_{j,1}}!$] (check-1) at (3,0) {$4j+2$};
					\node[p1,check,label=south:$\overline{x_{j,2}}!$] (check-2) at (5,0) {$4j+2$};
					\node[p1,check,label=south:$\overline{x_{j,3}}!$] (check-3) at (7,0) {$4(n+1)$};
					
					\path[draw] (equal-begin) edge node[anchor=north] {$\mathbf{3j+1}$} (check-1);
					\path[draw] (check-1) edge (check-2);
					\path[draw] (check-2) edge (check-3);
					\path[draw,dashed] (check-3) edge (a1-true);
					\path[draw,dashed] (c-j) edge (equal-begin);
				\end{scope}
			
			\end{tikzpicture}
			
			\caption{Gadgets checking the assignment of true to~$x_j$ (top) or to~$\overline{x_j}$ (bottom).}
			\label{fig:pspace-hard:check-eval}

	\end{figure}

	The last part of the arena consists of one gadget for each literal $x_1$, $\overline{x_1}$ through $x_n$, $\overline{x_n}$ occurring in $\varphi$.
	These gadgets check whether or not the literal picked in the middle part was actually set to true in the first part of the arena.
	We show these gadgets in Figure~\ref{fig:pspace-hard:check-eval}.
	
	In these gadgets, neither player has a non-trivial choice.
	Thus, the play proceeds by first answering requests for all colors smaller than $c_{x_j}$ and $c_{\overline{x_j}}$.
	It then either grants the request for color~$c_{x_j}$ after $3j + 2$ steps, or the request for color~$c_{\overline{x_j}}$ after $3j + 1$ steps, both counted from the vertices $x_j?$ and $\overline{x_j}?$, respectively.
	Since traversing the middle part of the arena incurs a constant cost of $2$, a request for color~$c_{x_j}$ has incurred a cost of $3(n-j) + 3$ at~$x_j?$ and $\overline{x_j}?$, while a request for color~$c_{\overline{x_j}}$ has incurred a cost of $3(n-j) + 4$ at these vertices.
	Hence, the total cost incurred by the request for color~$c_{x_j}$ is $(3(n-j)+3) + (3j+2) = 3n + 5$ in the gadget corresponding to $x_j$, and $(3(n-j)+3) + (3j + 3) = 3n + 6$ in the gadget corresponding to $\overline{x_j}$.
	The dual reasoning holds true for requests for color~$c_{\overline{x_j}}$.
	Hence, the bound of $3n + 5$ is only achieved if the request corresponding to the chosen literal was posed in the initial part of the arena.
	
	After traversing this last gadget, all requests are answered in the vertex~$x_{j,3}!$ or~$\overline{x_{j,3}}!$ and the play resets to the initial vertex via an edge to $a_1$.
	
	The size of $\arena$ is polynomial in $\card{\varphi}$:
	The first part consists of one constant-size gadget per variable, while the second part is of linear size in the number of clauses in $\varphi$.
	The final part contains a gadget of size $\bigo(n)$ for each literal occurring in $\varphi$.
	Thus, the size of the arena is in $\mathcal O(n^2+m)$.
It remains to argue that Player~$0$ has a strategy $\sigma$ with $\cost(\sigma) = 3n+ 5$ if and only if $\varphi$ evaluates to true.
	For any quantifier-free boolean formula~$\psi$ that contains variables $x_1$ through $x_n$ and any partial assignment $\alpha\colon \set{x_1,\dots,x_j} \rightarrow \set{\ttrue, \ffalse}$ we denote by $\alpha(\psi)$ the formula resulting from replacing the variables in $\alpha$'s domain with their respective truth values.
	
	It suffices to argue about finite plays that begin and end in $a_1$, as all plays start in $a_1$, visit $a_1$ infinitely often, and all requests are answered before moving back to $a_1$.
	Hence, for the remainder of this proof, we only consider a finite play infix~$\pi$ starting and ending in $a_1$.
	
	First assume that $\varphi$ evaluates to true.
	We construct a strategy $\sigma$ for Player~$0$ with the properties described above.
	Pick $j$ as some index such that $x_j$ is existentially quantified and consider the prefix $\pi'$ of $\pi$ up to, but not including $a_j$.
	We associate $\pi'$ with an assignment $\alpha_{j-1}\colon\set{x_1,\dots,x_{j-1}}\rightarrow\set{\ttrue, \ffalse}$, where $\alpha_{j-1}(x_k) = \ttrue$ if~$\pi'$ visits~$x_k$, and $\alpha_{j-1}(x_k) = \ffalse$ if~$\pi'$ visits~$\overline{x_k}$.
	Due to the structure of the arena, exactly one of these cases holds true, hence $\alpha_{j-1}$ is well-defined.

	For $j=1$, $\exists x_j\cdots Q_n x_n \alpha_{j-1}(\psi)$ evaluates to true by assumption.
	Let $t \in \set{\ttrue, \ffalse}$ such that $\forall x_{j+1} \cdots Q_n x_n (\alpha_{j-1}[x_j \mapsto t])(\psi)$ evaluates to true as well, where $\alpha_{j-1}[x_j \mapsto t]$ denotes the mapping $\alpha_{j-1}$ augmented by the mapping $x_j \mapsto t$.
	Moreover, we define $\sigma(\cdots a_j) = x_j$ if $t = \ttrue$, and $\sigma(\cdots a_j) = \overline{x_j}$ otherwise.
	We proceed inductively, constructing $\sigma(\cdots a_j)$ for all existentially quantified variables $x_j$ according to the Boolean values that satisfy the formulas $\exists x_j Q_{j+1} x_{j+1} \ldots Q_n x_n \alpha_{j-1}(\psi)$, until we reach the vertex $\psi$.
	
	At this point, the analysis of the play prefix so far yields an assignment of truth values to variables $\alpha_n \colon \set{x_1,\dots,x_n} \rightarrow \set{\ttrue, \ffalse}$, such that $\alpha_n(\psi)$ evaluates to true.
	Let $\alpha = \alpha_n$.
	
	At vertex $\psi$ there exist $n$ open requests.
	As previously argued, if $\alpha(x_j) = \ttrue$, then there is an open request for $c_{x_j}$ with cost $3(n-j) + 1$.
	Otherwise, there is an open request for $c_{\overline{x_j}}$ with cost $3(n-j) + 2$.
	At vertex $\psi$, Player~$1$ picks a clause $C_j$ by moving to its vertex.
	Since $\alpha(\psi)$ evaluates to true, there exists a $k \in \set{1,2,3}$ with $\alpha(\ell_{j,k}) = \ttrue$.
	We pick $\sigma(\cdots C_j) = \ell_{j,k}?$.

	If $\ell_{j,k} = x_l$, then $\alpha(x_l) = \ttrue$ and hence, there is an open request for $c_{x_l}$.
	As argued previously, this request is then answered with cost~$3n+5$, since we picked the gadget corresponding to $x_l$.
	Similarly, if $\ell_{j,k} = \overline{x_l}$, then $\alpha(x_l) = \ffalse$ and thus there is an open request for $x_{\overline{x_l}}$, which is answered with cost~$3n+5$ as well.
	All other open requests are answered with cost at most $3n+5$, as argued previously.

	After this traversal of the final gadget, all requests are answered, and the play automatically moves to vertex $a_1$ to begin anew.
	The same reasoning then applies ad infinitum.
	Thus, Player~$0$ is able to answer all requests with a cost of at most $3n + 5$.
	
	Now assume that $\varphi$ evaluates to false.
	Then, irrespective of the choices made by Player~$0$ when constructing $\alpha$ in the first part of the arena, Player~$1$ can pick truth values for the universally quantified variables such that $\alpha(\psi)$ evaluates to false and then pick a clause $C_j$ such that $\alpha(C_j)$ evaluates to false.
	Hence, Player~$0$ has to pick some $\ell_{j,k}$ with $\alpha(\ell_{j,k}) = \ffalse$.
	If $\ell_{j,k} = x_l$, then there is an open request for $c_{\overline{x_j}}$ at $x_{l,1}!$, which is answered with cost $3n+6$.
	Similarly, if $\ell_{j,k} = \overline{x_l}$, then $\alpha(x_l) = \ttrue$, hence there is an open request for $c_{x_l}$, which is also answered with cost~$3n+6$.
	Thus, in each round Player~$1$ can open a request that is only answered with cost at least~$3n+6$, i.e., Player~$0$ has no strategy with cost~$3n+ 5$. 
\end{proof}

\section{Memory Requirements of Optimal Strategies in Parity Games with Costs}\label{sec:memory}

Next, we study the memory needed by both players to play optimally in parity games with costs.
Recall that Player~$0$ always has a positional winning strategy if she wins the game, while Player~$1$ requires infinite memory.
In contrast, our main result of this section shows that the memory requirements of optimal strategies are exponential for Player~$0$, i.e., playing optimally comes at a price in terms of memory, too.

Standard complexity theoretic assumptions already rule out the existence of small strategies:
If they existed, guessing and verifying such a strategy would contradict \pspace-completeness of solving finitary parity games with respect to a given bound.
However, here we explicitly present games in which either player requires exponential memory, which we later use to demonstrate gradual tradeoffs between memory and cost in Section~\ref{sec:tradeoffs}.
We obtain our lower bound by a generalization of a construction of Chatterjee and Fijalkow~\cite{CF13arxiv} which yielded a linear lower bound.

First, however, let us state a corollary of the construction of the parity game $\game'$ in the proof of Lemma~\ref{lemma_pspacemembership}, which gives an exponential upper bound on the necessary memory states for both players.
Recall that the memory structure used in that proof has one counter with a range of size $b+2$ for each odd color.
Furthermore, it has an additional counter that is bounded by $n$, which counts the number of times the bound $b$ is exceeded.
Using similar techniques to \cite{FZ14}, it is possible to remove the overflow counter for Player~$0$: She can play assuming the largest value for this latter counter that still allows her to win.

\begin{cor}
	\label{cor:parity:memory:upper-bound}
	Let $\game$ be a parity game with costs containing~$n$ vertices and~$d$ odd colors.
	\begin{itemize}
		\item If Player~$0$ has a strategy $\sigma$ for $\game$ with $\cost(\sigma) = b$, then she also has a strategy $\sigma'$ with $\cost(\sigma') \leq b$ and $\card{\sigma'} = (b+2)^d$.
		\item If Player~$1$ has a strategy~$\tau$ for~$\game$ with~$\cost(\tau) = b$, then he also has a strategy~$\tau'$ with~$\cost(\tau') \geq b$ and~$\card{\tau'} = n (b+2)^d$.
	\end{itemize}
\end{cor}
Again, our matching lower bounds already hold for finitary parity games, i.e., parity games with costs in which all edges are increment-edges.
These proofs reuse principles underlying the \pspace-hardness proof presented in Section~\ref{sec_pspacehardness}, e.g., having a fixed bound requires a player, in the worst case, to store all open requests in order to answer them timely. 

We begin by showing the lower bound for Player~$0$.

\begin{thm}
	\label{thm:memory-upper-p0}
	For every $d \geq 1$, there exists a finitary parity game~$\game_d$ such that
	\begin{itemize}
	
		\item $\game_d$ has $d$ odd colors and $\card{\game_d} \in \bigo(d^2)$, 
		
		\item Player~$0$ has a strategy~$\sigma$ in~$\game_d$ with~$\cost(\sigma) = d^2 + 2d$, but
	
		\item every strategy~$\sigma$ for Player~$0$ in~$\game_d$ with~$\cost(\sigma) \leq d^2 + 2d$ has size at least~$2^{d-1}$. 
	
	\end{itemize}
	
\end{thm}

\begin{proof}
	Let $d \geq 1$.
	We construct a finitary parity game~$\game_d$ that has the stated properties.
	To this end, after defining~$\game_d$, we construct a strategy with cost~$d^2 + 2d$ for Player~$0$ and argue that it is optimal, followed by the proof that every optimal strategy has at least size~$2^{d-1}$.
	
	The game $\game_d$ is played in rounds.
	In each round, which starts at the initial vertex of~$\game_d$, Player~$1$ poses $d$ requests for odd colors in the range~$1$ through~$2d-1$.
	Subsequently, Player~$0$ gives~$d$ answers using colors in the range $2$ through $2d$.
	If she recalls the choices made by Player~$1$ in the first part of the round, she is able to answer each request optimally.
	Otherwise, we show that Player~$1$ can exploit her having not enough memory in order to force requests to go unanswered for more than~$d^2+2d$ steps.
	After each round, the play returns to the initial vertex in order to allow for infinite plays.
	
	The arena $\arena$ consists of gadgets that each allow exactly one request or response to be made.
	Moreover, each path through a gadget has the same length~$d+2$, including the edge connecting a gadget to its successor.
	However, low-priority requests and responses are made earlier than high-priority ones when traversing such a gadget, due to its structure.
	We show both gadgets in Figure~\ref{fig:memory-upper-bound:player-0:gadgets}.
	The dashed lines denote the edges to the pre- and succeeding gadget and the edge between the final and the initial gadget.
	As the owner of the succeeding vertex depends on the succeeding gadget's owner, we draw it as a diamond.
	
\begin{figure}
	\centering
			\begin{tikzpicture}[thick,yscale=.95]
			
			\newcommand{\rowOne}{0}
			\newcommand{\rowTwo}{-1}
			\newcommand{\rowThree}{-2}
				\foreach \offset/\player/\labelA/\labelB/\labelC/\labelD [ 
					evaluate = \offset as \colOne using (\offset + 1.5),
					evaluate = \offset as \colTwo using (\offset + 2.5),
					evaluate = \offset as \colThree using (\offset + 3.5),
					evaluate = \offset as \colFour using (\offset + 4.5),
					evaluate = \offset as \exitCol using (\offset + 6)
					]
				in {
					0/p1/1/3/2d-3/2d-1,
					6/p0/2/4/2d-2/2d
				} {
				
					\IfStrEq{\player}{p1}{\def\MyStyle{assign}}{\def\MyStyle{choice}}
					\foreach \x/\label/\xCount in {\colOne/\labelA/0,\colTwo/\labelB/1,\colFour/\labelD/3} {
						\node[\player,\MyStyle] (\xCount-0) at (\x,\rowOne) {$0$};
						\node[p1,\MyStyle] (\xCount-1) at (\x,\rowTwo) {$\label$};
						\node[p1,\MyStyle] (\xCount-2) at (\x,\rowThree) {$0$};
					}
					
					\node (2-0) at (\colThree,\rowOne) {$\dots$};
					\node (2-1) at (\colThree,\rowTwo) {$\dots$};
					\node (2-2) at (\colThree,\rowThree) {$\dots$};
					
					\foreach \x in {0,1,3} {
						\path (\x-0) edge (\x-1);
						\path (\x-1) edge (\x-2);
					}
					
					\foreach \y in {0,2} {
						\foreach \x [evaluate=\x as \nextX using int(\x + 1)] in {0,...,2} {
							\path (\x-\y) edge (\nextX-\y);
						}
					}

					\node[p1] (entry) at (\offset,\rowOne) {};
					\begin{scope}
						\clip (entry.south west) -- (entry.south east) -- (entry.north east) -- cycle;
						\path[fill=myblue] (entry.south west) -- (entry.south east) -- (entry.north east) -- cycle;
						\draw[mydarkblue] (entry.south west) -- (entry.south east) -- (entry.north east);
					\end{scope}
					\begin{scope}
						\clip (entry.south west) -- (entry.north west) -- (entry.north east) -- cycle;
						\path[fill=myred] (entry.south west) -- (entry.north west) -- (entry.north east) -- cycle;
						\draw[mydarkred] (entry.south west) -- (entry.north west) -- (entry.north east);
					\end{scope}
					\node[p1,dashed] (entry) at (\offset,\rowOne) {};
					
					\node[diamond,minimum size = .7cm] (exit) at (\exitCol,\rowThree) {};
					\begin{scope}
						\clip (exit.north) -- (exit.west) -- (exit.south) -- cycle;
						\path[fill=myblue] (exit.north) -- (exit.west) -- (exit.south) -- cycle;
						\draw[mydarkblue] (exit.north) -- (exit.west) -- (exit.south);
					\end{scope}
					\begin{scope}
						\clip (exit.north) -- (exit.east) -- (exit.south) -- cycle;
						\path[fill=myred] (exit.north) -- (exit.east) -- (exit.south) -- cycle;
						\draw[mydarkred] (exit.north) -- (exit.east) -- (exit.south);
					\end{scope}
					\node[diamond,minimum size = .7cm,draw,dashed] (exit) at (\exitCol,\rowThree) {};
					
					\path[dashed] (entry) edge (0-0);
					\path[dashed] (3-2) edge (exit);
				}

			\end{tikzpicture}
			\caption{The gadgets for Player~$1$ (left) and Player~$0$ (right).}
			\label{fig:memory-upper-bound:player-0:gadgets}

	\end{figure}
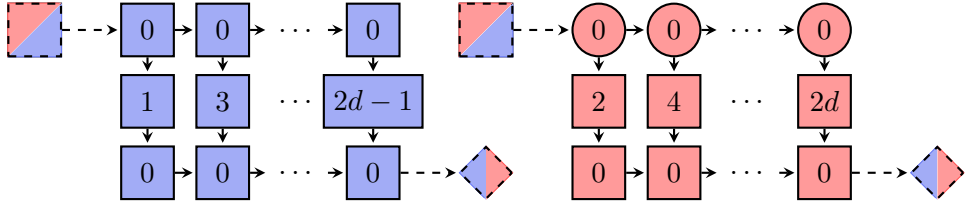
		
	More precisely, the arena~$\arena$, consists of~$d$ repetitions of the gadget for Player~$1$, followed by $d$ repetitions of the gadget for Player~$0$.
	The initial vertex $v_\initmark$ of the arena is the top-left vertex of the first gadget for Player~$1$.
	Moreover, the final gadget of Player~$0$ has a single back-edge to the initial vertex.
	Clearly, $\arena$ satisfies the first statement of the theorem.
			
	Similarly to the proof of Lemma~\ref{lemma_pspacehardness}, it suffices to consider finite plays infixes.
	Even though the requests are not necessarily all answered after each round, we argue that Player~$0$ can always do so while playing optimally.
	
	We now construct an optimal strategy from $v_\initmark$ for Player~$0$.
	In order to play optimally, Player~$0$ needs to track the requests made by Player~$1$ in the first part of each round.
	Instead of tracking each request precisely, however, it suffices to store the order in which the relevant requests were posed.
	In order to use memory efficiently, we do not initialize our memory structure with the empty sequence, but rather with the memory element encoding that all requests are relevant.
	Recall that relevant requests can only be opened by visiting some larger color than all currently open requests.
	Hence, we define the set of strictly increasing odd sequences 
	\begin{align*}
		\incseq_d = \set{ (c_1,\ldots,c_d) \mid &1 \leq c_1 \leq \cdots \leq c_d = 2d-1, \\
		&c_j \neq 2d-1 \text{ implies } c_j < c_{j+1}, \text{ all } c_j \text{ are odd}}
	\end{align*}
	and use them as the set of memory states $M_d = \incseq_d$.~\footnote{This definition differs from the one presented in the conference version~\cite{WeinertZimmermann16}. There, the padding to length~$d$ (here with requests for~$2d-1$) was kept implicit. We show that an optimal strategy for Player~$0$ can be implemented with~$\card{M_d} = 2^{d-1}$ many memory states. This invalidates the lower bound of~$2^d-2$ claimed in the conference version.}
	Note that $\card{M_d} = 2^{d-1}$, as each increasing sequence is isomorphic to a subset of~$\set{1,3,5,\dots,2d-3}$.
	In order to obtain~$\mem$, we define~$m_\initmark = (1,3,\dots,2d-3,2d-1)$.
	Moreover, we define the update function~$\update$ as follows:~$\update(m,e) = m_\initmark$ if~$e$ leads to the initial vertex of~$\game_d$, and as~$\update(m, e) = m$, if~$e$ leads to some other vertex of even color.
	If~$e$ leads to the (unique) vertex of odd color~$c$ in the~$j$-th gadget of Player~$1$, however, we differentiate two cases.
	Let~$m = (c_1,\dots,c_d)$.
	If~$c_j \geq c$, then~$\update(m,e) = m$.
	Otherwise, we define
	\[ \update(m,e) = (c_1,\dots,c_{j-1},c,c+2,c+4,\dots,2d-1,2d-1,\dots,2d-1) \enspace. \]
	Note that this definition of~$\update$ implies that the memory state is fixed once a partial play leaves the gadgets of Player~$1$, remains unchanged throughout the traversal of the gadgets of Player~$0$ and is only reset upon moving to the initial vertex of~$\game_d$.
	Finally, we define the next-move function~$\nxt$ such that, if the play leaves the gadgets of Player~$1$ with memory state~$m = (c_1,\dots,c_d)$, then Player~$0$ moves to color~$c_j + 1$ in her~$j$-th gadget.
	
	Now consider a play in which Player~$0$ plays according to this strategy and consider the request for color~$c$ made by Player~$1$ in his~$j$-th gadget.
	If the request for color~$c$ is the largest request made by Player~$1$ so far in the current round,, then Player~$0$ answers the $j$-th request from Player~$1$ in her $j$-th gadget.
	The cost of this request then consists of three components.
	First, the play has to leave Player~$1$'s $j$-th gadget, incurring a cost of $d - (c+1) / 2 + 2$.
	Then, the play passes through $d-1$ gadgets, incurring a cost of~$d+2$ in each.
	Finally, moving to color~$c+1$ in Player~$0$'s gadget incurs a cost of $(c+1)/2$. 
	Hence, answering Player~$1$'s request incurs a cost of $d^2 + 2d$.
	If there, however, exists some~$j' < j$ such that Player~$1$ has already requested some color~$c' > c$ in his~$j'$-th gadget, then the request for color~$c'$ will be answered in gadget~$j'$ of Player~$0$.
	Hence, Player~$0$ answers the request with cost at most $(d-1)(d+2) < d^2 + 2d$.
	
	This cost is indeed optimal.
	Consider the play in which Player~$1$ always requests $2d-1$.
	Even if Player~$0$ answers this request in her first gadget, it still incurs a cost of $d^2 + 2d$.
			
	It remains to show that no optimal strategy of size less than $\card{M_d}$ exists.
	To this end we show that Player~$1$ can exploit a strategy of Player~$0$ with less than~$\card{M_d}$ memory states and pose requests that will be answered suboptimally.
	
	We associate with each~$m \in M_d$ the partial play~$\req(m)$ which starts in the initial vertex, where Player~$1$ requests the colors occurring in~$m$ in order.
	Clearly,~$m \neq m' \in M_d$ implies~$\req(m) \neq \req(m')$.
	
	Let~$\sigma$ be a strategy for Player~$0$ that is implemented by $(M, \init, \update)$ with $\card{M} < \card{M_d}$ and let $m \in M$.
	Due to the pigeon-hole principle, there exist $m'_1 \neq m'_2 \in M_d$, such that $\update^+(m, \req(m'_1)) = \update^+(m, \req(m'_2))$, i.e., Player~$0$ answers both sequences of requests in the same way.
	Since $\req(m'_1) \neq \req(m'_2)$, there exists a gadget of Player~$1$ in which the requests posed during $\req(m'_1)$ and $\req(m'_2)$ differ.
	Pick $j$ as the minimal index of such a gadget and assume that in his $j$-th gadget, Player~$1$ requests color $c$ during $\req(m'_1)$, and color $c'$ during $\req(m'_2)$, where, w.l.o.g., $c < c'$.
	If Player~$0$ has already answered the request for $c'$ before her $j$-th gadget, then some earlier request is not answered optimally when reacting to~$\req(m'_2)$, as requests are posed in strictly increasing order.
	Thus, assume that the request for color~$c'$ has not been answered upon entering Player~$0$'s $j$-th gadget.
	If she visits some color $c'' < c'$ in this gadget, she will only answer $c'$ in some later gadget, thereby incurring a cost of more than $d (d+2) = d^2 + 2d$ when reacting to~$\req(m'_2)$.
	If she visits some color $c'' \geq c'$, then she does not answer the request for $c$ optimally, thus incurring a cost of at least $d^2 + 2d + (c' - c)/2 > d^2 + 2d$ when reacting to~$\req(m'_1)$.
	Hence, one of the sequences of requests $\req(m'_1)$ or $\req(m'_2)$ leads to Player~$0$ answering at least one request non-optimally.
	As such sequences $m'_1$ and $m'_2$ exist for each $m \in M$, Player~$1$ can force such an \myquot{expensive} request in each round.
	Thus, $\cost(\sigma) > d^2 + 2d$, i.e., $\sigma$ is not~optimal.
\end{proof}

After having shown that Player~$0$ requires exponential memory to keep the cost of the play below a given bound, we now show a similar result for Player~$1$:
In general, he also requires exponential memory to enforce costs above a given bound.

\begin{thm}
	\label{thm:memory-upper-p1}
	For every $d \geq 1$, there exists a finitary parity game~$\game_d$ such that
	\begin{itemize}
	
		\item $\game_d$ has $2d$ odd colors and $\card{\game_d} \in \bigo(d)$,
	
		\item Player~$1$ has a strategy~$\tau$ in $\game_d$ with $\cost(\tau) = 5(d-1)+7$, but
	
		\item every strategy $\tau$ for Player~$1$ in $\game_d$ with $\cost(\tau) \geq 5(d-1)+7$ has size at least $2^d$. 
	
	\end{itemize}
\end{thm}

\begin{proof}
	Fix some~$d \geq 1$.
	Similarly to the previous proof, we construct an arena~$\arena$ using two kinds of gadgets, one for each player, each of which is repeated~$d$ times.
	In~$\arena$, first Player~$0$ opens~$d$ requests and subsequently picks one of these requests to be answered.
	If Player~$1$ recalls the requests, then he can delay the answer to any chosen request up to~$5(d-1) + 7$ steps.
	Otherwise, Player~$0$ can find a sequence of requests that is answered optimally.
	
	We show the gadgets in Figure~\ref{fig:memory-upper-bound:player-1:gadgets} together with their coloring, and call them~$G^0_j$ and~$G^1_j$ for the $j$-th gadget of Player~$0$ and Player~$1$, respectively.
	The gadget~$G^0_j$ contains the colors~$0$,~$4j-3$, and~$4j-1$, while the gadget~$G^1_j$ contains the colors~$0$,~$4j-2$, and~$4j$.
	The complete arena~$\arena$ is shown in Figure~\ref{fig:memory-upper-bound:player-1:game}.
	We fix the initial vertex of $G^0_1$ to be the initial vertex~$v_\initmark$ of~$\game_d$.
	Clearly,~$\arena$ satisfies the first item of the theorem.
	
	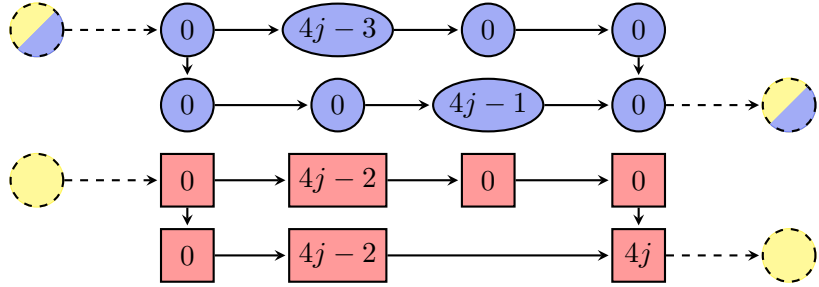
\begin{figure} \centering
		\begin{tikzpicture}[thick,xscale=2]
			\begin{scope}
				\node[p0,draw=mydarkyellow,check] (p0-entry) at (-1,0) {};
				\begin{scope}
					\clip ($(p0-entry) + (-1,-2)$) -- ($(p0-entry) + (-1,2)$) -- ($(p0-entry) + (1,2)$) -- cycle;
					\node[p0,draw=mydarkyellow,check] at (p0-entry) {};
					\node[p0,dashed,check] at (p0-entry) {};
				\end{scope}
				\begin{scope}
					\clip ($(p0-entry) + (-1,-2)$) -- ($(p0-entry) + (1,-2)$) -- ($(p0-entry) + (1,2)$) -- cycle;
					\node[p0,draw=mydarkblue,assign] at (p0-entry) {};
					\node[p0,dashed,assign] at (p0-entry) {};
				\end{scope}
			
				\node[p0,assign] (p0-0-2) at (0,0) {$0$};
				\node[p0,assign,ellipse, inner sep =0] (p0-1-2) at (1,0) {$4j-3$};
				\node[p0,assign] (p0-2-2) at (2,0) {$0$};
				\node[p0,assign] (p0-3-2) at (3,0) {$0$};
				\node[p0,assign] (p0-0-1) at (0,-1) {$0$};
				\node[p0,assign] (p0-0-0) at (1,-1) {$0$};
				\node[p0,assign,ellipse, inner sep =0] (p0-1-0) at (2,-1) {$4j-1$};
				\node[p0,assign] (p0-3-0) at (3,-1) {$0$};
				
				\path
					(p0-entry) edge[dashed] (p0-0-2)
					(p0-0-2) edge (p0-1-2) edge (p0-0-1)
					(p0-1-2) edge (p0-2-2)
					(p0-2-2) edge (p0-3-2)
					(p0-3-2) edge (p0-3-0)
					(p0-0-1) edge (p0-0-0)
					(p0-0-0) edge (p0-1-0)
					(p0-1-0) edge (p0-3-0);
					
				\node[p0,draw=mydarkyellow,check] (p0-exit) at (4,-1) {};
				\begin{scope}
					\clip ($(p0-exit) + (-1,-2)$) -- ($(p0-exit) + (-1,2)$) -- ($(p0-exit) + (1,2)$) -- cycle;
					\node[p0,draw=mydarkyellow,check] at (p0-exit) {};
					\node[p0,dashed,check] at (p0-exit) {};
				\end{scope}
				\begin{scope}
					\clip ($(p0-exit) + (-1,-2)$) -- ($(p0-exit) + (1,-2)$) -- ($(p0-exit) + (1,2)$) -- cycle;
					\node[p0,draw=mydarkblue,assign] at (p0-exit) {};
					\node[p0,dashed,assign] at (p0-exit) {};
				\end{scope}
				
				\path (p0-3-0) edge[dashed] (p0-exit);
			\end{scope}
			
			\begin{scope}[shift={(0,-2)}]		
				\node[p0,draw=mydarkyellow,check] (p1-entry) at (-1,0) {};
				\node[p0,check,dashed] at (p1-entry) {};
				
				\node[p1,choice] (p1-0-1) at (0,0) {$0$};
				\node[p1,choice] (p1-1-1) at (1,0) {$4j-2$};
				\node[p1,choice] (p1-2-1) at (2,0) {$0$};
				\node[p1,choice] (p1-3-1) at (3,0) {$0$};
				\node[p1,choice] (p1-0-0) at (0,-1) {$0$};
				\node[p1,choice] (p1-1-0) at (1,-1) {$4j-2$};
				\node[p1,choice] (p1-3-0) at (3,-1) {$4j$};
				
				\path
					(p1-entry) edge[dashed] (p1-0-1)
					(p1-0-1) edge (p1-1-1) edge (p1-0-0)
					(p1-0-0) edge (p1-1-0)
					(p1-1-1) edge (p1-2-1)
					(p1-2-1) edge (p1-3-1)
					(p1-3-1) edge (p1-3-0)
					(p1-1-0) edge (p1-3-0);
					
				\node[p0,draw=mydarkyellow,check] (p1-exit) at (4,-1) {};
				\node[p0,check,dashed] at (p1-exit) {};
				
				\path (p1-3-0) edge[dashed] (p1-exit);
			
			\end{scope}
		\end{tikzpicture}
		\caption{The gadgets $G^0_j$ (above) and $G^1_j$ (below).}
		\label{fig:memory-upper-bound:player-1:gadgets}
	\end{figure}
	
	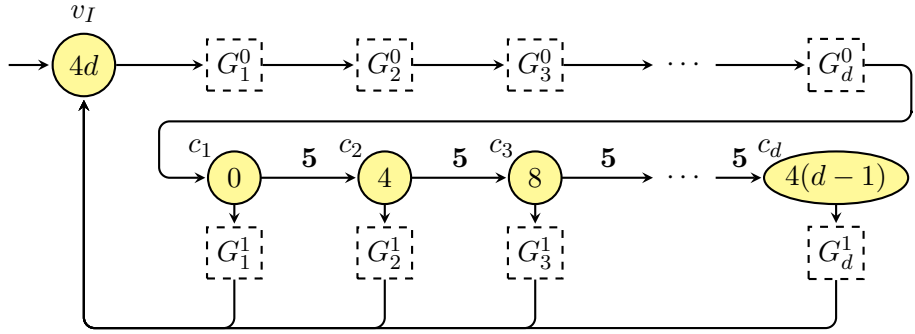
\begin{figure} \centering
		\begin{tikzpicture}[thick,xscale=2]
		
			\node[p0,check,label=north:$v_\initmark$] (reset) at (-1,0) {$4d$};
			\path ($(reset) - (.5,0)$) edge (reset);
			
			\begin{scope}[shift={(0,0)}]
				
				\node[draw,dashed] (G0-1-3) at (0,0) {$G^0_1$};
				\node[draw,dashed] (G0-5-7) at (1,0) {$G^0_2$};
				\node[draw,dashed] (G0-9-11) at (2,0) {$G^0_3$};
				\node (G0-dots) at (3,0) {$\cdots$};
				\node[draw,dashed] (G0-2d-3-2d-1) at (4,0) {$G^0_d$};
				
				\path[draw,rounded corners,->]
					(reset) edge (G0-1-3)
					(G0-1-3) edge (G0-5-7)
					(G0-5-7) edge (G0-9-11)
					(G0-9-11) edge (G0-dots)
					(G0-dots) edge (G0-2d-3-2d-1);

			\end{scope}
			
			\begin{scope}[shift={(0,-1.5)}]
				\node[p0,check,outer sep=0] (choice-0) at (0,0) {$0$};
				\node[anchor=south east,inner sep=.5pt] at (choice-0.north west) {$c_1$};
				
				\node[draw,dashed] (G1-2-4) at (0,-1) {$G^1_1$};
				
				\path
					(choice-0) edge (G1-2-4);
					
				\path[draw,->,rounded corners]
					(G1-2-4) -- (0,-2) -| (reset);
				
			\end{scope}
			
			\begin{scope}[shift={(1,-1.5)}]
				\node[p0,check,outer sep=0] (choice-1) at (0,0) {$4$};
				\node[anchor=south east,inner sep=.5pt] at (choice-1.north west) {$c_2$};
				
				\node[draw,dashed] (G1-4-8) at (0,-1) {$G^1_2$};
				
				\path
					(choice-1) edge (G1-4-8);
					
				\path[draw,->,rounded corners]
					(G1-4-8) -- (0,-2) -| (reset);
			\end{scope}
			
			\path (choice-0) edge node[anchor=south] {$\mathbf{5}$} (choice-1);
			
			\begin{scope}[shift={(2,-1.5)}]
				\node[p0,check,outer sep=0] (choice-2) at (0,0) {$8$};
				\node[anchor=south east,inner sep=.5pt] at (choice-2.north west) {$c_3$};
				
				\node[draw,dashed] (G1-10-12) at (0,-1) {$G^1_3$};
				
				\path
					(choice-2) edge (G1-10-12);
					
				\path[draw,->,rounded corners]
					(G1-10-12) -- (0,-2) -| (reset);
			\end{scope}
			
			\path (choice-1) edge node[anchor=south] {$\mathbf{5}$} (choice-2);
			\node (choice-dots) at (3,-1.5) {$\cdots$};
			
			\path
				(choice-2) edge node[anchor=south] {$\mathbf{5}$} (choice-dots);
			
			\begin{scope}[shift={(4,-1.5)}]
				\node[p0,check,outer sep=0,ellipse,inner sep=0] (choice-d) at (0,0) {$4(d-1)$};
				\node[anchor=south east,inner sep=.5pt] at (choice-d.north west) {$c_d$};
				
				\node[draw,dashed] (G1-4d-2-4d) at (0,-1) {$G^1_d$};
				
				\path
					(choice-d) edge (G1-4d-2-4d);
					
				\path[draw,->,rounded corners]
					(G1-4d-2-4d) -- (0,-2) -| (reset);
			\end{scope}
			
			\path (choice-dots) edge node[anchor=south] {$\mathbf{5}$} (choice-d);
			
			\path[draw,rounded corners,->] (G0-2d-3-2d-1) -| (4.5,-.5) |- (-.5,-.75) |- (choice-0);
		\end{tikzpicture}
		\caption{The arena~$\arena$ of~$\game_d$ showing necessity of exponential memory for Player~$1$.}
		\label{fig:memory-upper-bound:player-1:game}
	\end{figure}

The game is played in rounds. 
Each such round starts and ends in the initial vertex~$v_\initmark$ that answers every request. 
Thus, it suffices to analyze a single round of the game:
In each round, Player~$0$ starts by posing~$d$ requests.
	With her~$j$-th request, she may either request the color~$4j-3$ and take three steps before leaving~$G^0_j$, or she may request the color~$4j-1$ and take a single step before leaving~$G^0_j$.
	After posing~$d$ requests, Player~$0$ then moves to some~$G^1_j$ for~$1 \leq j \leq d$ while answering all requests for colors~$c \leq 4(j-1)$.
	In~$G^1_j$, Player~$1$ answers the request posed in~$G^0_j$.
	After he has done so, all requests are reset and the next round begins.
	All requests posed in~$G^0_{j'}$ for $j' < j$ are answered before entering $G^1_j$ due to the structure of the arena.
	Similarly, all requests posed in~$G^0_{j'}$ for $j' > j$ are answered immediately after leaving~$G^1_j$. 
	We show that by remembering all requests, which takes $2^d$ memory states, Player~$1$ can ensure that one request is only answered with cost~$5(d-1)+7$.
	If he uses less memory states, we show that Player~$0$ is able to answer every request with cost~$5(d-1)+6$, i.e., that strategy has cost at most~$5(d-1) + 6$.

	First, consider the strategy~$\tau$ for Player~$1$ that is defined as follows:
	During Player~$0$'s part of the round, Player~$1$ stores the requests that she makes using $2^d$ memory states.
	Assume Player~$0$ then moves to~$G^1_j$.
	If Player~$0$ requested color~$4j-3$ during her part of the round, Player~$1$ picks the upper branch shown in Figure~\ref{fig:memory-upper-bound:player-1:gadgets}, while he chooses the lower branch in case Player~$0$ requested color~$4j-1$.

	It remains to argue that this strategy indeed enforces a cost at least $5(d-1)+7$.
	Consider a play prefix that starts in~$v_\initmark$ and ends in~$c_j$ for some $1 \leq j \leq d$.
	At this point, the request with the highest cost incurred so far is either a request for color~$4j-3$ with cost~$5(d-1)+4$, or a request for color~$4j-1$ with cost~$5(d-1)+2$.
	All requests for colors~$c < 4j-3$ have already been answered due to the structure of the arena.
	By coloring the vertices resulting from the subdivision of the edges labeled with cost~$\mathbf 5$ with the color of the target of these edges, these requests incur cost at most~$5(d-1) + 5$.
	Moreover, all requests for colors~$c > 4j - 1$ have incurred a cost of at most~$5(d-1) - 1$.
	Now, assume that Player~$0$ enters $G^1_j$.
	First, consider the case with an open request for color~$4j-3$.
	Then Player~$1$ moves through the lower branch of his gadget, answering this request with cost~$5(d-1)+7$.
	If there is, however, an open request for color~$4j-1$, then Player~$1$ moves through the upper branch of his gadget, answering the open request with cost~$5(d-1)+7$ as well.
	Thus, the strategy~$\tau$ has cost~$5(d-1) + 7$.
	All other requests are answered immediately after leaving the gadget at vertex~$v_\initmark$ with a cost of at most $5(d-1) + 5$.
	
	We now show that Player~$1$ indeed needs at least~$2^d$ many memory states to enforce a cost of at least~$5(d-1)+7$.
	Towards a contradiction assume that Player~$1$ has a finite-state strategy~$\tau$ with~$\cost(\tau) \geq 5(d-1)+7$ that is implemented by a memory structure~$(M, m_\initmark, \update)$, where $\card{M} < 2^d$.
	We inductively construct a play~$\rho$ consistent with $\tau$ such that $\cost(\rho) \leq 5(d-1)+6$.
	Assume we have already defined a prefix~$\pi$ of $\rho$ ending in the initial vertex~$v_\initmark$.
	We determine a sequence of~$d$ requests and a choice of~$1 \leq j \leq d$ and prolong $\pi$ by letting Player~$0$ first pick the sequence of requests and then move into some gadget~$G^1_j$.
	Then, Player~$1$ applies his strategy, which leads back to the initial vertex. 
	
	To this end, let~$m = \update^+(m_\initmark, \pi)$.
	Since $\card{M} < 2^d$, and since there exist~$2^d$ play infixes leading from the unique successor of~$v_\initmark$ to~$c_1$, there exist two such infixes~$\pi_1$ and~$\pi_2$, such that $\update^+(m, \pi_1) = \update^+(m, \pi_2)$.
	Let~$j$ be minimal such that the choices made in~$G^0_j$ by Player~$0$ differ in~$\pi_1$ and~$\pi_2$, and w.l.o.g. assume that Player~$0$ poses a request for color~$4j-3$ when playing~$\pi_1$, while she poses a request for color~$4j-1$ when playing~$\pi_2$.
	
	Now consider the response of Player~$1$ consistent with~$\tau$ if Player~$0$ moves to~$G^1_j$ after the play prefix~$\pi \pi_1$ and note that this response is the same as the one to the play prefix~$\pi \pi_2$ due to~$\update^+(m, \pi_1) = \update^+(m, \pi_2)$.
	If Player~$1$ traverses the upper branch of~$G^1_j$ after witnessing~$\pi \pi_1$ or~$\pi \pi_2$, then he answers the request for~$4j-3$ posed during the traversal of~$\pi_1$ with cost~$5(d-1) + 6$.
	If he, however, traverses the lower branch of~$G^1_j$ after witnessing~$\pi \pi_1$ or~$\pi \pi_2$, then he answers the request for~$4j-1$ posed during~$\pi_2$ with cost~$5(d-1) + 6$.
	In the former case, we continue~$\pi$ by letting Player~$0$ play according to~$\pi_1$, while in the latter case we continue~$\pi$ by letting her play according to~$\pi_2$.
	In either case, we move to~$G^1_j$ afterwards.
	In~$G^1_j$, Player~$1$ plays consistently with~$\tau$.

	In both cases all requests posed in gadgets~$G^0_{j'}$ for~$j' < j$ are answered after at most~$5(d-1)+5$ steps, namely upon reaching the first vertex of the subdivision of the edge leading to the vertex~$c_{j' + 1}$.
	Also, the request posed in~$G^0_j$ is answered after~$5(d-1) + 6$ steps.
	Finally, all requests posed in gadgets~$G^0_{j'}$ for~$j' > j$ are answered after at most $5(d+j-j'-1) + 10 \leq 5(d-1) + 5$ steps upon visiting the vertex~$v_\initmark$.
	
	Since all requests are reset when reaching~$v_\initmark$, and since the reasoning above holds true for any memory state~$m$ reached at the end of any~$\pi$ as above, the play~$\rho$ resulting from an inductive application of this construction has $\cost(\rho) \leq 5(d-1) + 6$.
	Since~$\rho$ is consistent with $\tau$, this contradicts~$\cost(\tau) \geq 5(d-1) + 7$ and concludes the proof of the theorem.
\end{proof}

Note that Player~$1$ does not win the games~$\game_d$ with regards to the classical finitary parity condition, i.e., he cannot unbound the cost of an open request arbitrarily.

\section{Tradeoffs Between Memory and Cost}\label{sec:tradeoffs}

In the previous section, we have shown that an optimal strategy for either player in a parity game with costs requires exponential memory in general.
In contrast, winning strategies of minimal size for Player~$0$ in parity games with costs are known to be positional~\cite{FZ14}, while winning strategies for Player~$1$ require infinite memory already in the case of finitary parity games in order to violate every bound infinitely often.
Here we show that, in general, there exists a gradual tradeoff between the size and the cost of a strategy for both players.
For Player~$0$, this means that she can choose to lower the guaranteed bound~$b$ by using a larger winning strategy.
Dually, Player~$1$ can reduce the amount of memory he has to use by not violating every bound, but only a fixed bound~$b$.

\begin{thm}
\label{thm:tradeoffs:p0}
  
Fix some $d \geq 1$ and let the game $\game_d$ be as defined in the proof of Theorem~\ref{thm:memory-upper-p0}. For every $j$ with~$1 \leq j \leq d$ there exists a strategy $\sigma_j$ for Player~$0$ in $\game_d$ such that
\begin{itemize}
	
	\item $d^2 + 3d - 1 = \cost(\sigma_1) > \cost(\sigma_2) > \cdots > \cost(\sigma_d) = d^2 + 2d$, and
	
	\item $1 = \size{\sigma_1} < \size{\sigma_2} < \cdots < \size{\sigma_d} = 2^{d-1}$.
 
\end{itemize}
Also, for every strategy $\sigma'$ for Player~$0$ in $\game_d$ with $\cost(\sigma') \leq \cost(\sigma_j)$ we have $\card{\sigma'} \geq \card{\sigma_j}$.
\end{thm}

\begin{proof}
	Recall that we defined the set of strictly increasing odd sequences $\incseq_d$ in the proof of Theorem~\ref{thm:memory-upper-p0} and showed that a memory structure using $\incseq_d$ as memory states implements an optimal strategy with cost $d^2 + 2d$.
	Intuitively, such a strategy stores up to $d-1$ requests made by Player~$1$ in his part of each round, as the final element of each increasing sequence is fixed to be~$2d-1$.
	The idea behind the construction of the strategies~$\sigma_j$ is to restrict the memory of Player~$0$ such that she can only store up to $j-1$ requests.
	In the extremal cases of $j=1$ and $j=d$ this implements a positional strategy and the strategy from the proof of Theorem~\ref{thm:memory-upper-p0}, respectively.
	
	We implement $\sigma_j$ by again using strictly increasing odd sequences, where we restrict the maximal number of entries that differ from the maximal value of~$2d-1$.
	Hence, in strategy $\sigma_j$, Player~$0$ stores at most $j-1$ requests.
	
	To this end, we define the length-restricted set of strictly increasing odd sequences 
	\[ \incseq_d^j = \incseq_d \cap \set{s = (c_1,\dots,c_{j-1},2d-1,\dots,2d-1) \mid s \in \nats^d } \]
	and pick $M_d^j = \incseq_d^j$.
	Note that $M_d^d = M_d$ as defined in the proof of Theorem~\ref{thm:memory-upper-p0} and that~$M_d^1$ is a singleton set.
	Clearly, the second claim of the theorem holds true, since $\incseq^{j-1}_d \subsetneq \incseq^j_d$ for each $d \geq 1$ and each~$j$ with~$1 \leq j \leq d$.
	The initial memory state is $(1,3,\dots,2j-3,2d-1,\dots,2d-1)$, the update function only stores the first~$j$ relevant requests, and the next-move function~$\nxt_j$ for Player~$0$ is the same as that from the proof of Theorem~\ref{thm:memory-upper-p0} in order to obtain the memory structure $\mem_j$ implementing $\sigma_j$ via~$\nxt_j$.
		
	It remains to show that each strategy $\sigma_j$ realizes a cost of $d^2 + 3d - j$ and that it is minimal for its respective cost.
	To this end, we fix some $j$ with $1 \leq j \leq d$ for the remainder of this proof.
	First, we show that Player~$1$ can enforce a cost of $d^2 + 3d - j$ if Player~$0$ plays consistently with $\sigma_j$.
	Intuitively, Player~$1$ fills the memory of Player~$0$ as quickly as possible, and requests the minimal color that has not yet been requested afterwards.
	Thus, he maximizes the gap between the smallest unstored request and the \myquot{default} answer of~$2d$.

	More precisely, in each turn Player~$1$ requests the colors $1,3,\dots,2j-3,2j-1,2j-1,\dots,2j-1$.
	Playing consistently with~$\sigma_j$, Player~$0$ answers these requests with $2, 4, \dots, 2j-2, 2d, 2d, \dots, 2d$.
	Hence, the cost of the resulting play is that incurred by answering a request for $2j-1$ in the $j$-th gadget of Player~$1$ with~$2d$ in the $j$-th gadget of Player~$0$.
	As argued in the proof of Theorem~\ref{thm:memory-upper-p0}, the cost incurred by such a request-response-pair amounts to 
	\[\left[d - j + 2\right] + \left[(d-1)(d+2)\right] + d = d^2 + 3d - j.\]
	As the game restarts after Player~$0$'s turn, Player~$1$ can enforce this cost infinitely often.
	Hence, $\cost(\sigma_j) \geq d^2 + 3d - j$.

	This sequence of requests is indeed optimal for Player~$1$, i.e., he cannot enforce a higher cost.
	Assume that Player~$1$ does not pose requests as specified above, but poses the requests $c_1, \dots, c_d$.
	Then either there exist some $k$ and $k'$ with $k < k' \leq j$, such that $c_k \geq c_{k'}$, or there exists a $k \leq j$ with $2j-1 < c_k \leq 2d-1$.
	
	In the former case, let~$k'$ be minimal such that such a~$k$ exists.
	Player~$0$ answers the first~$k' - 1$ requests optimally before answering all remaining requests with costs at most $(d-1)(d+2)$, as she ignores the request for $c_{k'}$.
	In the latter case, Player~$0$ again answers all requests up to the first request as described above optimally.
	Afterwards, she answers all succeeding requests with cost at most $d^2 + 2d + (2d - 1 - c_k) / 2 \leq d^2 + 3d - j$.
	Hence, there exists no play $\rho$ consistent with $\sigma_j$ and $\cost(\rho) > d^2 + 3d - j$.
		
	To conclude the proof, we observe that there exists no strategy $\sigma'$ with $\card{\sigma'} < \card{\sigma_j}$ and $\cost(\sigma') \leq \cost(\sigma_j)$.
	The argument is nearly identical to the argument of minimality of the strategy constructed in the proof of Theorem~\ref{thm:memory-upper-p0} and can in fact be obtained by replacing all occurrences of $2^{d-1}$ and $d^2+2d$ by $\card{\sigma_j}$ and $d^2 + 3d - j$, respectively.
	Hence, the strategies~$\sigma_j$ are minimal for their respective cost.
	\end{proof}
	
The similar result for Player~$1$ has a much simpler proof.
Again, it suffices to reuse the games from the lower bound.

\begin{thm}
For each~$d \geq 1$ there exists a finitary parity game~$\game_d$ with~$\card{\game_d} \in \bigo(d^2)$, such that for every~$j$ with $1 \leq j \leq d$, there exists a strategy~$\tau_j$ for Player~$1$ in~$\game_d$ such that
	\begin{itemize}
		\item $7 = \cost(\tau_1) < \cost(\tau_2) < \cdots < \cost(\tau_d) = 5(d-1)+7$, and
		\item $2 = \card{\tau_1} < \card{\tau_2} < \cdots < \card{\tau_d} = 2^d$.
	\end{itemize}
	Moreover, for every strategy~$\tau'$ with~$\cost(\tau') \geq \cost(\tau_j)$, we have~$\card{\tau'} \geq \card{\tau_j}$.
\end{thm}

\begin{proof}
	For each~$j$ with~$1 \leq j \leq d$, let~$\game'_j$ be the game with~$2j$ odd colors defined in Theorem~\ref{thm:memory-upper-p1}.
	We construct~$\game_d$ such that it contains a dedicated initial vertex~$v_\initmark$ of color~$0$, from which Player~$1$ may choose to move to the initial vertex of any of the~$\game'_j$.
	Once the play~$\rho$ has moved into some~$\game'_j$, it never leaves that part of the arena, i.e., the suffix starting at the second position of~$\rho$ is a play of~$\game'_j$.
	
	For each~$j$, the strategy~$\tau_j$ defined in Theorem~\ref{thm:memory-upper-p1}, augmented by a single move from~$v_\initmark$ to the sub-game~$\game'_j$, satisfies the properties above.
	Moreover, as the sub-games~$\game'_j$ are isolated from each other, each strategy~$\tau'$ for Player~$1$ in~$\game_d$ can be trivially transformed into a strategy for him in the subgame~$\game'_j$ that~$\tau'$ chooses at the beginning of~$\game_d$.
	Hence, every strategy~$\tau_j$ is of minimal size for the cost that it realizes.
\end{proof}

\section{Parity Games with Integer Costs}
\label{sec:concrete}

In this section, we extend the model of parity games with costs to non-negative integer costs.
Recall that the cost function~$\cost\colon E \rightarrow \set{\eps, \inc}$ assigns to each edge of~$\game$ either no cost, or a unit cost.
We now consider cost functions that assign to each edge a natural number.
Formally, let~$\arena$ and~$\col$ be an arena and a coloring as defined previously and let~$\cost\colon E \rightarrow \nats$ be a cost function.
We call~$\game=(\arena, \bincp(\col, \cost))$ a parity game with (non-negative) integer-valued costs.
Now the cost of a play infix is the sum of the costs of its edges and the winning condition~$\bincp(\col, \cost)$ is defined analogously to the case of abstract costs.
We denote the highest cost assigned by~$\cost$ by~$W$ if~$\cost$ is clear from the context.
Since we assume~$\cost$ to be encoded in binary encoding,~$W$ may be of exponential size in the number of bits required to describe~$\cost$.
To distinguish these games from those we considered previously, we call a game~$\game=(\arena, \cp(\col, \cost))$ with~$\cost\colon E \rightarrow \set{\eps, \inc}$ a parity game with abstract costs and omit this qualifier if it is clear which game is meant from the context.

Formally, the size of a parity game with integer-valued costs is defined as $\size{\game} = \size{\arena} + \log W$.
Similarly to the abstract case, we call an edge~$e$ an $\eps$- or an increment-edge if~$\cost(e) = 0$ or~$\cost(e) > 0$, respectively.

We first show that determining the winner in~$\game$ with respect to a given bound~$b$ is as easy as determining the winner in a parity game with abstract costs.
Since the former problem subsumes the latter one, this yields~\pspace-completeness of the former problem via Lemma~\ref{lemma_pspacehardness}.
Afterwards, we briefly discuss the memory requirements of both players as well as the potential tradeoffs present in this setting.
In particular, we show that the tradeoffs between the cost and the size of a strategy are, in general, more pronounced in the case of integer-valued costs.

\subsection{The Complexity of Solving Parity Games with Integer-Valued Costs Optimally}

We first examine the complexity of playing optimally in parity games with integer-valued costs and show that this problem is not harder than the special case of abstract costs.
Afterwards, we argue that this implies an exponential upper bound on the size of optimal strategies for both players.

\begin{thm}
\label{thm:binary:pspace-completeness}
The following problem is $\pspace$-complete: \myquot{Given a parity game with integer-valued costs~$\game$ and a bound~$b \in \nats$ (in binary encoding), does Player~$0$ have a strategy~$\sigma$ for $\game$ with $\cost(\sigma) \le b$?}
\end{thm}

\pspace-hardness of this problem follows directly from Lemma~\ref{lemma_pspacehardness}, as every finitary parity game is a parity game with integer-valued costs.
Thus, the remainder of this section is dedicated to showing \pspace-membership of the given problem.
We follow the same approach that solved the case of abstract costs, i.e., we first extend~$\game$ with the same memory structure yielding the qualitative parity game~$\extgame$, which is still only of exponential size.
In contrast to the abstract case, however, plays of~$\extgame$ are not settled after polynomially many moves, as the upper bound on~$b$ is now exponential.
Thus, we construct a novel finite-duration game~$\jumpgame$, which extends~$\finitegame$ by a shortcut mechanism. 
This reduces unsettled play prefixes to at most polynomial length.
By showing equivalence of~$\extgame$ and~$\jumpgame$ and by the fact that~$\jumpgame$ can be simulated on a polynomially time-bounded alternating Turing machine, we obtain \pspace-membership of the given problem.
Due to the similarities between the two proofs, we reuse notation from Section~\ref{sec_pspacemembership}.

We begin by bounding the parameter~$b$ from above.
While we could assume~$b < n$ in the abstract case, this is not the case with integer-valued costs anymore: Here, a positional winning strategy for Player~$0$, which still exists, has cost at most~$nW$.

\begin{cor}
\label{cor:binary:cost-upperbound}
Let $\game$ be a parity game with integer-valued costs with $n$ vertices and largest cost~$W$.
If Player~$0$ wins $\game$, then she has a strategy~$\sigma$ with $\cost(\sigma) \le nW$, i.e., an optimal strategy has cost at most~$nW$.
\end{cor}

Due to this corollary, if~$b \geq n W$, we check for the existence of a winning strategy for Player~$0$ in~$(\arena, \cp(\cost', \col))$, where~$\cost'(e) = \eps$, if~$\cost(e) = 0$, and~$\cost'(e) = \inc$ otherwise.
This can be decided in polynomial space due to Theorem~\ref{thm:previouswork}(\ref{thm:previouswork:cost}) and is equivalent to deciding the given problem.
Hence, we assume~$b < n W$ for the remainder of this proof.

Given~$\game$ and~$b$, we define~$\extgame = (\arena \times \mem, \parity(\col'))$ as before, with~$M = [n+1] \times R$, where~$[n+1]$ implements the overflow counter and~$R = (\set{\bot} \cup [b+1])^D$ is the set of request functions.
Since the proof of Lemma~\ref{lem:cost-parity-to-parity} does not rely on~$b < n$, Player~$0$ wins~$\extgame$ if and only if she has a strategy of cost at most~$b$ for~$\game$.
Recall that~$v'_\initmark$ is the initial vertex of~$\arena'$ and that we call a play settled if it either contains a dominating cycle or if the overflow counter is saturated.

In contrast to the abstract case, plays of~$\extgame$ are not settled after polynomially many steps, as up to~$nW$ many increment-edges may have to be traversed before an overflow occurs, which is exponential in the size of the game.
In order to be able to declare the winner of a play after polynomially many steps, we define~$\jumpgame$.
In this game, we skip infixes that form cycles when projected to~$\arena$, during which the costs incurred by the relevant requests increase, but the set of these requests is stable.
We say that such infixes satisfy the shortcut criterion.
When such an infix~$\pi$ is traversed, Player~$1$ has demonstrated that he can increase the cost of the currently open relevant requests by $\cost(\pi)$ without answering any such request, i.e., he can strictly improve his situation.
In~$\jumpgame$, Player~$1$ does not have to traverse this infix (or others) until an overflow occurs, but the play continues at a position that is as \myquot{close} to an overflow as Player~$1$ can get by repeatedly traversing~$\pi$ without actually causing an overflow.

Formally, we say that a play infix~$\pi = (v_0,o_0,r_0)\cdots(v_j,o_j,r_j)$ in~$\arena'$ satisfies the shortcut criterion if~$v_0 = v_j$, if $o_0 = o_j$, if~$\relreq(r_0) = \relreq(r_{j'}) \neq \emptyset$ for all~$0 \leq {j'} \leq j$, if $\cost(\pi) > 0$, and if~$r_j(c^*) + \cost(\pi) \leq b$ for~$c^* = \argmax_{c}r_j(c)$.
Note that the condition~$r_j(c^*) + \cost(\pi) \leq b$ is equivalent to demanding~$\cost(\pi) \leq \frac{b - r_0(c^*)}{2}$.

For the sake of readability, we refrain from defining the arena underlying~$\jumpgame$ formally.
We rather define the set of play prefixes of~$\jumpgame$ inductively, which are subsequences of plays in~$\arena'$.
In particular, the vertices in~$\jumpgame$ inherit the coloring from~$\extgame$.
First,~$v'_\initmark$ is a play prefix of~$\jumpgame$.
Now, let $\pi = (v_0,o_0,r_0) \cdots (v_j, o_j, r_j)$ be a play prefix of~$\jumpgame$ and let~$(v',o',r')$ be a successor of $(v_j, o_j, r_j)$ in~$\arena'$.
If there exists no~$j'$ such that the infix $\pi' = (v_{j'}, o_{j'}, r_{j'}) \cdots (v_j, o_j, r_j) (v',o',r')$ of~$\pi$ satisfies the shortcut criterion, then~$\pi(v',o',r')$ is a play prefix of~$\jumpgame$.
If, however, such a~$j'$ exists, let it be the maximal one, let~$c^* = \argmax_{c }r'(c)$, let~$s = \cost(\pi')$, and let $t = \max\set{ t' > 0 \mid r'(c^*) + s \cdot t' \leq b}$.
Moreover, define~$r^*$ as~$r^*(c) = r'(c) + s \cdot t$ if~$r'(c) \neq \bot$ and~$r^*(c) = \bot$, otherwise.
Then,~$(v_0,o_0,r_0)\cdots(v_j, o_j, r_j) (v',o',r^*)$ is a play prefix of~$\jumpgame$.
We define the cost of the transition from~$(v_j, o_j, r_j)$ to~$(v',o',r^*)$ as~$\cost(v_j,v') + s \cdot t$ in~$\jumpgame$ and redefine the notions of the cost of a play accordingly in order to obtain uniform notation.

Moreover, we use the following notions:
\begin{itemize}
	\item The transition from $(v_j, o_j, r_j)$ to $(v',o',r^*)$ is a \emph{shortcut}.
	\item The infix $(v_{j'}, o_{j'}, r_{j'}) \cdots (v_j, o_j, r_j)(v',o',r^*)$ is a \emph{shortcut cycle}, where we call the vertex~$(v',o',r^*)$ its \emph{destination}.\footnote{Note that, similarly to the case of dominating cycles, a shortcut cycle is only a cycle when projected to its first component.}
	\item The infix~$(v_{j'},o_{j'},r_{j'})\cdots(v_j, o_j, r_j) (v',o',r')$ is the \emph{detour} corresponding to the shortcut cycle, with \emph{destination}~$(v',o',r')$.
\end{itemize}

Let~$(v_j, o_j, r_j)(v',o',r^*)$ be a shortcut with~$(v',o',r')$ as the target of its associated detour as defined above.
We obtain~$\relreq(r') = \relreq(r^*)$ and~$r' \dominatedby r^*$.
Moreover, if~$c^*$ is the open request that has incurred the highest cost in~$r_{j'}$, this shortcut closes at least half the distance between~$r_{j'}(c^*)$ and~$b$, i.e.,~$r^*(c^*) \geq r_{j'}(c^*) + \frac{b - r_{j'}(c^*)}{2}$.
Hence, no infix~$\pi$ containing a shortcut~$\pi'$ satisfies the shortcut criterion, as the cost of~$\pi'$ is already larger than half the cost that would cause an overflow.
Thus,~$\pi$ violates the condition that it must be able to be traversed at least twice without causing an overflow.

However, a shortcut may be part of a dominating cycle.
If the maximal color along the detour associated with some shortcut is odd, then the shortcut cycle is an odd dominating cycle, i.e., a play may be settled due to a dominating cycle that contains a shortcut.

The introduction of shortcuts in~$\jumpgame$ ensures that plays in~$\jumpgame$ are settled in polynomial time.
Fix~$\ell_\jump = (\log(nW) + 1)(n+1)^6$, which is polynomial in the size of~$\game$.

\begin{lem}
\label{lem:pspace-mem:binary:unsettled-bound}
Let~$\pi$ be a play prefix of~$\jumpgame$.
If~$\card{\pi} > \ell_\jump$, then~$\pi$ is settled.	
\end{lem}

\begin{proof}
Let $\pi = (v_0, o_0, r_0)\cdots(v_j, o_j, r_j)$.
First, due to the same argument as in the proof of Lemma~\ref{lem:pspace-mem:unary:unsettled-bound},~$\pi$ may not contain a vertex repetition, as it is unsettled.
Also note that~$\pi$ only contains at most~$n$ overflow positions, each type 1 infix contains at most~$n$ debt-free positions, each type 2 infix contains at most~$d$ request-adding positions and each type 3 infix contains at most~$d$ relevance-reducing positions, again due to the same arguments as in Lemma~\ref{lem:pspace-mem:unary:unsettled-bound}.

Fix a non-empty type 4 infix~$\pi_4$ and recall that there is at least one request continuously open throughout~$\pi_4$.
Let~$c^*$ be the request that has incurred the greatest cost~$s$ at the beginning of~$\pi_4$ and note that the request for~$c^*$ is continuously open throughout~$\pi_4$, as~$\pi_4$ has no relevance-reducing positions.
We define the first \emph{halving position} as the minimal position~$k$ such that~$r_k(c^*) \geq s + \frac{b - s}{2}$ (if it exists).
Inductively, if~$k$ is a halving position, then the minimal~$k' > k$ with~$r_{k'}(c^*) \geq r_k(c^*) + \frac{b - r_k(c^*)}{2}$ is a halving position as well (if it exists).
Since at each halving position the difference between the currently incurred cost of~$c^*$ and the bound~$b$ has been halved when compared to the previous halving position, there exist at most~$\log(b) \leq \log(nW)$ many such positions.
Hence, splitting~$\pi_4$ at its halving positions yields at most~$\log(nW) + 1$ many infixes without overflow, debt-free, request-adding, relevance-reducing or halving positions.
We say such an infix has type~$4'$.

Fix a non-empty type~$4'$ infix~$\pi_{4'}$.
We show that~$\pi_{4'}$ contains at most~$n$ increment-edges.
Towards a contradiction assume that it contains~$n+1$ increment-edges and let~$c^*$ be the request that has incurred the highest cost~$s$ at the beginning of~$\pi_{4'}$.
Since there exist more than~$n$ increment-edges, there exist two such edges leading to the vertices~$(v_k,o_k,r_k)$ and~$(v_{k'}, o_{k'}, r_{k'})$ with~$k < k'$ and~$v_k = v_{k'}$.
If the cost of the infix~$(v_k,o_k,r_k) \cdots (v_{k'}, o_{k'}, r_{k'})$ is larger than~$\frac{b-s}{2}$, then~$\pi_{4'}$ contains a halving position, which yields the desired contradiction.
If the cost is lower, however, then this infix satisfies the shortcut condition, since~$\pi_{4'}$ does not contain overflow positions and the relevant requests are stable throughout~$\pi_{4'}$.
Hence, the vertex~$(v_{k'}, o_{k'}, r_{k'})$ is the destination of a shortcut cycle, which contradicts the infix from~$(v_k,o_k,r_k)$ to~$(v_{k'}, o_{k'}, r_{k'})$ having a cost of less than~$\frac{b-s}{2}$.
Thus,~$\pi_{4'}$ contains at most~$n$ increment-edges and, by splitting $\pi_{4'}$ at the increment-edges, we obtain a decomposition of $\pi_{4'}$ into at most $b+1$ infixes, each without increment-edges and without request-adding, debt-free, and overflow positions.

Moreover, as none of the infixes contain increment-edges, they also contain no shortcuts.
Hence, each infix is of type 5, i.e., at most of length~$n$ as argued in the proof of Lemma~\ref{lem:pspace-mem:unary:unsettled-bound}.
Aggregating all these bounds yields an upper bound of~$(\log(nW) + 1)(n+1)^6$ on the length of an unsettled play prefix~$\pi$.
\end{proof}

We again define the winner of a play~$\rho$ of~$\jumpgame$ such that Player~$0$ wins~$\rho$ if the minimal settled prefix of~$\rho$ is settled due to an even dominating cycle.
Otherwise, Player~$1$ wins~$\rho$.
Due to Lemma~\ref{lem:pspace-mem:binary:unsettled-bound}, the winner of a play of~$\jumpgame$ is determined after finitely many moves, hence~$\jumpgame$ is determined~\cite{Zermelo13}.
It remains to show that solving~$\jumpgame$ actually solves~$\extgame$.

\begin{lem}
Player~$0$ wins~$\extgame$ if and only if she wins~$\jumpgame$.	
\end{lem}

\begin{proof}
We first show that, if Player~$0$ wins~$\jumpgame$, then she also wins~$\extgame$.
To this end, let~$\sigma_\jump$ be a winning strategy for Player~$0$ in~$\jumpgame$.
We construct a winning strategy~$\sigma'$ for her in~$\extgame$ by mimicking the moves made in~$\extgame$ in~$\jumpgame$ using a simulation function~$h$ mapping play prefixes in~$\extgame$ to play prefixes in~$\jumpgame$.
This simulation function satisfies the same invariant as in the proof of Lemma~\ref{lem:unary:infinite-to-finite}, i.e.:
\begin{center}
\begin{minipage}{0.8\linewidth}
	Let $\pi$ be consistent with $\sigma'$ and end in $(v,o,r)$.
	Then, $h(\pi)$ is consistent with $\sigma_\jump$, is unsettled, and ends in $(v,o',r')$ with $(o',r') \dominates (o,r)$. 
\end{minipage}
\end{center}

We define~$h$ and~$\sigma'$ inductively and simultaneously, starting with~$h(v'_\initmark) = v'_\initmark$, which obviously satisfies the invariant.
Now let~$\pi$ be a play prefix of~$\game'$ consistent with~$\sigma'$, ending in~$(v,o,r)$, and assume~$h(\pi)$ is defined.
Due to the invariant,~$h(\pi)$ ends in some~$(v,o_\jump,r_\jump)$ with~$(o,r) \dominatedby (o_\jump,r_\jump)$.
If~$(v,o,r) \in V_0'$, let~$\sigma_\jump(h(\pi)) = (v^*,o^*_\jump,r^*_\jump)$ and define~$\sigma'(\pi) = (v^*, o^*, r^*)$, where~$(o^*, r^*) = \update((o,r),(v,v^*))$.
Otherwise, if $(v,o,r) \in V_1'$, let~$(v^*,o^*,r^*)$ be an arbitrary successor of~$(v,o,r)$ in~$\arena'$.
In either case, let~$\pi^* = \pi \cdot (v^*,o^*,r^*)$.

It remains to define~$h(\pi^*)$.
To this end, let~$(o^*_\jump, r^*_\jump)$ be the unique memory state such that~$\pi^*_\jump = h(\pi) \cdot (v^*, o^*_\jump, r^*_\jump)$ is a play prefix of~$\jumpgame$.
If~$\pi^*_\jump$ is unsettled, we define $h(\pi^*) = \pi^*_\jump$.
This choice satisfies the invariant:
If the vertex~$(v^*, o^*_\jump, r^*_\jump)$ is the destination of a shortcut, then let~$(v^*,o^*_\rightarrow, r^*_\rightarrow)$ be the destination of its corresponding detour.
We obtain $(o^*_\jump, r^*_\jump) \dominates (o^*_\rightarrow, r^*_\rightarrow) \dominates (o^*, r^*)$ due to Lemma~\ref{lem:dominating-memory:stable-concatenation}.
Otherwise, i.e., if~$(v^*, o^*_\jump, r^*_\jump)$ is not the destination of a shortcut, then Lemma~\ref{lem:dominating-memory:stable-concatenation} yields the invariant directly.
Now consider the case that~$\pi^*_\jump$ is settled.
Then it is settled due to containing an even dominating cycle as a suffix, due to the invariant and due to~$\pi^*_\jump$ being consistent with the winning strategy~$\sigma_\jump$ for Player~$0$.
We define~$h(\pi^*)$ by removing the settling dominating cycle similarly to the proof of Lemma~\ref{lem:unary:infinite-to-finite}.
Using the same argument as in the proof of that lemma, we obtain that~$h(\pi^*)$ satisfies the invariant.

It remains to show that~$\sigma'$ is winning for Player~$0$.
To this end, consider a play~$\rho$ consistent with $\sigma'$ and let~$\pi_j$ be the prefix of length~$j$ of~$\rho$.
As all~$\pi_j$ are consistent with~$\sigma'$, due to the invariant, neither the overflow counter of the~$h(\pi_j)$, nor that of the~$\pi_j$ reaches~$n$.
Hence, the colors of the last vertices of~$\pi_j$ and~$h(\pi_j)$ coincide.
Recall the argument in the proof of Lemma~\ref{lem:unary:infinite-to-finite}: If the largest color~$c$ appearing infinitely often in~$\rho$ is odd, then it can only occur finitely often on even dominating cycles.
Hence, after some prefix, every time a vertex of color~$c$ is visited, this vertex is added to the simulated play prefix and never removed.
This unbounded growth contradicts the~$h(\pi_j)$ being unsettled, as every play prefix of~$\jumpgame$ of length~$\ell_\jump+1$ is settled due to Lemma~\ref{lem:pspace-mem:binary:unsettled-bound}.
Thus,~$\rho$ satisfies the parity condition, i.e.,~$\sigma'$ is indeed winning for Player~$0$.

For the other direction, we show that if Player~$1$ wins~$\jumpgame$, he wins~$\extgame$, which suffices due to determinacy.
To this end, let~$\tau_\jump$ be a winning strategy for Player~$1$ in~$\jumpgame$.
We construct a winning strategy~$\tau'$ for him in~$\extgame$ by simulating play prefixes in~$\extgame$ by such prefixes in~$\jumpgame$, from which we remove shortcut- and dominating cycles.
We again define a simulation function~$h$ that maintains the same invariant as in the proof of Lemma~\ref{lem:unary:infinite-to-finite}, i.e.:
\begin{center}
\begin{minipage}{0.8\linewidth}
	Let $\pi$ be consistent with $\tau'$ and end in $(v,o,r)$ with $o <n$.
	Then, $h(\pi)$ is consistent with $\tau_\jump$, is unsettled, and ends in $(v,o',r')$ with $(o',r') \dominatedby (o,r)$. 
\end{minipage}
\end{center}
We define~$h$ and~$\tau'$ inductively and simultaneously, starting with~$h(v'_\initmark) = v'_\initmark$, which clearly satisfies the invariant.
Now let~$\pi$ be a play prefix of~$\extgame$ consistent with~$\tau'$ and ending in~$(v,o,r)$.
If~$(v,o,r) \in V_0'$, then let~$(v^*, o^*, r^*)$ be an arbitrary successor of~$(v,o,r)$ in~$\arena'$.
Otherwise, if~$(v,o,r) \in V_1'$, let~$\tau_\jump(h(\pi)) = (v^*,o^*_\jump,r^*_\jump)$ and define~$\tau'(\pi) = (v^*, o^*, r^*)$, where~$(o^*, r^*) = \update((o,r),(v,v^*))$.
In either case, let~$\pi^* = \pi \cdot (v^*,o^*,r^*)$.

It remains to define~$h(\pi^*)$ in the case that~$o^* < n$.
To this end, let~$(o^*_\jump, r^*_\jump)$ be the unique memory state such that~$\pi^*_\jump = h(\pi) \cdot (v^*, o^*_\jump, r^*_\jump)$ is a play prefix of~$\jumpgame$.
If~$\pi^*_\jump$ is unsettled and if~$(v^*, o^*_\jump, r^*_\jump)$ is not the destination of a shortcut, we define $h(\pi^*) = \pi^*_\jump$, which satisfies the invariant due to Lemma~\ref{lem:dominating-memory:stable-concatenation}.
If~$\pi^*_\jump$ is settled due to~$o^*_\jump = n$, then, due to the invariant and Lemma~\ref{lem:dominating-memory:stable-concatenation}, we obtain~$o^* = n$, i.e., we can define~$h(\pi^*)$ arbitrarily.
If~$\pi^*_\jump$ is settled due to reaching a dominating cycle, we remove this cycle from~$\pi^*_\jump$ similarly to the construction from the proof of Lemma~\ref{lem:unary:infinite-to-finite}.

Finally, consider the case that~$\pi^*_\jump$ is unsettled and~$(v^*, o^*_\jump, r^*_\jump)$ is the destination of a shortcut, with~$(v^*, o^*_\rightarrow, r^*_\rightarrow)$ as the destination of the corresponding detour.
We differentiate whether the destination~$(v^*, o^*_\jump, r^*_\jump)$ of the shortcut merely allows Player~$1$ to catch up to the play prefix constructed in~$\extgame$, or whether it is more advantageous for him than the position~$(v^*,o^*,r^*)$ actually reached in~$\extgame$.
In the former case, i.e., if~$(o^*,r^*) \dominates (o^*_\jump, r^*_\jump)$, we define~$h(\pi^*) = \pi^*_\jump$, which satisfies the invariant by assumption.
In the latter case, however, i.e., if $(o^*,r^*) \dominates (o^*_\jump, r^*_\jump)$ does not hold true, we remove the shortcut cycle similarly to the removal of a settling dominating cycle, obtaining~$\pi_\jump$, and define~$h(\pi^*_\jump) = \pi_\jump$.
This satisfies the invariant due to~$(o^*,r^*) \dominates (o^*_\rightarrow, r^*_\rightarrow)$, which we obtain via Lemma~\ref{lem:dominating-memory:stable-concatenation}, and the definition of the shortcut condition.

It remains to show that~$\tau'$ is indeed winning for Player~$1$ in~$\game'$.
To this end, consider a play~$\rho$ consistent with $\tau'$ and let~$\pi_j = (v_0,o_0,r_0) \cdots (v_j,o_j,r_j)$ be the prefix of length~$j+1$ of~$\rho$.
If the overflow counter along~$\rho$ eventually saturates,~$\rho$ is clearly winning for Player~$1$.
Hence, assume the opposite, and note that, due to the invariant of~$h$, the colors of the last vertices of~$\pi_j$ and~$h(\pi_j)$ coincide for all~$j \in \nats$.
Let~$c$ be the largest color occurring infinitely often along~$\rho$ and assume towards a contradiction that~$c$ is even.
Similarly to the previous argument,~$c$ must either occur on odd dominating cycles or on removed shortcuts after some finite prefix, as these are the only play infixes that are removed from the simulation. 
Similarly to the proof of Lemma~\ref{lem:unary:infinite-to-finite},~$c$ can only occur finitely often on odd dominating cycles, as each such occurrence implies one occurrence of some larger, odd color.

Now assume that~$c$ occurs infinitely often on removed shortcut cycles.
Since the overflow counter along~$\rho$ never saturates, none of the~$h(\pi_j)$ contains a saturated overflow counter either.
Moreover, as both the removal of an odd dominating cycle and that of a shortcut retain the value of the overflow counter, the values of the overflow counter of the~$h(\pi_j)$ eventually stabilize.
Let~$h(\pi_j) = (v^j_0,o^j_0,r^j_0)\cdots(v^j_{k_j},o^j_{k_j},r^j_{k_j})$.
Pick~$p$ such that~$o_p = o_j$ and $o^p_{k_p} = o^j_{k_j}$ for all~$j > p$, and such that~$c$ is the largest color occurring on the suffix of~$\rho$ starting at position~$p$.

If~$o^p_{k_p} < o_p$, then~$h(\pi_j)$ results from~$h(\pi_{j-1})$ by removing a shortcut cycle only finitely often.
In fact, after reaching~$\pi_p$, no shortcut cycle is removed anymore:
If a shortcut is used in the move from~$h(\pi_{j-1})$ to~$h(\pi_j)$, then~$(o_j, r_j) \dominates (o^j_{k_j}, r^j_{k_j})$, i.e., the shortcut cycle is not removed.
Hence, only finitely many shortcut cycles are removed, which contradicts~$c$ occurring on infinitely many of these.
Thus, we obtain~$o^p_{k_p} = o_p$, which implies~$r^p_{k_p} \dominatedby r_p$ due to the invariant of~$h$.
In particular, for each relevant request that is open in~$r^j_{k_j} $, some larger one is open in~$r_j$ for each~$j > p$.

If~$c$ occurs on a removed shortcut cycle, then~$c$ must be smaller than the smallest relevant request that is open during the witnessing infix: Otherwise it would answer that relevant request, due to~$c$ being even.
Hence the detour corresponding to the infix would violate the shortcut condition.
While there may be some open requests for colors~$c' < c$ in the corresponding infix in~$\rho$, visiting~$c$ does not answer all relevant requests in that corresponding infix in~$\rho$, as argued before.
This implies traversing the shortcut cycle increases the cost of some request in~$\rho$.
Furthermore, since~$c$ is the maximal color visited in the considered suffix, this request eventually overflows after traversing at most~$b+1$ many increment-edges.
This contradicts the choice of~$p$ such that no overflows occur after~$\pi_p$.
If less than~$b+1$ increment-edges occur during the remainder of the play, then also at most~$b+1$ shortcuts occur, since each shortcut requires the traversal of at least one increment-edge.
This in turn contradicts~$c$ occurring on infinitely many removed shortcuts.

Hence, since vertices of color~$c$ only occur finitely often on odd dominating cycles and on removed shortcut cycles, after some finite prefix, each visited vertex of color~$c$ is added to the simulated play and never removed.
Thus, the~$h(\pi_j)$ grow increasingly longer.
Such unbounded growth contradicts them being unsettled, as required by the invariant.
This is due to every play prefix of length at least~$l_\jump$ being settled, due to Lemma~\ref{lem:pspace-mem:binary:unsettled-bound}.
Hence,~$c$ is odd, i.e.,~$\rho$ is winning for Player~$1$.
\end{proof}

Having shown~$\jumpgame$ to be equivalent to~$\extgame$, which is in turn equivalent to~$\game$, we can use the same construction as in the proof of Lemma~\ref{lemma_pspacemembership}, i.e., simulate~$\jumpgame$ on an alternating Turing machine, in order to decide the winner of~$\jumpgame$.
Due to Lemma~\ref{lem:pspace-mem:binary:unsettled-bound}, and due to~$\log(nW)$ being polynomial in the size of the description of~$\game$, this Turing machine is polynomially time-bounded.
Thus, a similar proof to that of Lemma~\ref{lemma_pspacemembership} yields \pspace-membership of the given problem.
Together with the previously stated \pspace-hardness of the given problem, this concludes the proof of Theorem~\ref{thm:binary:pspace-completeness}

Due to the same reasoning as for the results of Section~\ref{sec:memory}, we obtain asymptotically tight exponential bounds for the memory required by both players in order to win with respect to a given bound~$b$.
The upper bounds are obtained as a corollary of the equivalence of~$\game$ and~$\extgame$. 

\begin{cor}
	Let $\game$ be a parity game with integer-valued costs containing~$n$ vertices and~$d$ odd colors.
	\begin{itemize}
		\item If Player~$0$ has a strategy $\sigma$ for $\game$ with $\cost(\sigma) = b$, then she also has a strategy $\sigma'$ with $\cost(\sigma') \leq b$ and $\card{\sigma'} = (b+2)^d$.
		\item If Player~$1$ has a strategy~$\tau$ for~$\game$ with~$\cost(\tau) = b$, then he also has a strategy~$\tau'$ with~$\cost(\tau') \geq b$ and~$\card{\tau'} = n(b+2)^d$.
	\end{itemize}
\end{cor}

We obtain matching lower bounds from Theorem~\ref{thm:memory-upper-p0} and Theorem~\ref{thm:memory-upper-p1}.

\subsection{Tradeoffs Between Memory and Cost in Parity Games with Integer-Valued Costs}

As every finitary parity game is a parity game with integer-valued costs, Theorem~\ref{thm:tradeoffs:p0} holds true for the latter kind of games as well.
In this result, the tradeoff in terms of costs ranges over an interval of size~$d$, while the size of strategies ranges from positional ones to strategies of exponential size.
Hence, even a small improvement in the cost realized by a strategy comes at the price of an exponential increase in memory.
In the case of a binary encoding of integer-valued costs, however, an exponential increase in the size of memory used by a strategy may yield an exponential improvement in terms of the cost realized by the strategy.

\begin{thm}
\label{thm:tradeoffs:binary:p0}
  
  For each~$d \geq 1$, there exists a parity game with integer-valued costs~$\game_d$ with~$\card{\game_d} \in \bigo(d^2)$, such that for every~$j$ with~$1 \leq j \leq d$ there exists a strategy~$\sigma_j$ for Player~$0$ in~$\game_d$ such that

\begin{itemize}[beginpenalty=10000]
	\item $(d+1)2^d + 2^{d-1} = \cost(\sigma_1) > \cost(\sigma_2) > \cdots > \cost(\sigma_d) = (d+1)2^d$, and
	\item $1 = \size{\sigma_1} < \size{\sigma_2} < \cdots < \size{\sigma_d} = 2^{d-1}$.
\end{itemize}
Also, for every strategy $\sigma'$ for Player~$0$ in $\game_d$ with $\cost(\sigma') \leq \cost(\sigma_j)$ we have $\card{\sigma'} \geq \card{\sigma_j}$.
\end{thm}

\begin{proof}
We reuse the arena from the game~$\game_d$ constructed in the proof of Theorem~\ref{thm:memory-upper-p0} and redefine the cost function~$\cost$.
We do so by assigning a cost of~$2^{c-1}$ to each edge leading to some vertex colored with either the odd color~$2c -1$ or the even color~$2c$.
Moreover, we assign a cost of~$2^d - 2^{c-1}$ to each edge leading away from some such vertex.
All other edges are edges of cost~$0$.
The traversal of such a gadget with costs incurs a uniform cost of~$2^d$, regardless of the path taken through it.
As the cost of each edge can be encoded using~$d$ bits,~$\game_d$ is indeed of size~$\bigo(d^2)$.

For any~$j$ with~$1 \leq j \leq d$ let~$\sigma_j$ be a strategy as defined in Theorem~\ref{thm:tradeoffs:p0}.
Due to similar arguments as in the proof of that theorem, we obtain~$\cost(\sigma_j) = (d+1)2^d + 2^{d-1} - 2^{j-1}$, which satisfies the first property stated in the given theorem.
As the strategies remain unchanged from the proof of Theorem~\ref{thm:tradeoffs:p0}, they also satisfy~$\card{\sigma_j} < \card{\sigma_{j+1}}$.
Finally, similar reasoning to the proof of Theorem~\ref{thm:tradeoffs:p0} yields that each~$\sigma_j$ is minimal for its cost.
\end{proof}


\section{Streett Games with Costs}
\label{sec:streett}

In this section, we consider the Streett condition with costs~\cite{FZ14}, which generalizes both the parity condition with costs as well as the classical Streett condition.
We show that, given some Streett game with costs~$\game$ and bound~$b$, the problem of deciding whether there exists a strategy for Player~$0$ in~$\game$ with cost at most~$b$, is $\exptime$-complete.
Thus, this problem is harder to solve than that of solving classical Streett games (unless $\conp = \exptime$), and as hard as solving both finitary Streett games and Streett games with costs~\cite{FZ14}.
As a corollary of this result and of those previously obtained in this work we furthermore obtain tight exponential bounds on the memory required by both players in such a game.
We begin by formally defining the Streett condition with costs, before examining its complexity and the memory required by both players.

Let $\arena = (V, V_0, V_1, E, v_\initmark)$ be an arena and let $\Gamma = (Q_c, P_c)_{c \in [d]}$ for some $d \geq 1$ be a non-empty, finite family of so-called Streett pairs of subsets of $V$.
Intuitively, for each $c \in [d]$, the set $Q_c$ denotes vertices requesting condition~$c$, which are answered by visiting some vertex in $P_c$.
Finally, let $\cost = (\cost_c)_{c \in [d]}$ be a family of cost functions, where $\cost_c\colon E \rightarrow \nats$ for each $c \in [d]$, which we extend to cost functions over plays as usual.
We denote the highest cost assigned by any $\cost_c \in \cost$ by~$W$.

Let $\rho = v_0 v_1 v_2\cdots$ be a play in~$\arena$ and let $j \in \nats$ be a position.
We first define the cost-of-response for a single Streett pair $(Q_c, P_c)$ as
\[
	\streettdist_c(\rho, j) =
	\begin{cases}
 		\min \set{ \cost (v_j \cdots v_{j'}) \mid j' \ge j \text{ and } v_{j'} \in P_c } &\text{if } v_j \in Q_c \\
 		0 &\text{otherwise}
 	\end{cases}
\]
with~$\min\emptyset = \infty$.
Note that, in contrast to parity games, the visit to $v_j$ may open multiple requests, as there may exist multiple $c$ such that $v_j \in Q_c$.
Thus, we define the (total) cost-of-response at position~$j \in \nats$ of $\rho$ by 
\[
	\streettdist(\rho, j) = \max \set { \streettdist_c(\rho, j) \mid c \in [d] } \enspace.
\]
Thus, $\streettdist(\rho, j)$ is the cost of the infix of $\rho$ from position~$j$ to the earliest position where all requests opened at position $j$ are answered, and $\infty$, if at least one such request is not answered.
Moreover, $\streettdist(\rho, j)$ is zero if no requests are opened at position~$j$.

The Streett condition with costs is then defined as
\[\streettc(\Gamma, \cost) = \set{ \rho \in V^\omega \mid \limsup\nolimits_{j\rightarrow \infty} \streettdist(\rho , j) < \infty } \enspace,\]
i.e., $\rho$ satisfies the condition if there exists a bound~$b \in \nats$ such that all but finitely many requests are answered with cost less than $b$.
In particular, only finitely many requests may be unanswered, even if they only incur finite cost.
Similarly to the case of the parity condition with costs, the bound~$b$ may depend on the play $\rho$.

A game~$\game = (\arena, \streettc(\Gamma, \cost))$ is called a Streett game with costs.
If all $\cost_c$ assign~$0$ to every edge, then $\streettc(\Gamma, \cost)$ is a classical Streett condition \cite{streett81}, denoted by $\streett(\Gamma)$.
Dually, if all $\cost_c$ assign~$1$ to every edge, then $\streettc(\col, \cost)$ is equal to the finitary Streett condition over $\Gamma$, as introduced by Chatterjee et al.~\cite{ChatterjeeHenzingerHorn09} and denoted by $\finstreett(\Gamma)$.
In these cases, we refer to $\game$ as a Streett or a finitary Streett game, respectively.

We assume the cost functions to be given in binary encoding.\footnote{All lower bounds shown for this setting already hold for that of finitary Streett games.}
Hence, in general, the largest cost~$W$ is exponential in~$2^{\card{\cost}}$, where $\card{\cost}$ is the length of the encoding of $\cost$.
Thus, we define~$\card{\game} = \card{\arena} + d + \log W$.

\begin{thm}
\label{thm:previouswork-streett}\leavevmode
	\begin{enumerate}
		
		\item\label{thm:previouswork-streett:streett}
		 Solving Streett games is $\conp$-complete.
		 If Player~$0$ wins, then she has a winning strategy of size~$d!$, while Player~$1$ has uniform positional winning strategies~\cite{Horn05}.
		
		\item\label{thm:previouswork-streett:finitary}
		Solving finitary Streett games is $\exptime$-complete.\footnote{Shown in unpublished work by Chatterjee, Henzinger, and Horn, obtained by a minor modification to the proof of \exptime-hardness of solving request-response games \cite{ChatterjeeHenzingerHorn11}.}
		If Player~$0$ wins, then she has a winning strategy of size $d2^d$, but Player~$1$ has in general no finite-state winning strategy~\cite{ChatterjeeHenzingerHorn09}.

		\item\label{thm:previouswork-streett:cost}
		Solving Streett games with costs is $\exptime$-complete.
		If Player~$0$ wins, then she has a winning strategy of size $2^d((2d)!)$, but Player~$1$ has in general no finite-state winning strategy~\cite{FZ14}.
	
	\end{enumerate}
\end{thm}
We define the cost of strategies for Streett games with costs analogously to the parity case.
In contrast to that case, however, we obtain an exponential upper bound on the cost of an optimal strategy for a Street game with costs by using Theorem~\ref{thm:previouswork-streett}(\ref{thm:previouswork-streett:cost}) and applying the same pumping argument as for Corollary~\ref{corollary_costupperbound}.

\begin{cor}
\label{corollary_costupperbound_streett}
Let $\game$ be a Streett game with costs with $n$ vertices and $d$ Streett pairs.
Moreover, let~$W$ be the largest cost in~$\game$.
If Player~$0$ wins $\game$, then she has a strategy~$\sigma$ with $\cost(\sigma) \le nW \cdot 2^d((2d)!)$.
\end{cor}

Similarly to the case of parity games with costs, this bound is tight.
The games demonstrating the lower bound are, however, no longer trivial.
We adapt the games used in~\cite{ChatterjeeHenzingerHorn11} for demonstrating the necessity for exponential memory for Player~$0$ in request-response games in order to show this bound.

\begin{thm}
For each~$d \geq 0$, there exists a finitary Streett game~$\game_d$ with~$\bigo(d)$ vertices and~$d + 1$ Streett pairs, such that Player~$0$ has a strategy with cost~$b = 3 (2^d - 1) + 2$, but no strategy with cost less than~$b$.
\end{thm}

\begin{proof}
	Figure~\ref{fig:streett-lower-bound} shows the game~$\game_3$.
	In general, the game~$\game_d$ consists of an initializing prefix, i.e., the vertices~$P$ and~$Q$, a central vertex~$m$, and one branch for each of the~$d+1$ Streett pairs, which can be entered from~$m$.
	Each branch consists of one path leading back to the central vertex as well as one path restarting the game via moving to the initializing prefix.
	
	The vertices~$P$ and~$Q$ answer and open all~$d+1$ requests, respectively, while a visit to~$m$ neither opens nor answers any requests.
	When moving into the branch associated with condition~$c$, the initial visit to~$P_c$ answers the request for condition~$c$.
	Afterwards, visiting~$c$ or~$\overline{c}$ opens requests for all conditions~$c' < c$ and answers all requests for conditions~$c' > c$, respectively.
	\begin{figure}
	\centering
	\begin{tikzpicture}[thick,yscale=2,xscale=2]
	
		\tikzset{
			prefixnode/.style={fill=myred},
			middlenode/.style={fill=myyellow},
			bladenode/.style={fill=myblue}
		}
		
		\node[p0,prefixnode] (entry) at (0,-1.25) {$P$};
		\node[p0,prefixnode] (request) at (0,-.625) {$Q$};
		\node[p0,middlenode] (center) at (0,0) {$m$};
		
		\path
			($(entry) - (0,.5)$) edge (entry)
			(entry) edge (request)
			(request) edge (center);
		
		\begin{scope}[shift={(1,-.5)}]
			\node[p1,bladenode] (blade-1-start) at (0,0) {$P_0$};
			\node[p1,bladenode] (blade-1-trap) at (1,0) {$\overline{0}$};
			\node[p0,bladenode] (blade-1-end) at (2,0) {$0$};
			
			\path
				(center) edge (blade-1-start)
				(blade-1-start) edge (blade-1-trap) edge [bend left=20] (blade-1-end)
				(blade-1-trap) edge [loop below,looseness=8] (blade-1-trap) edge [bend left=5] (entry)
				(blade-1-end) edge[out=150,in=-5] (center);
		\end{scope}
		
		\begin{scope}[shift={(1,.5)}]
			\node[p1,bladenode] (blade-2-start) at (0,0) {$P_1$};
			\node[p1,bladenode] (blade-2-trap) at (1,0) {$\overline{1}$};
			\node[p0,bladenode] (blade-2-end) at (2,0) {$1$};
			
			\path
				(center) edge (blade-2-start)
				(blade-2-start) edge (blade-2-trap) edge [bend right=20] (blade-2-end)
				(blade-2-trap) edge [loop above,looseness=8] (blade-2-trap)
				(blade-2-end) edge[out=-150,in=5] (center);
				
			\path[draw,thick,->,rounded corners]
				(blade-2-trap) -- ($(blade-2-trap) + (.5,.25)$) -| ($(center) + (3.5,0)$) |- (entry);
		\end{scope}
		
		\begin{scope}[shift={(-1,.5)}]
			\node[p1,bladenode] (blade-3-start) at (0,0) {$P_2$};
			\node[p1,bladenode] (blade-3-trap) at (-1,0) {$\overline{2}$};
			\node[p0,bladenode] (blade-3-end) at (-2,0) {$2$};
			
			\path
				(center) edge (blade-3-start)
				(blade-3-start) edge (blade-3-trap) edge [bend left=20] (blade-3-end)
				(blade-3-trap) edge [loop above,looseness=8] (blade-3-trap)
				(blade-3-end) edge[out=-30,in=-185] (center);
				
			\path[draw,thick,->,rounded corners]
				(blade-3-trap) -- ($(blade-3-trap) + (-.5,.25)$) -| ($(center) - (3.5,0)$) |- (entry);
		\end{scope}
		
		\begin{scope}[shift={(-1,-.5)}]
			\node[p1,bladenode] (blade-4-start) at (0,0) {$P_3$};
			\node[p1,bladenode] (blade-4-trap) at (-1,0) {$\overline{3}$};
			\node[p0,bladenode] (blade-4-end) at (-2,0) {$3$};
			
			\path
				(center) edge (blade-4-start)
				(blade-4-start) edge (blade-4-trap) edge [bend right=20] (blade-4-end)
				(blade-4-trap) edge [loop below,looseness=8] (blade-4-trap) edge [bend right=5] (entry)
				(blade-4-end) edge[out=30,in=185] (center);
		\end{scope}
	\end{tikzpicture}
	
	\caption{The game~$\game_3$.}
	\label{fig:streett-lower-bound}	
	\end{figure}
	
	Intuitively, by moving to~$P_c$, Player~$0$ claims that~$c$ is the smallest index for which a request is open and answers the request for that condition.
	If this claim holds true, the best choice for Player~$1$ is to move to~$c$, where requests for all conditions~$c' < c$ are opened and the play returns to~$m$.
	If, on the other hand, this claim does not hold true, then Player~$1$ can move to~$\overline{c}$, where all requests for conditions~$c' > c$ are answered and Player~$1$ can increase the cost of the remaining requests arbitrarily before moving to~$P$ and thereby starting the next round.
	Note that staying in~$\overline{c}$ infinitely long is losing for Player~$1$, as no requests are opened in that vertex and thus, the cost of the resulting play is~$0$.
	
\makeatletter
\newcommand*{\shifttext}[2]{%
  \settowidth{\@tempdima}{#2}%
  \makebox[\@tempdima]{\hspace*{#1}#2}%
}
\makeatother
	
	The optimal strategy for Player~$0$ implements a binary counter: Whenever the play reaches~$m$, she has to recall the smallest~$c$ for which there is an open request and move to~$P_c$.
	Since moving to~$P_c$ and returning to~$m$ via~$c$ reopens requests for all conditions~$c' < c$, Player~$0$ has to repeat answering these requests from smallest to largest before she can answer the outstanding request for condition~$c+1$.
	Figure~\ref{fig:street-lower-bound:example} shows the cost of requests during a play consistent with this strategy in~$\game_3$.
	\begin{figure}
		\begin{tabular}{l*{17}{c}} \toprule
			Branch & & \shifttext{-1.5em}{0} & \shifttext{-1.5em}{1} & \shifttext{-1.5em}{0} & \shifttext{-1.5em}{2} & \shifttext{-1.5em}{0} & \shifttext{-1.5em}{1} & \shifttext{-1.5em}{0} & \shifttext{-1.5em}{3} & \shifttext{-1.5em}{0} & \shifttext{-1.5em}{1} & \shifttext{-1.5em}{0} & \shifttext{-1.5em}{2} & \shifttext{-1.5em}{0} & \shifttext{-1.5em}{1} & \shifttext{-1.5em}{0} & $\cdots$ \\ \midrule
			$r(0)$ & $1$ & $\bot$ & $1$ & $\bot$ & $1$ & $\bot$ & $1$ & $\bot$ & $1$ & $\bot$ & $1$ & $\bot$ & $1$ & $\bot$ & $1$ & $\bot$ & $\cdots$ \\
			$r(1)$ & $1$ & $4$ & $\bot$ & $\bot$ & $1$ & $4$ & $\bot$ & $\bot$ & $1$ & $4$ & $\bot$ & $\bot$ & $1$ & $4$ & $\bot$ & $\bot$ & $\cdots$ \\
			$r(2)$ & $1$ & $4$ & $7$ & $10$ & $\bot$ & $\bot$ & $\bot$ & $\bot$ & $1$ & $4$ & $7$ & $10$ & $\bot$ & $\bot$ & $\bot$ & $\bot$ & $\cdots$ \\
			$r(3)$ & $1$ & $4$ & $7$ & $10$ & $13$ & $16$ & $19$ & $22$ & $\bot$ & $\bot$ & $\bot$ & $\bot$ & $\bot$ & $\bot$ & $\bot$ & $\bot$ & $\cdots$ \\ \bottomrule
		\end{tabular}
		\caption{An optimal play for Player~$0$ in~$\game_3$. We only consider the positions at which the play is at vertex~$m$. The topmost row denotes the branch visited in-between two visits to that vertex.}
		\label{fig:street-lower-bound:example}
	\end{figure}
	
	Note that, in this example, after answering the request for~$3$, Player~$0$ continues implementing a binary counter until no more requests are open.
	At this point, Player~$0$ may choose an arbitrary branch.
	If Player~$1$ then moves to the trap-vertex and subsequently to~$P$, he starts the next round of the game.
	If he, however, returns to~$m$, Player~$0$ obtains new open requests which she has to answer as she did before.
	Since no request for~$3$ can be opened without moving to~$P$, not moving on to the next round will not yield higher costs than doing so.
	
	As illustrated in Figure~\ref{fig:street-lower-bound:example}, answering the request for condition~$d$ posed at the beginning of each round requires~$2^d-1$ many visits to branches, as well as an additional step into the first vertex of the branch of condition~$d$.
	As each of the visits to the branches implies the traversal of three edges, this request is answered~$b$ steps after it is posed.
	Hence, the strategy implementing a binary counter has a cost of~$b = 3 (2^d - 1) + 2$.
	
	Moreover, as argued above, in each round in which Player~$0$ deviates from this strategy, Player~$1$ can move the play into the \myquot{trap}-vertex~$\overline{c}$ of the current branch upon the first deviation, where he can loop until the cost of an open request increases beyond~$b$, before he moves to the next round.
	Hence, Player~$0$ has to adhere to this strategy after finitely many rounds, i.e., each strategy for her has cost at least~$b$.
\end{proof}

\subsection{The Complexity of Solving Streett Games with Costs Optimally}

We now show that solving Streett games with costs with respect to a given bound~$b$ is \exptime-complete.
Since solving finitary Streett games is complete for the same complexity class, and since an exponential~$b$ suffices for Player~$0$ to win in such games due to Corollary~\ref{corollary_costupperbound_streett}, we can encode the problem of solving a finitary Streett game as the given problem with only a linear blowup, due to binary encoding of~$b$.
Hence, the latter problem is \exptime-hard.

In order to show membership of the given problem in \exptime, we reduce it to that of solving a classical Streett game with exponentially many vertices, but with only a single additional Streett pair.
Since Streett games with~$n$ vertices and~$d$ Streett pairs can be solved in time~$\bigo(nd(d!))$ using the algorithm from Piterman and Pnueli~\cite{PitermanPnueli06}, this construction yields \exptime-membership of the given problem.

\begin{thm}
\label{thm:streett:complexity:completeness}
The following problem is $\exptime$-complete: \myquot{Given a Streett game with costs~$\game$ and a bound~$b \in \nats$ in binary encoding, does Player~$0$ have a strategy~$\sigma$ for $\game$ with $\cost(\sigma) \le b$?}
\end{thm}

\begin{proof}

$\exptime$-hardness of this problem follows from the $\exptime$-hardness of deciding the winner in a finitary Streett game~\cite{FZ14} as argued above.

We now show the given problem to be in $\exptime$.
The idea for this proof is the same as for the proof of Lemma~\ref{lemma_pspacemembership}, i.e., we reduce the problem to solving a classical Streett game $\game'$.
Instead of simulating $\game'$ on the fly, however, we construct and solve it explicitly.

Let $\game = (\arena, \streettc(\Gamma, \cost))$ be a Streett game with costs with $n$ vertices and $d$ Streett pairs.
Moreover, let~$W$ be the largest cost assigned by any~$\cost_c$ in~$\cost$.
If $b \geq nW \cdot 2^d((2d)!)$, then we construct~$\game' = (\arena, \streettc(\Gamma, \cost'))$ with~$\cost'(e) = \eps$, if $\cost(e) = 0$ and $\cost'(e) = \inc$ otherwise and solve $\game'$ using an exponential-time algorithm~\cite{FZ14}.
If Player~$0$ can ensure some upper bound on the cost incurred in~$\game'$, then she can also do so in~$\game$ using the same strategy.
Thus, by Corollary~\ref{corollary_costupperbound_streett}, she can bound the cost from above by~$b$.
Similarly, if Player~$1$ wins~$\game'$, then he can still do so in~$\game$ using the same strategy.
Hence, solving~$\game'$ solves~$\game$ with respect to~$b$.

Thus, assume $b < nW \cdot 2^d ((2d)!)$.
For the reduction of $\game$ to a Streett game we again use a memory structure that keeps track of the costs of responses accumulated so far, up to the bound $b$, while allowing this counter to overflow at most~$n$ times.
Hence, let $\mem = (M, \init, \update)$ be the memory structure with memory states $M = [n+1] \times (\set{\bot} \cup [b+1])^{[d]}$, where we define~$\init$ and~$\update$ analogously to the parity case.
Note that $\mem$ is of size $\card{\mem} = (n + 1) \cdot (b + 2)^d$.

Also, given $\Gamma = (Q_c, P_c)_{c \in [d]}$, we construct $\Gamma' = (Q'_c, P'_c)_{c \in [d + 1]}$ as $(Q'_c, P'_c) = (Q_c \times M, P_c \times M)$ for $c \in [d]$, and $(Q'_d, P'_d) = (V \times (\set{\bot} \cup [b+1])^{[d]} \times \set{n},\emptyset)$.
Thus, $\Gamma'$ is the conjunction of the extension of $\Gamma$ to the game $\game \times \mem$ and one additional Streett pair which causes Player~$0$ to lose once the overflow counter reaches the value~$n$.
Hence~$\card{\Gamma'} = \card{\Gamma} + 1$.

We define $\game' = (\arena \times \mem, \streett(\Gamma'))$.
Player~$0$ wins $\game'$ if and only if she has a strategy with cost at most~$b$ in~$\game$, due to the same arguments as in the proof of Lemma~\ref{lem:cost-parity-to-parity}.

Moreover, using the algorithm presented in \cite{PitermanPnueli06}, we can solve $\game'$ in exponential time in $\card{\game}$.
The Streett game $\game'$ has $n' \in \bigo(n^2 b^d)$ vertices.
Since we assume $b \leq nW \cdot 2^d((2d)!)$, we obtain
	$n' \in \bigo(( nW \cdot 2^d((2d)!) ) ^d)$.
Moreover, $\game'$ has $d' = d+1$ many Streett pairs.

As discussed above, by using the algorithm for solving Streett games by Piterman and Pnueli~\cite{PitermanPnueli06}, we obtain an algorithm that is polynomial in the number of vertices~$n$ of~$\game$, while it is exponential in the number of Streett pairs~$d$ and~$\log W$.
Hence, the given problem is in \exptime.
\end{proof}

This result also holds true if the bound~$b$ is given in unary encoding.
As every number in unary encoding can be rewritten in binary encoding in polynomial time, membership in \exptime\ follows directly.
Moreover, recall that \exptime-hardness of the problem of Theorem~\ref{thm:streett:complexity:completeness} follows from \exptime-hardness of solving finitary Streett games.
This problem is in turn shown to be \exptime-hard via a reduction from the word problem for polynomially time-bounded alternating Turing machines.
A minor modification of that proof yields that a polynomial bound suffices for Player~$0$ in order to win the resulting finitary parity game.
Hence, the problem is still \exptime-hard when considering a unary encoding of the bound~$b$.

\subsection{Memory Requirements of Optimal Strategies in Streett Games with Costs}

As we have shown in the previous section, Streett games with costs can be solved by reducing them to classical Streett games of exponential size, but with only linearly many Streett pairs.
Similarly to Corollary~\ref{cor:parity:memory:upper-bound}, we obtain an exponential upper bound on the memory necessary for both players to win a Streett game with costs with respect to a given bound as a corollary of Theorems~\ref{thm:previouswork-streett}(\ref{thm:previouswork-streett:streett}) and~\ref{thm:streett:complexity:completeness}.

Also, similarly to Corollary~\ref{cor:parity:memory:upper-bound}, we are able to remove the overflow counter from our memory structure for Player~$0$ by letting her play according to the largest value of this counter for which she still has a winning strategy.
This yields an improved upper bound on her required memory.
Moreover, we obtain matching lower bounds for both players as a corollary of Theorems~\ref{thm:memory-upper-p0} and~\ref{thm:memory-upper-p1}, as every finitary parity game with costs with~$d$ colors can be turned into a finitary Streett game of the same size with~$d$ Streett pairs.

\begin{cor}
	\label{cor:streett:memory:upper-bound}
	Let~$\game$ be a Streett game with costs with~$n$ vertices and~$d$ Streett pairs.
	Moreover, let~$b \in \nats$ be some bound.
	\begin{enumerate}[beginpenalty=10000]
		\item If Player~$0$ has a strategy~$\sigma$ in~$\game$ with $\cost(\sigma) \leq b$, then she also has a strategy~$\sigma'$ with~$\cost(\sigma') \leq b$ and~$\card{\sigma'} = (d+1)! \cdot (b + 2)^d$.
		\item If Player~$1$ has a strategy~$\tau$ in~$\game$ with $\cost(\tau) \geq b$, then he also has a strategy~$\tau'$ with~$\cost(\tau') \geq b$ and~$\card{\tau'} = n(b + 2)^d$.
	\end{enumerate}
	These bounds are asymptotically tight already for finitary Streett conditions.
\end{cor}


\section{Conclusion}
\label{sec:conc}

In this work we have shown that playing parity games with costs optimally is harder than just winning them, both in terms of computational complexity as well as in terms of memory requirements of strategies. 
We proved checking an upper bound on the value of an optimal strategy to be complete for polynomial space, while just solving such games is in~$\up \cap \coup$, respectively in~$\ptime$ for the special case of finitary parity games.
Moreover, we have shown that optimal strategies in general require exponential memory, but also that exponential memory is always sufficient to implement optimal strategies.
In contrast, winning strategies in these games are positional.
Finally, we have shown that, in general, there exists a gradual tradeoff between the size and the cost of strategies.

Also, we considered Streett games with costs: checking an upper bound on the cost of an optimal strategy is $\exptime$-complete and exponential memory is sufficient to implement optimal strategies.
Thus, playing optimally is as hard as just winning.

All our proofs can be adapted for the case of bounded parity (Streett) games and bounded parity (Streett) games with costs~\cite{ChatterjeeHenzingerHorn09,FZ14}.
While the parity condition with costs only restricts the cost-of-response in the limit, the bounded parity condition prohibits any unanswered request with cost~$\infty$ (but still allows finitely many unanswered requests with finite cost).
The other conditions are defined similarly. 

In further research, we are considering two additional directions in which to extend the cost function: By allowing negative integers as costs and by allowing multiple cost functions, i.e., by extending the parity winning condition with a family of cost functions similar to the case of Streett conditions.
Moreover, we investigate tradeoffs in delay games with quantitative winning conditions.
Preliminary results exhibited a tradeoff between costs of strategies and delay~\cite{Zimmermann17}, but there are no results involving the size of strategies.

\bibliographystyle{alpha}
\bibliography{main.bib}

\end{document}